\documentclass[twoside]{article}

\usepackage[accepted]{aistats2026}
\usepackage[round]{natbib}

\usepackage[hidelinks]{hyperref}       % hyperlinks
\usepackage{url}            % simple URL typesetting
\usepackage{booktabs}       % professional-quality tables
\usepackage{amsfonts}       % blackboard math symbols
\usepackage{nicefrac}       % compact symbols for 1/2, etc.
\usepackage{microtype}      % microtypography
\usepackage{xcolor}         % colors
\usepackage{amsmath}
\usepackage{amssymb}
\usepackage{amsthm}
\usepackage{mathtools}
\usepackage{subfigure}
\usepackage{enumerate}
\usepackage{enumitem}
\usepackage{thm-restate}
\usepackage{algorithm}
\usepackage{pdflscape}
\usepackage{algpseudocode}
\usepackage{tikz}
\usetikzlibrary{shapes,arrows}
\usepackage{wrapfig}

\theoremstyle{plain}
\newtheorem{theorem}{Theorem}[section]

\newtheorem{lemma}[theorem]{Lemma}
\newtheorem{corollary}[theorem]{Corollary}
\theoremstyle{definition}

\theoremstyle{remark}
\newtheorem{remark}[theorem]{Remark}

\newenvironment{claim}[1]{\par\noindent\underline{Claim:}\space#1}{}
\newenvironment{claimproof}[1]{\par\noindent\underline{Proof:}\space#1}{\hfill $\square$} % \blacksquare

\renewcommand{\P}{\mathbb{P}} %probability
\newcommand{\Prob}{\P}
\newcommand{\E}{\mathbb{E}} %expectation
\newcommand{\Var}{\mathrm{Var}} %variance
\newcommand{\Cov}{\mathrm{Cov}} %correlation
\newcommand\indep{\protect\mathpalette{\protect\independenT}{\perp}} \def\independenT#1#2{\mathrel{\rlap{$#1#2$}\mkern2mu{#1#2}}} %independence
\newcommand{\1}{\mathbf{1}} %indicator function
\newcommand{\R}{\mathbb{R}} %real numbers
\newcommand{\N}{\mathbb{N}} %natural numbers
\renewcommand{\O}{\mathcal{O}} %big O
\renewcommand{\o}{o}
\DeclarePairedDelimiter\norm{\lVert}{\rVert} %makes sure subscripts of the norms are placed properly
\newcommand\myeq{\mkern2.5mu{=}\mkern2.5mu} %equal sign with small spacing
\renewcommand\mid{\mkern4mu{|}\mkern4mu} %equal sign with small spacing

\newcommand{\F}{\ensuremath{{\mathcal F}}}
\renewcommand{\H}{\ensuremath{{\mathcal H}}}
\newcommand{\A}{\ensuremath{{\mathcal A}}}
\newcommand{\I}{\ensuremath{{\mathcal I}}}

\newcommand{\HS}{\ensuremath{{\mathcal S}}}

\newcommand{\forestass}[1]{\textbf{(F#1)}}
\newcommand{\dataass}[1]{\textbf{(D#1)}}
\newcommand{\kernelass}[1]{\textbf{(K#1)}}

\newcommand{\Ybf}{\mathbf{Y}}
\newcommand{\Yibf}{\mathbf{Y}_i}
\newcommand{\ybf}{\mathbf{y}}

\newcommand{\Xbf}{\mathbf{X}}
\newcommand{\xbf}{\mathbf{x}}

\newcommand{\Zbf}{\mathbf{Z}}

\newcommand{\Zcal}{\mathcal{Z}}

\newcommand{\PYgXz}{\P_{\Ybf \mid \Xbf \myeq \xbf}^0}
\newcommand{\PYgXo}{\P_{\Ybf \mid \Xbf \myeq \xbf}^1}

\newcommand{\hwix}{\hat w_i(\xbf)}
\newcommand{\kYi}{k(\Yibf, \cdot)}

\newcommand{\Kbf}{\mathbf{K}}
\newcommand{\kbf}{\mathbf{k}}
\newcommand{\tauk}{\tau_{k}(\xbf)}
\newcommand{\htauk}{\hat\tau_{n,k}(\xbf)}

\newcommand{\htaukHS}{\hat\tau_{n,k}^{\HS}(\xbf)}
\newcommand{\htaukHSb}{\hat\tau_{n,k}^{\HS_b}(\xbf)}

\newcommand{\Lcal}{\mathcal{L}}
\newcommand{\wbf}{\mathbf{w}}

\newcommand{\hwbf}{\hat\wbf}

\newcommand{\bandybf}{\mathcal{B}(\ybf)}

\begin{document}

% If your paper is accepted and the title of your paper is very long,
% the style will print as headings an error message. Use the following
% command to supply a shorter title of your paper so that it can be
% used as headings.
%
\runningtitle{Causal-DRF}

% If your paper is accepted and the number of authors is large, the
% style will print as headings an error message. Use the following
% command to supply a shorter version of the author names so that
% they can be used as headings (for example, use only the surnames)
%
%\runningauthor{Surname 1, Surname 2, Surname 3, ...., Surname n}

\twocolumn[

\aistatstitle{Causal-DRF: Conditional Kernel Treatment Effect Estimation using Distributional Random Forest}

\aistatsauthor{ Jeffrey N\"af \And Junhyung Park \And Herbert Susmann }

\aistatsaddress{ University of Geneva \And  ETH Z\"urich \And NYU Grossman School of Medicine } ]

\begin{abstract}
  The conditional average treatment effect (CATE) is a commonly targeted statistical parameter for measuring the effect of a treatment conditional on covariates. However, the CATE will fail to capture effects of treatments beyond differences in conditional expectations. Inspired by causal forests for CATE estimation, we develop a forest-based method to estimate the conditional kernel treatment effect (CKTE), based on the recently introduced Distributional Random Forest (DRF) algorithm. Adapting the splitting criterion of DRF, we show how one forest fit can be used to obtain a consistent and asymptotically normal estimator of the CKTE, as well as an approximation of its sampling distribution. This allows to study the difference in distribution between control and treatment group and thus yields a more comprehensive understanding of the treatment effect. 
\end{abstract}

\section{Introduction}
The definition, identification and estimation of treatment effects have gained increased attention in statistics and machine learning. Consider random variables \((Y^1, Y^0, W, \mathbf{X})\), where $Y^1, Y^0$ correspond to potential outcomes of treatment and control, respectively \citep{RubinPotentialoutcome}, $W$ is the treatment assignment, and $\mathbf{X}$ additional covariates. The so-called ``fundamental problem of causal inference" \citep{Holland01121986} is that only one of the potential outcomes is observed; that is, one only observes $Y = WY^1 + (1 - W) Y^0$. For example, treatment $W$ could be thought of as indicating the prescription of a drug treatment ($W=1$), $Y$  a measure of health such as blood pressure, and $\mathbf{X}$ consisting of individual baseline covariates. A popular choice to investigate the effectiveness of the treatment given $\mathbf{X}=\mathbf{x}$ is the conditional average treatment effect (CATE), defined as $\mathbb{E}[Y^1 - Y^0\mid \mathbf{X}=\mathbf{x}]$. Under the assumption of no unmeasured confounding, the CATE is identifiable in terms of the observed data $(Y, W, \mathbf{X})$. The CATE is an attractive target, as its conditional structure contains information about treatment heterogeneity across population subgroups, as opposed to the average treatment effect (ATE), which averages treatment effects over the entire population. However, the CATE is limited in that a treatment might affect the outcome beyond just its conditional expectation. For example, giving a drug to several new patients with the same covariates $\mathbf{x}$ could decrease blood pressure on average, but also lead to an undesirable increase in variance. Therefore, it is natural to consider more general measures of differences between treatment and control distributions to obtain a comprehensive picture of the treatment effect. Previous work on causal estimands that go beyond identifying differences in expectations focused, for example, on estimating cumulative distribution functions of the potential outcomes \citep{Chernozhukov2013,
Chernozhukov2020}, or on estimating the difference in treatment and control groups of specified distributional parameters, such as the quantile treatment effect \citep{AI2, CATEGeneralization}. Recently, \citep{GRFDistributionalCausalEffects} provided a unified doubly-robust approach to estimate a variety of distributional parameters to compare treatment and control groups. 

In this paper, we study the conditional kernel treatment effect (CKTE). The CKTE can be seen as an analog to the CATE in a reproducing kernel Hilbert space (RKHS). For a reproducing kernel $k$ with associated Hilbert space $\H$ and a potential multivariate outcome $\Ybf^{1}, \Ybf^{0}$, the CKTE is defined as
\begin{align}\label{taukdef}
   \tauk = \E[k(\Ybf^{1}, \cdot) \mid \Xbf=\xbf] - \E[k(\Ybf^{0}, \cdot) \mid \Xbf=\xbf],
\end{align}
see e.g., \citep{CATEGeneralization}. If the outcome is univariate and we choose the trivial kernel $k(y_1, y_2)=y_1 y_2$, then $\tauk$ reduces to the traditional CATE; thus, the CKTE can be seen as a generalization of the CATE. With other kernel choices, the CKTE provides a more comprehensive distributional comparison of the control and treatment groups. For instance, via the conditional witness function $\ybf \mapsto \tauk(\ybf)$ and with \textit{characteristic} kernels, one can read off the sign of the differences in the local conditional densities around \(\ybf\). That is, $\tauk(\ybf) > 0$ if and only if the difference between the conditional density of the treatment group and that of the control group in the neighborhood of $\ybf$ is greater than zero, and likewise if $\tauk(\ybf) \leq 0$ \citep{CATEGeneralization}. The magnitude indicates how pronounced this local discrepancy is, and the kernel bandwidth controls the size of the neighborhood. Moreover, for characteristic kernels, the hypothesis $H_0 \colon \norm{\tauk}_{\H}=0$ is equivalent to testing the hypothesis $H_0\colon\PYgXz=\PYgXo$. Thus, for characteristic kernels, our pointwise test at a fixed covariate value \(\xbf\) is a no distributional effect test: it assesses whether the conditional treatment and control distributions coincide. By itself, this test does not indicate the direction of the effect, nor which regions of the outcome space are more likely under treatment versus control. That information is instead provided by the conditional witness function \(\ybf\mapsto\tau_k(\ybf)\), whose sign indicates whether outcomes near \(\ybf\) are locally more likely under treatment or control.

Estimating the CKTE involves estimating the conditional mean embedding (CME) $\E[k(\Ybf^{j}, \cdot) \mid \Xbf=\xbf]$, for the groups $j=0,1$. This was done in \citep{CATEGeneralization} using classical kernel regression for each group. More recently, \citep{näf2023confidence} used two distributional random forests (DRFs) for the task. Building on random forests~\citep{breiman2001random}, DRFs provide forest-based estimates of the kernel mean embedding, conditional on $\Xbf$. Using DRF instead of the usual kernel estimators 
% used in \citep{CATEGeneralization} 
is attractive, because DRF inherits the favorable empirical performance of forest-based methods. In particular, DRF tends to deliver accurate estimates, even when $\Xbf$ is high dimensional, and does so without the need for parameter tuning. For instance, the CKTE estimator in \citep{CATEGeneralization} requires to choose two kernels and two regularization parameters, while \citep{näf2023confidence} simply use the default choices of tuning parameters and kernel advocated in \citep{DRF-paper}.

We build on the DRF-based approach of \citet{näf2023confidence} for estimating the CKTE. Specifically, inspired by causal forests \citep{wager2018estimation}, we modify the DRF splitting criterion so that the CKTE is estimated directly from a single forest fit shared across both groups. 
% This contrasts with a two-model approach, analogous to a T-learner, that estimates the two embeddings separately. We do not view the direct, single-forest approach as uniformly superior: rather, the two approaches impose different inductive biases, and their relative performance depends on the underlying structure of the problem \citep{AI4}. Our contribution is to provide a principled one-fit, causal-forest-style estimator for the CKTE, based on a CKTE-targeted splitting score, together with corresponding consistency, asymptotic normality, and one-fit uncertainty quantification. In settings where the two groups share learnable structure, the shared-partition design can lead to substantial gains. 
We refer to our new method as ``Causal-DRF''. Adapting the results derived in \citep{näf2023confidence}, we show how an asymptotically valid test of $H_0 \colon \norm{\tauk}_{\H}=0$, as well as uniformly valid confidence bands around $\tauk(\ybf)$, can be constructed. This is notoriously difficult for kernel approaches, as laid out, for instance, in \citep{martineztaboada2023efficient}; because kernel-based tests involve degenerate U-statistics, approximating the (asymptotic) distribution, even if known, is computationally intensive. This is solved in \citep{martineztaboada2023efficient} for the kernel treatment effect using data splitting, allowing for an elegant asymptotic normality result for their test statistic. In contrast, in our methodology an additional resampling step in the forest construction provides a natural approximation of the sampling distribution of our test statistic for fixed $\xbf$  without increasing the complexity of the algorithm. This allows for the construction of an asymptotically valid test of $H_0\colon\PYgXz=\PYgXo$ as well as uniformly valid confidence bands for the witness function from one fit of the forest. Our Causal-DRF thus brings similar advantages to the estimation of the CKTE as the widely known causal-forest of \citep{wager2018estimation} brought to the CATE estimation.

We illustrate the capabilities of Causal-DRF through simulations and by applying the method to the 1991 Survey of Income and Program Participation (SIPP) pension data (previously analyzed in \citep{Benjamin401k, 401KData}, and \citep{ GRFDistributionalCausalEffects}, among others). This dataset contains three wealth measures that form the variable of interest ($\Ybf$), 401(k) eligibility as treatment assignment ($W$), and a range of covariates such as age and income ($\Xbf$). With Causal-DRF we can directly obtain the CKTE jointly for the three measures of wealth, using a single forest fit. This is in contrast to \citep{401KData} who use 3 independent quantile regressions and to \citep{GRFDistributionalCausalEffects} who combine the wealth variables of the dataset into one. This allows us to analyze the effect of 401(k) eligibility on the three wealth measures using the witness function and to test the equality of the conditional distribution for both groups for single test points. 

\subsection{Related Work}

Several recent papers consider the (unconditional) kernel treatment effect, such as \citep{Counterfactualmeanembeddings} and \citep{fawkes2022doubly}. In particular, \citep{martineztaboada2023efficient} develop an efficient test for $H_0: \P_{\Ybf^{0}} = \P_{\Ybf^{1}}$ using the kernel treatment effect. However, just as with the ATE, the kernel treatment effect obscures the heterogeneity of the treatment across different individuals. As such, \citep{CATEGeneralization} and \citep{näf2023confidence}, which consider the \emph{conditional} kernel treatment effect by separately estimating the conditional mean embeddings on both the treatment and control groups, are of special interest. The closest prior work to our current proposal is the one taken in \citep{näf2023confidence} in which two separate DRFs are fitted to obtain the conditional kernel treatment effect with uncertainty quantification and our theoretical analysis is built directly on their work. In fact, on a technical level, we begin the theoretical analysis in Appendix \ref{DRFproofcorrection} by fixing a problematic argument in their proof of a crucial result. With this fix, our proofs continue similarly as in \citep{näf2023confidence}, although considerably more care is needed in case of Causal-DRF. 

The main methodological difference from \citet{näf2023confidence} is that Causal-DRF uses a single forest with a CKTE-targeted splitting criterion shared across both groups, rather than fitting two separate forests. In case of the CATE it is well-known that separate estimation of the mean function can be problematic. For instance, different estimation errors in the two mean estimates may lead to spurious effects, especially when one of the two groups has fewer observations available than the other. We refer to \citep[page 47]{wager2023causalinference} for a more detailed explanation of these two problems in the context of the CATE estimation. The estimation of the CKTE will likely display similar behavior. This is particularly apparent in the case of kernel estimates as in \citep{CATEGeneralization}, as the two kernel mean embedding estimates are likely to have different regularization parameters, leading to a regularization bias in the CKTE. Although there is no explicit regularization with DRF, similar problems likely arise. Similarly, covariate shifts in the two groups might lead to further bias if the propensity score varies strongly with $\Xbf$. In the simulation studies we consider, using two separate DRFs can have detrimental coverage for smaller sample sizes, even in a relatively simple example, while Causal-DRF successfully maintains coverage even for $n=250$ samples. Nonetheless, we note that this shared-partition construction is not uniformly superior to separate estimation; rather, the two approaches impose different inductive biases, as in the corresponding CATE literature \citep{AI4}. When the treatment and control groups share learnable structure, the shared-partition design can be advantageous, whereas separate fits may be preferable when group-specific structure dominates. In our simulation settings, Causal-DRF tends to yield improved finite-sample coverage and competitive or better estimation accuracy, while also requiring only one forest fit.

\subsection{Notation and RKHS Background}
Here we collect the notation used throughout the main paper. We observe an i.i.d. sample $\Zbf_i=(\Ybf_i, W_i, \Xbf_i)$, $i=1,\ldots,n$, where $\Ybf_i$ is an outcome in $\R^d$, $W_i \in \{0,1\}$ indicates treatment status and $\Xbf_i$ are covariates taking values in $\R^p$. We denote the conditional distributions in each group as $\PYgXz$ and $\PYgXo$, respectively. We collect the data as $\Zcal_{n}=\{\Zbf_1, \ldots, \Zbf_n\}$.

Let $\left(\H, \langle\cdot,\cdot\rangle\right)$ be the RKHS induced by the positive definite, bounded, and continuous kernel $k\colon\R^d\times\R^d \to \R$ (see \citealt[Chapter 2.7]{hilbertspacebook} for an introduction to RKHS theory). Continuity of $k$ ensures that $\H$ is separable~\citep[Theorem 2.7.5]{hilbertspacebook}. For a random element $\xi$ taking values in the (separable) Hilbert space $\H$ with $\E[\|\xi \|_{\H}] < \infty$, we define its expected value in $\H$ by
\(\E[\xi]= \int_{\Omega} \xi d \P \in \H\),
where the integral is to be understood in a Bochner sense~\citep[Chapter 3]{hilbertspacebook}. Because $\H$ is separable, this integral is well defined.  When $\E[\| \xi \|_{\H}^2] < \infty$, let the variance of $\xi\in\H$ be written as \(
\Var(\xi)=\E[\|\xi \|_{\H}^2] - \| \E[\xi] \|_{\H}^2\). For a sequence of random elements $ \xi_n$ in $\H$, we denote by $\xi_n \stackrel{D}{\to} \xi$ convergence in distribution, i.e., for all bounded and continuous functions $F\colon \H \to \R$, we have $\E[F(\xi_n)] \to \E[F(\xi)]$ as $n \to \infty$. If, for all $f \in \H$, we have $\langle \xi,f \rangle\sim N(0, \sigma_f^2)$ for some $\sigma_f > 0$, we write $\xi \sim N(0, \boldsymbol{\Sigma})$ with $\boldsymbol{\Sigma}$ a self-adjoint Hilbert-Schmidt (HS) operator satisfying $\langle \boldsymbol{\Sigma}f , f\rangle=\sigma_f^2$. 
In this case, we also write $\xi_n \stackrel{D}{\to} N(0, \boldsymbol{\Sigma})$, if $\xi_n \stackrel{D}{\to} \xi$.

\section{Motivational Examples}
We consider four toy examples of potential causal effects for univariate outcomes ($d=1$), covariates of dimension $p=5$, and sample size $n=1000$. In all examples, we take $X_1, \ldots, X_p$ to be drawn i.i.d. from the uniform distribution on $[2,3]$, and $\Prob(W=1\mid \Xbf)=\Prob(W=1)=0.5$. In the first example, there is no effect such that $\PYgXz, \PYgXo$ are both Gaussian with mean 0 and variance 1 ($N(0,1)$). For the second example, there is a mean effect, where $\PYgXz=N(0,1)$, while $\PYgXo=N(x_1,1)$. In the third example, there is no mean effect; however, the variance decreases when changing from the control to the treatment group: $\PYgXz=N(0,1)$ and $\PYgXo=N(0,1/x_1^2)$. Finally, we consider an example where both the variance and the mean shift when changing from control to treatment group $\PYgXz=N(0,1)$ and $\PYgXo=N(x_1,x_1^2)$. In all examples, we choose the test point to be $\xbf=(2.5, \ldots, 2.5)^{\top}$. In addition, we assume in this example that an increase in $Y$ is desired due to the treatment.

Figures \ref{Example_12} and \ref{Example_34} show both the true underlying conditional densities as well as our estimation method with confidence bands. As mentioned above, the estimated witness functions follow the difference in densities between the treatment group in red and the control group in blue. When the difference in densities is negative at $\ybf$, the value of the estimated witness function tends to be negative, and vice-versa. In particular, negative values indicate that it is less likely to observe an outcome from the treatment population at $\ybf$. Thus the witness function maps out the behavior of the treatment effect for fixed $\xbf$. As the additional confidence bands are simultaneously valid for all $y$, any crossing of zero of the lower or upper confidence bands indicates a rejection of $H_0\colon\PYgXz=\PYgXo$. In the left graph in Figure \ref{Example_12} the upper and lower confidence bands are far from zero, indicating that the method recognizes that there is no effect. For all other examples, Causal-DRF indicates correctly that there is a significant effect, as in each example the confidence bands have regions that do not intersect zero. Even in challenging cases such as Example 3, where there is only a variance shift, Figure \ref{Example_34} shows a clear rejection of $H_0$. 

These examples indicate how the CKTE, estimated through Causal-DRF, can give a clearer picture of conditional treatment effects that go beyond differences in expectation. Moreover, the confidence bands offer a convenient tool to assess whether the effect is statistically significant. We note that the analysis for the last two cases is quite different if only the CATE is considered. First, the CATE will miss the decrease in variance in the third example. Second, while the last example may appear to be a success when only expectations are considered, only examining CATE will lead to missing differing variances in the treatment and control groups.

\begin{figure}[ht!]
    % \centering
    \begin{minipage}{0.47\textwidth}
        \subfigure{
        \includegraphics[width=0.46\linewidth]{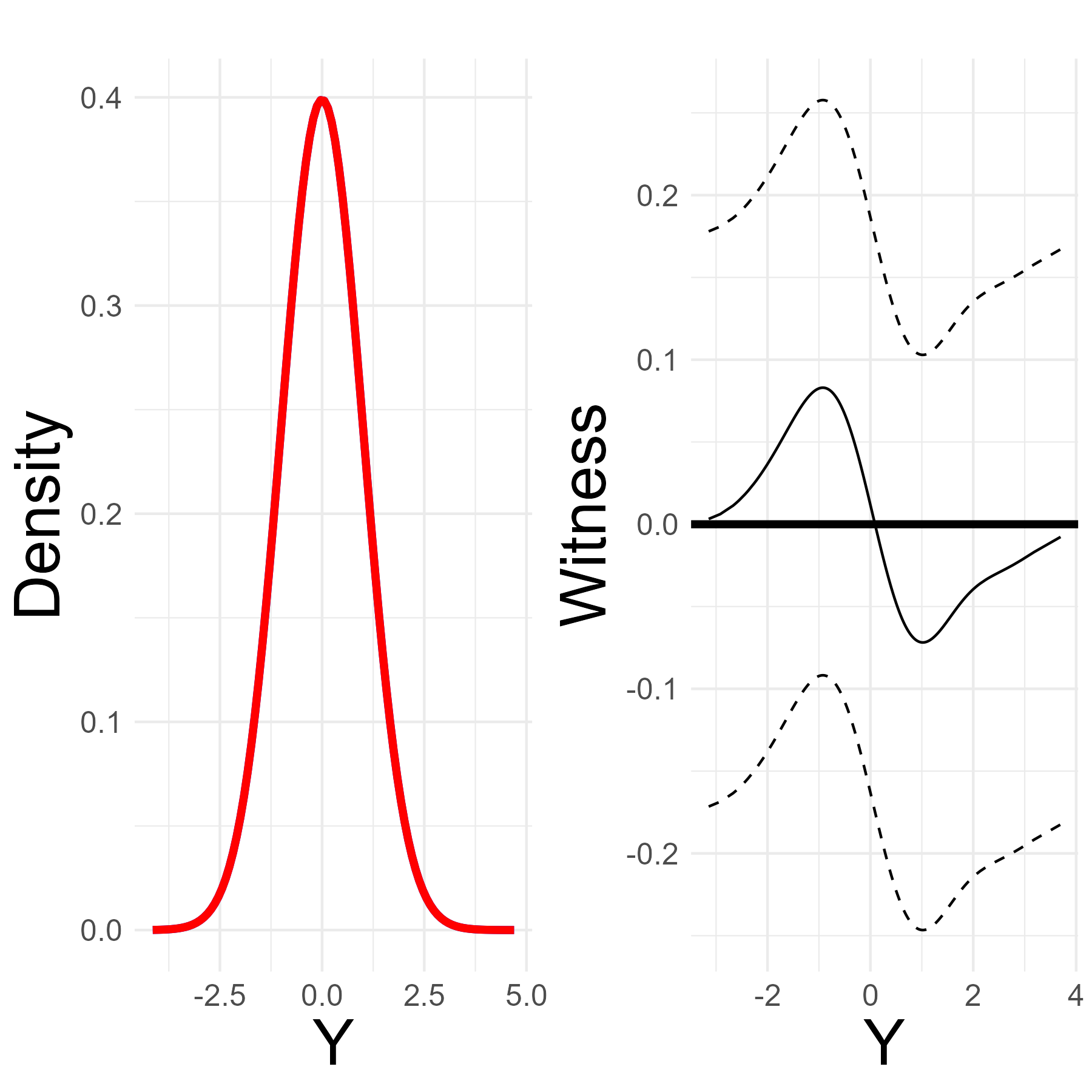}
        %\label{fig:image1}
    }
    \subfigure{
        \includegraphics[width=0.46\linewidth]{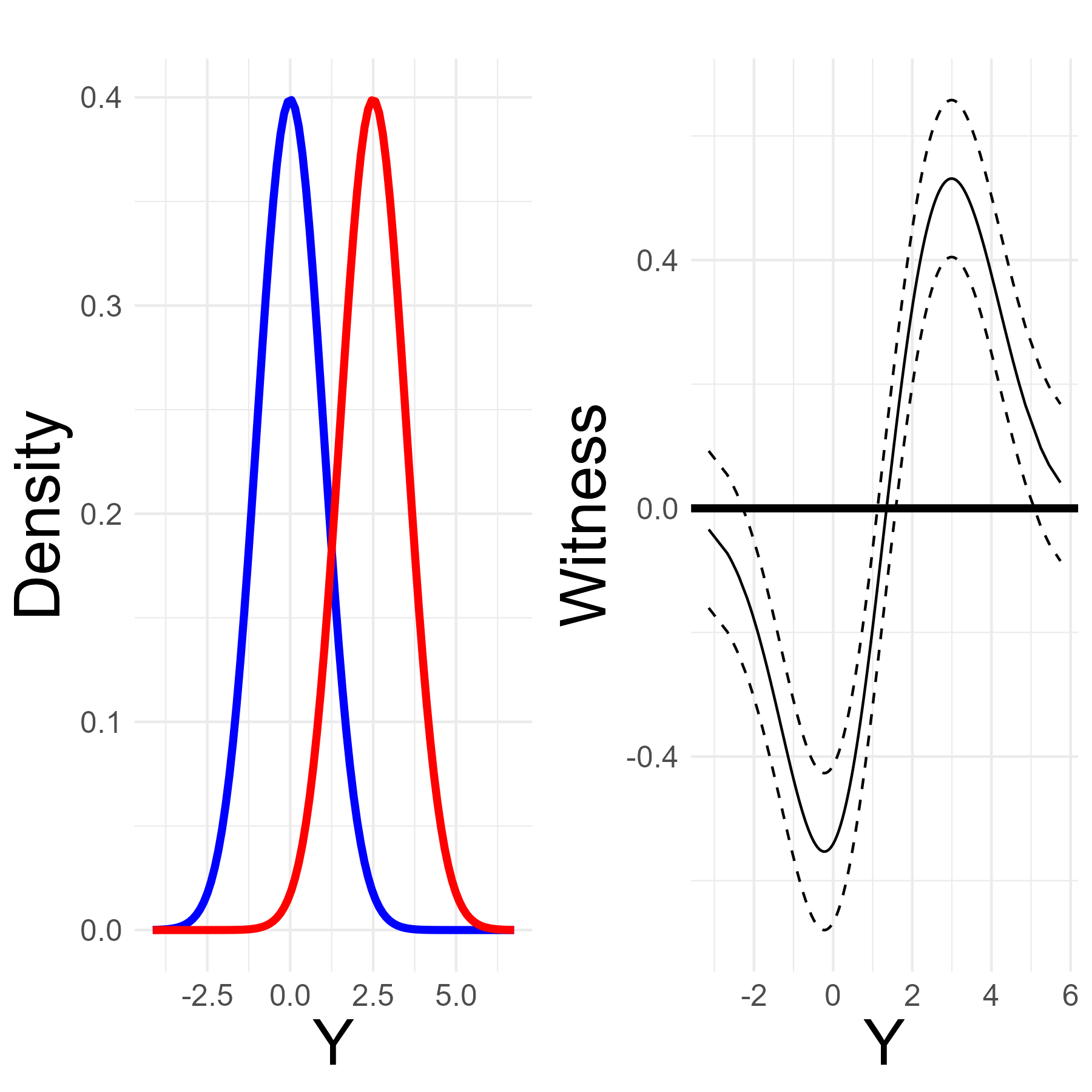}
        %\label{fig:image2}
    }
    \caption{Density and Causal-DRF estimates for Example 1 (left, no effect) and Example 2 (right, mean shift). Blue: Density for the control group $W=0$, Red: Density for the treatment group $W=1$.}
    \label{Example_12}
    \end{minipage}
    \hfill
    \begin{minipage}{0.47\textwidth}
        \subfigure{
        \includegraphics[width=0.46\linewidth]{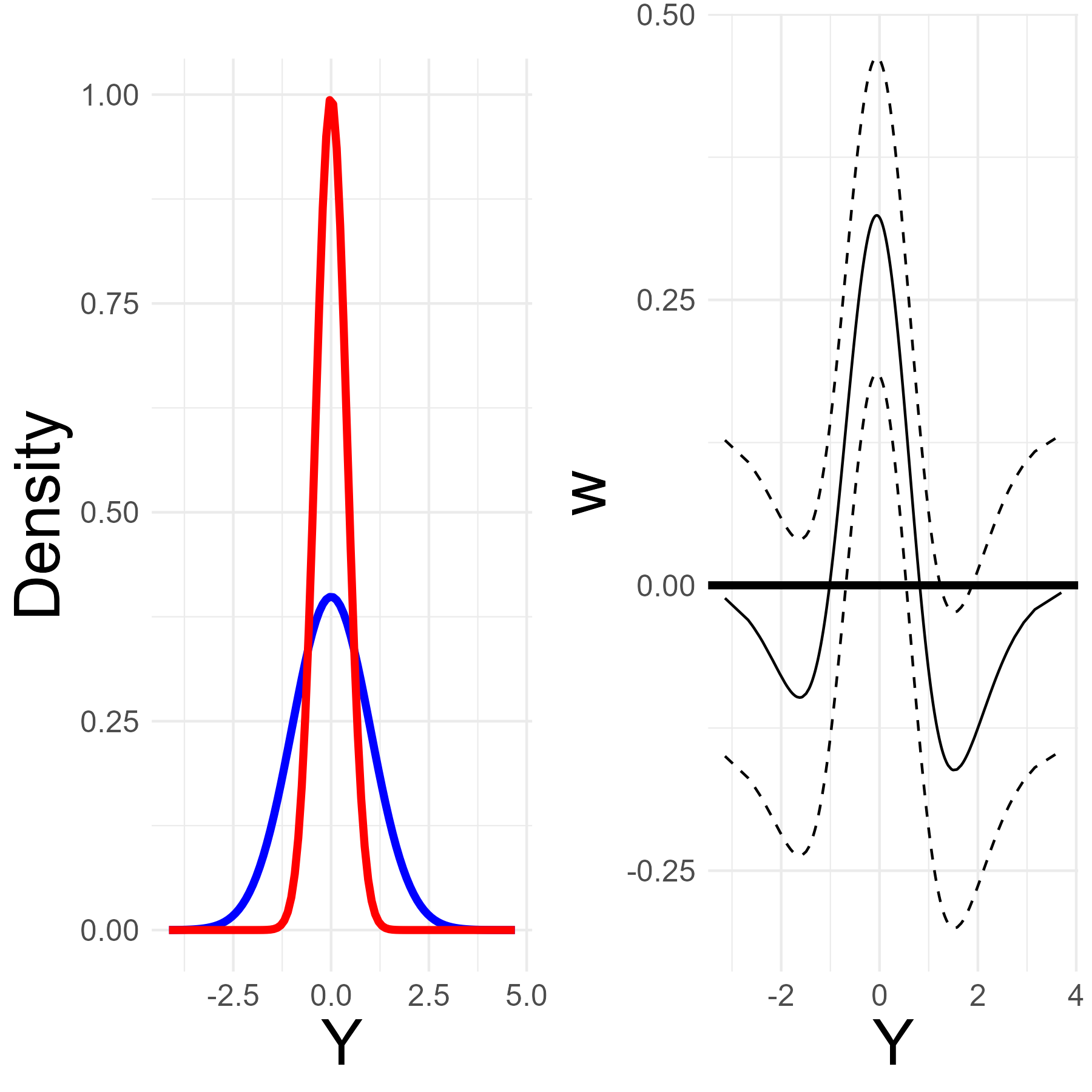}
        %\label{fig:image1}
    }
    \subfigure{
        \includegraphics[width=0.46\linewidth]{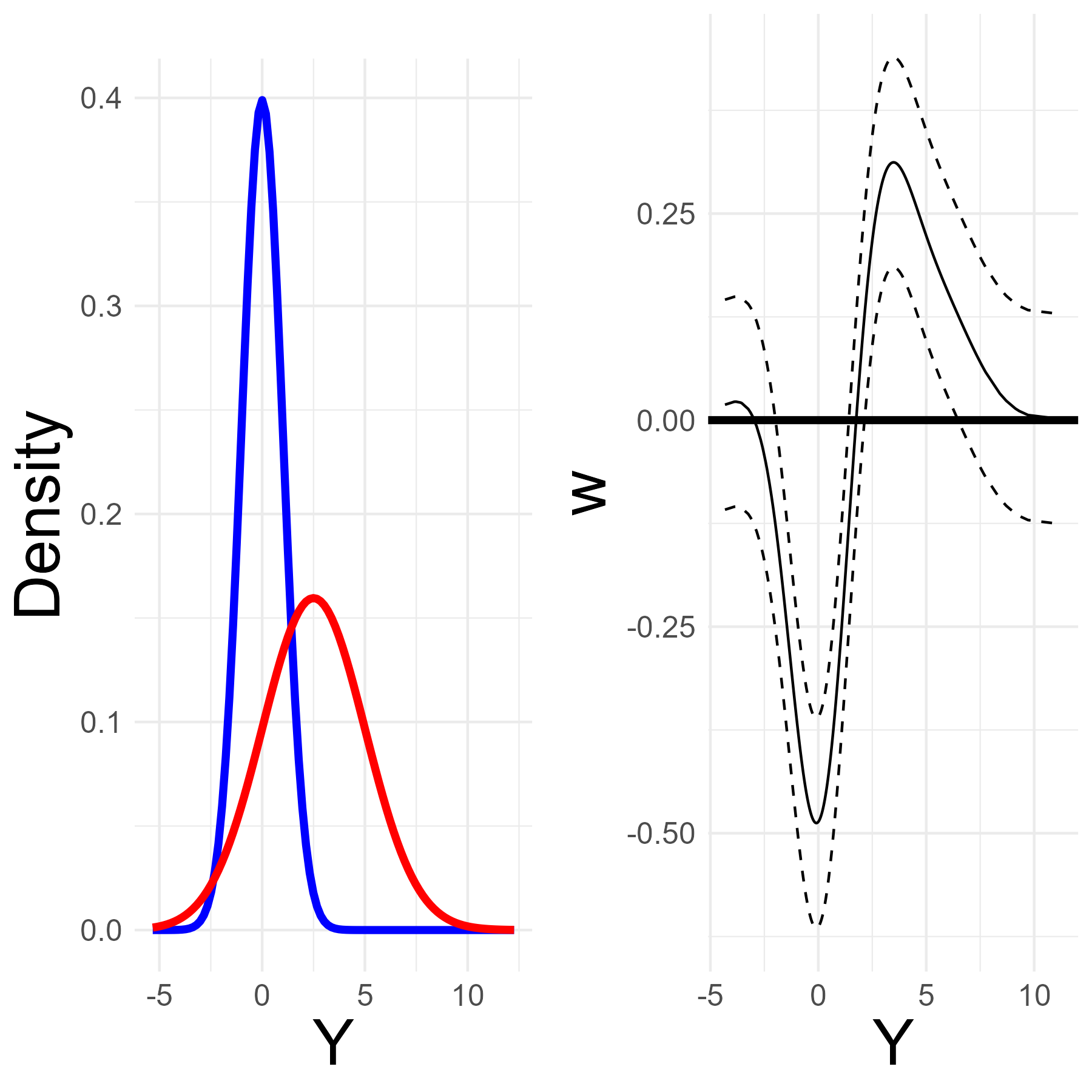}
        %\label{fig:image2}
    }
    \caption{Density and Causal-DRF estimates for Example 3 (left, variance effect) and Example 4 (right, mean shift + variance effect). Blue: Density for the control group $W=0$, Red: Density for the treatment group $W=1$.}
    \label{Example_34}
    \end{minipage}
\end{figure}

% \begin{figure}[ht!]
%     \centering
%     \subfigure{
%         \includegraphics[width=0.4\linewidth]{Example3.png}
%         %\label{fig:image1}
%     }
%     \subfigure{
%         \includegraphics[width=0.4\linewidth]{Example4.png}
%         %\label{fig:image2}
%     }
%     \caption{Density and Causal-DRF estimates for Example 3 (left, variance effect) and Example 4 (right, mean shift + variance effect). Blue: Density for the control group $W=0$, Red: Density for the treatment group $W=1$.}
%     \label{Example_34}
% \end{figure}

\section{Causal-Distributional Random Forests}\label{sec:DRF-recall}\label{sec:background}
%\subsection{Theoretical Background}\label{sec:background}

We start by reviewing the core DRF algorithm; for a full treatment, see \cite{DRF-paper} and \cite{näf2023confidence}. The DRF estimate of the CME can be viewed as a weighted average over the training points. These weights $\hwix$ are found by fitting a random forest type algorithm within the RKHS $(\H, k)$. Specifically, a forest of $N$ trees is constructed. Each tree recursively splits the observations into sets $\mathcal{I}_L=\{i: X_{ij} < c \}$ and $\mathcal{I}_R=\{i: X_{ij} \geq c \}$, with a splitting value $c \in \R$ and a splitting feature $j$ chosen by a splitting criterion that is allowed to depend on $\Ybf$. In the original DRF algorithm, the Maximum Mean Discrepancy (MMD) statistic~\citep{gretton2007kernel} is used as a splitting criterion. This choice of splitting criterion can be interpreted as an analogue of the traditional CART splitting criterion~\citep{breiman2001random}, but rather applied in the RKHS. Combining this idea with principles of forest construction developed by \citep{wager2018estimation, athey2019generalized}, we develop a novel splitting criterion in the RKHS that can be used to estimate $\tauk$. That is, if for a given node, the estimate in the left child is given as
\vspace{-.3em}
\begin{align*}
    \hat{\tau}_L&=\frac{1}{|\mathcal{I}_L, W_i=1|}\sum_{i\in\mathcal{I}_L, W_i=1}  \kYi\\
    &\qquad- \frac{1}{|\mathcal{I}_L, W_i=0|}\sum_{i\in\mathcal{I}_L, W_i=0}  \kYi,
\end{align*}
and similarly with $\hat{\tau}_R$, then we split such that a weighted version of
\begin{align} \label{eq: splitting}
 %   &\Big\lVert
 %   \frac{1}{|\mathcal{I}_L, W_i=1|}\sum_{i\in\mathcal{I}_L, W_i=1}  \kYi - \frac{1}{|\mathcal{I}_L, W_i=0|}\sum_{i\in\mathcal{I}_L, W_i=0}  \kYi
 %   - \\
 %   &\Big( \frac{1}{|\mathcal{I}_R, W_i=1|}\sum_{i\in\mathcal{I}_R, W_i=1} \kYi- 
 % \frac{1}{|\mathcal{I}_R, W_i=0|}\sum_{i\in\mathcal{I}_R, W_i=0} \kYi \Big)
 %   \Big\rVert_{\H}^2\\
  &\| \hat{\tau}_L - \hat{\tau}_R \|_{\H}^2 \nonumber\\
% &=\Big\lVert
%     \frac{1}{|\mathcal{I}_L, W_i=1|}\sum_{i\in\mathcal{I}_L, W_i=1}  \kYi - \frac{1}{|\mathcal{I}_L, W_i=0|}\sum_{i\in\mathcal{I}_L, W_i=0}  \kYi
%     - \nonumber \\
%     &\Big( \frac{1}{|\mathcal{I}_R, W_i=1|}\sum_{i\in\mathcal{I}_R, W_i=1} \kYi- 
%   \frac{1}{|\mathcal{I}_R, W_i=0|}\sum_{i\in\mathcal{I}_R, W_i=0} \kYi \Big)
%     \Big\rVert_{\H}^2 \nonumber\\
        &=\Big\lVert
    \sum_{i\in\mathcal{I}_L}   \left(\frac{W_i }{|\mathcal{I}_L, W_i=1|}  - \frac{  1- W_i }{|\mathcal{I}_L, W_i=0|} \right) \kYi \nonumber \\
    &- \sum_{i\in\mathcal{I}_R}   \left(\frac{W_i }{|\mathcal{I}_R, W_i=1|}  - \frac{  1- W_i }{|\mathcal{I}_R, W_i=0|} \right) \kYi    \Big\rVert_{\H}^2
\end{align}
is maximal. This is a weighted version of the MMD statistic. Thus, using the same random approximation as in \citep{DRF-paper} for a translation-invariant kernel (see \kernelass{3} in Appendix \ref{App_Preliminaries}) this can be efficiently approximated using $S \in \mathbb{N}$ random functionals. 
% \begin{align*}
%  & \| \hat{\tau}_L - \hat{\tau}_R \|_{\H}^2 \approx \frac{1}{S}   \sum_{s=1}^S \frac{n_L n_R}{n_P^2} \left| \sum_{i \in \mathcal{I}_L} \nu_{i,L} \tilde{\varphi}_{\mathbf{\omega}_s}(\Ybf_i) - \sum_{j \in \mathcal{I}_R} \nu_{j,L} \tilde{\varphi}_{\mathbf{\omega}_s}(\Ybf_j) \right|^2 
%     % =& \frac{1}{S}\sum_{s=1}^S  \frac{n_L n_R}{n_P^2} \bigg| \frac{1}{|\mathcal{I}_L, W_i = 1|} \sum_{i \in \mathcal{I}_L} W_i \tilde{\varphi}_{\mathbf{\omega}_s}(\Ybf_i) - \frac{1}{|\mathcal{I}_L, W_i = 0|} \sum_{i \in \mathcal{I}_L} (1 - W_i) \tilde{\varphi}_{\mathbf{\omega}_s}(\Ybf_i) \\
%     % &\quad\quad\quad\quad\quad\quad- \frac{1}{|\mathcal{I}_R, W_i = 1|} \sum_{j \in \mathcal{I}_R} W_j \tilde{\varphi}_{\mathbf{\omega}_s}(\Ybf_j) + \frac{1}{|\mathcal{I}_R, W_i = 0|} \sum_{j \in \mathcal{I}_R} (1 - W_j) \tilde{\varphi}_{\mathbf{\omega}_s}(\Ybf_j) \bigg|^2.
% \end{align*}
For a set of candidate variables, $\I_l$, $\I_R$ are chosen to maximize this criterion. Details are given in Appendix~\ref{appendix:splitting}.

For each fitted tree $l=1,\ldots,N$, let $\Lcal(\xbf)$ be the leaf node of $\xbf$. 
The estimate of $\tauk$ is then an average over the $N$ trees.
% \begin{align}\label{eq:htauk}
%     \htauk = &\frac{1}{N}\sum_{l=1}^N \Big(\frac{1}{ | \{j: \Xbf_j \in \Lcal_l(\xbf), W_j=1 \}|}\sum_{\Xbf_j\in\Lcal_l(\xbf), W_j=1} k(\Ybf_j,\cdot) - \nonumber \\
%     &\frac{1}{ | \{j: \Xbf_j \in \Lcal_l(\xbf), W_j=0 \}|}\sum_{\Xbf_j\in\Lcal_l(\xbf), W_j=0} k(\Ybf_j,\cdot)  \Big). 
% \end{align}
This approach has the same computational complexity of the original DRF algorithm, specifically $\O(S \times N \times p \times n \log n)$; see \citep{DRF-paper}. Collecting all $k(\Ybf_j, \cdot)$, one can rewrite $\htauk$ conveniently as
\begin{align}\label{eq: weightrepresentation}%\label{eq:htauk}
    \htauk = \sum_{i=1}^n \hat{w}_i(\xbf)  k(\Ybf_i, \cdot),
\end{align}
with
\begin{align*}
    \hat{w}_i(\xbf)=&\frac{1}{N} \sum_{l=1}^N \Big(\frac{\mathbf{1}\{ \Xbf_i \in \Lcal_l(\xbf)\}W_i}{ | \{j: \Xbf_j \in \Lcal_l(\xbf), W_j=1 \}|} \\
    &- \frac{\mathbf{1}\{ \Xbf_i \in \Lcal_l(\xbf)\}(1-W_i)}{ | \{j: \Xbf_j \in \Lcal_l(\xbf), W_j=0 \}|} \Big) .
\end{align*}

Finally, to construct tests and confidence bands, we consider the following subsampling approach similar to that of  \citep{näf2023confidence}. For a subset $\HS \subset \{1,\ldots,n\}$, denote by $\htaukHS$ the version of $\htauk$ that uses trees built with data from $\HS$. We draw $\HS$ by sampling $n$ i.i.d.\ random variables 
$U_i \sim \mathrm{Bernoulli}(1/2)$ and set 
$\HS=\{i\colon U_i=1 \}$. Thus, we fit $B$ subforests with $L$ trees, such that $N=B\times L$. The resulting estimates $\htaukHSb$, $b=1,\ldots, B$, correspond to an approximation of the sampling distribution of $\htauk$; we formalize this notion in the following section. Note that $|\HS|/n \to 1/2$ almost surely and this approach is motivated from the half-sampling approach of \cite{athey2019generalized}, though our independent sampling allows for a somewhat simpler theoretical treatment. The point estimate of the CME is then calculated by averaging the point predictions over all the half-samples, 
\(\htauk = \frac{1}{B} \sum_{b=1}^B \htaukHSb\)
which corresponds to \eqref{eq: weightrepresentation}. Algorithms~\ref{pseudocodeU} and~\ref{pseudocode} in Appendix \ref{App_Preliminaries} provide pseudocode for this procedure. This approach again has the same complexity as one fit of DRF \citep{DRF-paper}, except that $N=B \times L$ should be chosen quite large.

The Causal-DRF estimator can be conveniently written in matrix form. 
Consider the matrix $\Kbf=(k(\Yibf,\Ybf_{j} ))_{i=1,\ldots n,j=1,\ldots n}$. Denote by $\hwbf \in\R^{n}$ the vector that concatenates the weights in \eqref{eq: weightrepresentation} and similarly let $\hwbf^{\HS_b} \in\R^{n}$ be the weights of the $b$th subforest. Let $ \kbf=( k(\Ybf_1,\cdot ) , \ldots, k(\Ybf_{n},\cdot ) )^{\top}$, and denote by $ \kbf(\ybf)=( k(\Ybf_1,\ybf ) , \ldots, k(\Ybf_{n},\ybf ) )^{\top}$ for $\ybf\in\R^d$.
Then, $\htauk = \hwbf^{\top}  \kbf, \ \ \htaukHSb = ( \hwbf^{\HS_b})^{\top}  \kbf$. Following the software implementation, we will take $k$ to be the Gaussian kernel given in \eqref{Gaussiankernel} in Appendix \ref{App_Preliminaries}. Note that with this choice, it holds that $\sup_{\ybf, \ybf'} k(\ybf, \ybf') \leq 1$. For the Gaussian kernel a bandwidth $\sigma$ needs to be chosen, which may influence the performance of the algorithm. Here, we follow \citep{DRF-paper, näf2023confidence} and choose $\sigma$ using the median heuristic of \citep{gretton2012optimal}. This appears to work quite well, though we note the (empirically) optimal choice of $\sigma$ for (Causal-)DRF remains an open question, as does the influence of $\sigma$ on the empirical performance.

To calculate tests and confidence bands for $\htauk$, we first compute $q_{n,\alpha}$, the $1-\alpha$ quantile of the $B$ many draws from
\begin{align}\label{unscaledresampledstats}
    &\|\htaukHSb - \htauk  \|_{\H}^2= ( \hwbf^{\HS_b} -  \hwbf)^{\top}  \Kbf( \hwbf^{\HS_b} -  \hwbf).
\end{align}
Then a test for $H_0\colon\PYgXz=\PYgXo$ can be constructed as
\begin{align}\label{Test}
\phi(\Zcal_{n})=\1 \left \{\left \| \htauk \right \|_{\H}^2 > q_{n, \alpha} \right \},
\end{align}
while a simultaneous $1-\alpha$ confidence band around $\htauk$ can be constructed as:
\begin{align}\label{CB}
    \bandybf = [\htauk(\ybf)-\sqrt{q_{n,\alpha} }, \htauk(\ybf)+\sqrt{q_{n,\alpha} }].
\end{align}
The test in \eqref{Test} should be interpreted as a test of equality of the conditional outcome distributions at the covariate value \(\xbf\). In particular, rejection indicates that treatment has a distributional effect at \(\xbf\), but the test statistic alone does not reveal the direction of that effect or where the two distributions differ. Such directional and region-specific information is instead conveyed by the witness function and its confidence band in \eqref{CB}: regions where the band lies strictly above zero correspond to outcomes locally more likely under treatment, while regions where it lies strictly below zero correspond to outcomes locally more likely under control. We note that the confidence band in \eqref{CB} is uniform in $\ybf$ but only valid for a fixed $\xbf$. For several test points $\xbf$, one has to account for multiple testing issues.

In the next section we prove consistency and asymptotic normality of $\htauk$, that $\phi$ is an asymptotically valid test and that $\bandybf$ provides asymptotically uniform coverage over $\ybf \in \R^d$. 

\section{Theoretical Development}
In this section we establish theoretical guarantees for the Causal-DRF algorithm, including its consistency and asymptotic normality. Notably, we are able to establish similar guarantees as \citep{näf2023confidence}, who rely on fitting separate forests for the treatment and control groups. 

For convenience, denote the kernel embeddings as $\mu^{(j)}(\xbf)=\E[ k(\Ybf^j,\cdot)  \mid \Xbf \myeq \xbf]$, for $j \in \{0,1\}$ and $\mu(\xbf)=\E[ k(\Ybf,\cdot)  \mid \Xbf \myeq \xbf]$. To begin, we assume the forest satisfies Assumptions \ref{forestass1}--\ref{forestass6CausalDRF} and \ref{forestass1starCausalDRF} in Appendix \ref{Sec_Forestass}. These assumptions are a combination and refinement of the assumptions taken in \citep{wager2018estimation} and \citep{näf2023confidence}. In addition, we place the following assumptions on the data-generating process; \dataass{2}-\dataass{5} hold for \(j\in\{0,1\}\). 
\begin{enumerate}[label=(\textbf{D\arabic*})]
    \item\label{dataass1} $\mathbf{X}_1,\ldots, \mathbf{X}_n$ are i.i.d. on $[0,1]^p$, and admit a density bounded away from $0$ and infinity.
    \item\label{dataass2} The mappings $\xbf \mapsto \mu^{(j)}(\xbf) $ are Lipschitz
    % for $j \in \{0,1\}$
    .
    \item\label{dataass4} $\Var(k(\Ybf^j,\cdot)\mid \Xbf=\xbf) = \E[\| k(\Ybf^j,\cdot) \|_{\H}^2 \mid \Xbf=\xbf] - \|\E[ k(\Ybf^j,\cdot)  \mid \Xbf=\xbf]\|_{\H}^2 > 0$
    % , for $j \in \{0,1\}$
    .
    \item\label{dataass6} For all $f \in \H \setminus \{0\}$, $\Var( \langle k(\Ybf^j,\cdot), f \rangle\mid \Xbf=\xbf)= \Var( f(\Ybf^j)\mid \Xbf=\xbf) > 0$
    % , for $j \in \{0,1\}$
    .
    \item\label{dataass7} For all $f \in \H\setminus \{0\}$, $\xbf \mapsto \E[\left| f(\Ybf^j)\right|^{2}\mid \Xbf=\xbf]$ is Lipschitz
    % , for $j \in \{0,1\}$
    . 
\end{enumerate}
As noted in \citep[page 10]{näf2023confidence}, \dataass{3} is automatically satisfied for the Gaussian kernel under the assumption that the conditional distributions $\Ybf^j \mid \Xbf=\xbf$ for $j \in \{0,1\}$ are not point measures. We also adopt the following traditional causal assumptions.
\begin{enumerate}[label=(\textbf{C\arabic*})]
    \item\label{causalityass1} Positivity (overlap): there exists $\varepsilon > 0$ such that for all $\xbf$, 
    \begin{align}\label{overlapcondition}
        \varepsilon < \Prob(W=1 \mid \Xbf=\xbf)  < 1-\varepsilon
    \end{align}
   \item\label{causalityass2} Unconfoundedness (no unmeasured confounding): $(\Ybf^0, \Ybf^1) \indep W \mid \Xbf$
\end{enumerate}
Finally, we place assumptions \ref{kernelass1}--\ref{kernelass4} (Appendix~\ref{App_Preliminaries}) on the kernel function $k$. All of these assumptions are met by the Gaussian kernel 
% given in 
\eqref{Gaussiankernel}.

Next, we consider the setting where the number of trees goes to infinity ($N \to \infty$). The forest estimator $\htauk$ is then the average over all possible $\binom{n}{s_n}$ many trees:
\begin{equation}\label{finalestimator}
    \htauk = \frac{1}{\binom{n}{s_n}}  \sum_{i_1 < i_2 < \cdots < i_{s_n}} \E_{\mathcal{E}} \left[\Gamma(\xbf; \mathcal{E}, \{\mathbf{Z}_{i_1}, \ldots, \mathbf{Z}_{i_{s_n}}\})\right],
\end{equation}
where $\mathcal{E}$ encodes potential auxiliary randomness from the trees (such as the dimensions considered in the splitting criterion), and
\begin{align}\label{eq:onetree}
    &\E_{\mathcal{E}} \left[\Gamma(\xbf; \mathcal{E}, \{\mathbf{Z}_{i_1}, \ldots, \mathbf{Z}_{i_{s_n}}\})\right]\nonumber\\
    &= \frac{1}{ | \{j: \Xbf_j \in \Lcal(\xbf), W_j=1 \}|}\sum_{\Xbf_j\in\Lcal(\xbf), W_j=1} k(\Ybf_j,\cdot) \nonumber \\
    &-\frac{1}{ | \{j: \Xbf_j \in \Lcal(\xbf), W_j=0 \}|}\sum_{\Xbf_j\in\Lcal(\xbf), W_j=0} k(\Ybf_j,\cdot).
\end{align}
This simplification is also used in~\citep{wager2018estimation, athey2019generalized, DRF-paper}, and \citep{näf2023confidence}. Similarly, $\htaukHS$ is also an infinite forest predictor but using only the samples in $\HS$. Given these assumptions, we have that:
\begin{restatable}{theorem}{asymptoticnormality}\label{thm: asymptoticnormality}
Assume conditions~\ref{forestass1}--\ref{forestass6CausalDRF},~\ref{forestass1starCausalDRF},~\ref{dataass1}--\ref{dataass7},~\ref{kernelass1}--\ref{kernelass4} and \ref{causalityass1}--\ref{causalityass2} hold. Then, there exists $\sigma_n \to 0$, such that, for some self-adjoint HS operator $\boldsymbol{\Sigma}_{\xbf}$, we have
\begin{align}\label{eq:asymptoticnormality}
   \xi_n(\xbf)=\frac{1}{\sigma_n}  \left( \htauk -  \tauk \right) \stackrel{D}{\to} N(0, \boldsymbol{\Sigma}_{\xbf}).
\end{align}
% where $\boldsymbol{\Sigma}_{\xbf}$ is a self-adjoint HS operator. 
\end{restatable}
%This result is analogous to \citep[Theorem 6]{näf2023confidence} but requires weaker assumptions. 
We note that to prove consistency, more specifically, 
\begin{equation}\label{eq: rate} 
\norm{\htauk- \tauk}_\H = \O_{p}\left(n^{-\gamma} \right)
\end{equation}
for any $\gamma \leq \frac{1}{2} \min\left( 1- \beta, \frac{\log((1-\alpha)^{-1})}{\log(\alpha^{-1})} \frac{\pi}{p} \cdot \beta \right)$, it is enough to assume \ref{forestass1}--\ref{forestass5},~\ref{dataass1} and \ref{dataass2},~\ref{kernelass1}--\ref{kernelass4} and \ref{causalityass1}--\ref{causalityass2}. Here $\alpha \leq 0.2$ is the constant taken from Assumption \forestass{4}, while $\beta$ is the subsampling rate defining $s_n$ in \forestass{5}, which is chosen as $s_n=n^{\beta}$. Moreover, using similar arguments as in \citep[Proposition 3]{DRFVariableimportance}, the assumption of a bounded kernel \kernelass{1} also implies
\begin{equation}\label{eq: meanconsistency}
   \E[\norm{\hat{\tau}_{n,k}(\Xbf)- \tau_k(\Xbf)}_\H] \to 0 ,
\end{equation}
where expectation is taken over the training data as well as $\Xbf$.

\begin{remark}
    Theorem \ref{thm: asymptoticnormality} is similar to Theorem 6 in \cite{näf2023confidence}, but the new tree construction includes several additional technical challenges. In particular, the key to proving Theorem \ref{thm: asymptoticnormality} is to control the ratio of variances of the first order approximation (Theorem \ref{thm:varianceassumption}), which is now more difficult with $W_i$ involved. We refer to Remark \ref{rmk:Difficulty} and to Theorem \ref{thm:varianceassumption} for more details.
\end{remark}

Unfortunately, the convergence in \eqref{eq:asymptoticnormality} is of limited use in practice, as one would also need to estimate self-adjoint HS operator. Instead, we follow a resampling approach as motivated by~\citet{athey2019generalized, näf2023confidence} to approximate the sampling distribution of $\xi_n(\xbf)$. Consider the infinite Causal-DRF $\htaukHS$, constructed using only the samples in $\HS$, where $\HS$ was sampled as described in Section \ref{sec:DRF-recall}. Under such a sampling scheme, randomness in $\htaukHS$ can be decomposed into randomness in the data and from random assignment to $\HS$. Thus, conditional on the data, we show that
\begin{align}\label{xinHSdef}
   \xi_{n}^{\HS}(\xbf)=\frac{1}{\sigma_n}\left( \htaukHS -  \htauk \right)
\end{align}
converges to the same (Gaussian) variable as $\xi_n(\xbf)$ \eqref{eq:asymptoticnormality}. 
With this result in hand, we can randomly sample $\HS$ and use the empirical distribution of $\xi_n^{\HS}(\xbf)$ to approximate the asymptotic distribution of $\xi_n(\xbf)$. Formalizing this result requires extending standard bootstrap arguments \citep{B2, kosorok2008introduction}. Specifically, Theorem~\ref{thm: asymptoticnormalityhalfsampling} shows that
\begin{align}\label{eq:asymptoticnormalityhalfsampling}
   \xi_{n}^{\HS}(\xbf) \xrightarrow[W]{D} N(0, \boldsymbol{\Sigma}_{\xbf}),
\end{align}
where $\xrightarrow[W]{D}$ denotes ``conditional convergence in distribution", characterized by the condition
\begin{align}\label{conditionalconvergence}
 \sup_{h \in \text{BL}_1(\H) }  \left |  \E[ h(\xi_{n}^{\HS}(\xbf)) \mid \Zcal_{n}] - \E[ h(\xi(\xbf)) ] \right | \stackrel{p}{\to} 0,
\end{align}
where $\text{BL}_1(\H)$ denotes the space of all bounded Lipschitz functions from $\H$ to $\R$ with Lipschitz constant bounded by 1 (i.e. $h \in \text{BL}_1(\H)$ satisfies $\sup_{f \in \H} |h(f)| \leq 1$ and $|h(f_1) - h(f_2) | \leq \| f_1 - f_2 \|_{\H}$ for all $f_1, f_2 \in \H$). 
Condition~\eqref{conditionalconvergence} implies that $\xi_{n}^{\HS}(\xbf)$ converges to $\xi(\xbf)$ in distribution in probability, conditionally on the data $\Zcal_{n}$. For details, see for example~\citep{B3};~\citep[Chapter 10]{kosorok2008introduction}.

\begin{restatable}{theorem}{asymptoticnormalityhalfsampling}\label{thm: asymptoticnormalityhalfsampling}
Assume conditions~\ref{forestass1}--\ref{forestass6CausalDRF},~\ref{forestass1starCausalDRF},~\ref{dataass1}--\ref{dataass7},~\ref{kernelass1}--\ref{kernelass4} and \ref{causalityass1}--\ref{causalityass2} hold. Then,~\eqref{eq:asymptoticnormalityhalfsampling} holds.
\end{restatable}

Theorem~\ref{thm: asymptoticnormalityhalfsampling} allows us to approximate the distribution of $\xi_n(\xbf)$ with the distribution of $\xi_n^{\HS}(\xbf)$. Fortunately, the distribution of $\xi_n^{\HS}(\xbf)$ can be approximated using the results of a single forest fit, as in \citep{athey2019generalized} and \citep{näf2023confidence}. To do so, we construct the forest by (1) sampling $B$ subsets $\HS_1, \ldots, \HS_B$ of $\{1,\ldots,n\}$, (2) fitting a Causal-DRF with $L$ trees and calculating the prediction $\htaukHSb$ within each half-sample $b=1,\ldots,B$, and (3) forming the final prediction $\htaukHSb$ by averaging the predictions of all the half-samples: $(\hat{\tau}_n^{\HS_b}(\xbf))_{b=1}^B$. Combining the continuous mapping theorem with the above results, it follows immediately that
\begin{align}\label{normconvergenceS}
    &\frac{1}{\sigma_{n}^2} \left \|  \htaukHS-\htauk\right \|_{\H}^2 \xrightarrow[W]{D} \|\xi \|^2,\nonumber\\
    &\ \     \frac{1}{\sigma_{n}^2} \left \|  \htauk-\tauk\right \|_{\H}^2 \stackrel{D}{\to} \|\xi \|^2,
\end{align}
% and
% \begin{align}\label{normconvergence}
%     \frac{1}{\sigma_{n}^2} \left \|  \htauk-\tauk\right \|_{\H}^2 \stackrel{D}{\to} \|\xi \|^2,
% \end{align}
with $\xi \sim N(0, \boldsymbol{\Sigma}_{\xbf})$. Thus if we consider $H_0\colon\PYgXz=\PYgXo $, then $\tauk=0$ and $\frac{1}{\sigma_{n}^2}\left \| \htauk  \right \|_{\H}^2$ has the same limiting distribution as its resampling bootstrap version
\begin{align}\label{eq:scaled_bootstrap-stat}
    \frac{1}{\sigma_{n}^2} \left \|  \htaukHS-\htauk\right \|_{\H}^2
\end{align}
given the data. Regardless of whether $\PYgXz=\PYgXo$ holds, we can approximate the distribution under $H_0$ by sampling from $\HS$. Consider choosing $c_{n, \alpha}$ to be the smallest value obtained from $B$ draws, with $B$ large enough to ensure that we have
\begin{align}\label{finitealphabound}
   \Prob\left(  \frac{1}{\sigma_{n}^2} \left \|  \htaukHS-\htauk\right \|_{\H}^2 >c_{n, \alpha} \mid \Zcal_{n}  \right) \leq \alpha.
\end{align}
The associated test $\phi(\Zcal_{n})$ is given by
\begin{equation}\label{Test2}
  \phi(\Zcal_{n})=\1 \left \{\frac{1}{\sigma_{n}^2}\left \| \htauk \right \|_{\H}^2 > c_{n, \alpha} \right \}.   
\end{equation}
Note that \eqref{Test2} is a slight variation of \eqref{Test}, better suited for theoretical analysis. That is, \eqref{Test2} uses the theoretical version of $\htauk$, with $N\to \infty$ and similarly the quantile $q_{n,\alpha}$ replaced by $\sigma_{n}^2 c_{n, \alpha}$, with $c_{n, \alpha}$ obtained as in \eqref{finitealphabound}. Though $\sigma_{n}^2$ could be estimated as well, we directly consider the resampled unscaled statistic~\eqref{unscaledresampledstats} and identify its $1-\alpha$ quantile, which corresponds to $\sigma_{n}^2 c_{n, \alpha}$. The following theorem shows that $\phi$ is of level $\alpha$ with power converging to $1$.  

\begin{restatable}{theorem}{powerprop}\label{powerprop}
Assume conditions~\ref{forestass1}--\ref{forestass6CausalDRF},~\ref{forestass1starCausalDRF},~\ref{dataass1}--\ref{dataass7},~\ref{kernelass1}--\ref{kernelass4} and \ref{causalityass1}--\ref{causalityass2} hold. Then, as $n \to \infty$,
\begin{itemize}
    \item[(i)] $\phi$ has a valid type-I error. That is, if $\PYgXz=\PYgXo$,
    \[
    \limsup_{n \to \infty} \Prob\left(\frac{1}{\sigma_{n}^2}\left \| \htauk  \right \|_{\H}^2 > c_{n, \alpha}\right) \leq \alpha.
    \]
    \item[(ii)] $\phi$ has power going to 1. That is, if $\PYgXz \neq \PYgXo$,
    \[
    \lim_{n \to \infty} \Prob\left(\frac{1}{\sigma_{n}^2}\left \| \htauk  \right \|_{\H}^2 > c_{n, \alpha}\right) = 1.
    \]
\end{itemize}
\end{restatable}

A confidence band for the conditional witness function $\ybf\mapsto \tauk(\ybf)$ that is jointly valid for all $\ybf$-values can be obtained straightforwardly by inverting the testing procedure.
Let $c_{n, \alpha}$ be as in~\eqref{finitealphabound}. We propose constructing an interval given by
\begin{align}\label{CIband}
    & \bandybf=\nonumber\\
    & [\htauk(\ybf)-\sqrt{c_{n,\alpha} C}\sigma_{n}, \htauk(\ybf)+\sqrt{c_{n,\alpha} C}\sigma_{n}],
\end{align}
where $C=\sup_{\ybf} k(\ybf,\ybf)$. Note that 
the boundedness of the reproducing kernel ensures that $C$ is finite; see Assumption~\ref{kernelass2}. In particular, $C = 1$ for the Gaussian kernel. In the following theorem we show that $\bandybf$ is a uniform $1-\alpha$ confidence band for the conditional witness function as a function of $\ybf$. 
\begin{restatable}{theorem}{coverprop}\label{coverprop}
Assume conditions~\ref{forestass1}--\ref{forestass5},~\ref{dataass1}--\ref{dataass7},~\ref{kernelass1}--\ref{kernelass4} and \ref{causalityass1}--\ref{causalityass2} hold. Then, as $n \to \infty$, 
\begin{align}\label{CIguarantee}
  \liminf_{n \to\infty}  \Prob \left( \bigcap_{\ybf} \left\{ \tauk(\ybf) \in   \bandybf \right\} \right) \geq 1-\alpha.
\end{align}
\end{restatable}
The confidence band in \eqref{CIband} is adapted analogously from the version in \eqref{CB}. We note that this band has uniform coverage in $\ybf$, thus the reason why any crossing with zero is a rejection, but pointwise coverage in $\xbf$. Thus, the test will not be valid over different test point $\xbf$ without a correction for multiple testing.

\section{Empirical Study: Impact of 401(k) Eligibility}\label{sec:Experimental}

We study the impact of 401(k) eligibility on wealth, based on a sample of households from Wave 4 of the 1990 Survey of Income and Program Participation (SIPP).
This dataset was used in several studies of the relationship between 401(k) plans and wealth, e.g., \citep{Benjamin401k, 401KData, GRFDistributionalCausalEffects}, and includes 9,915 observations. The variables of interest $\Ybf$ are three free wealth measures: ``Net Financial Assets'', ``Net Non-401(k) Financial Assets'' and ``Total Wealth''. Following \citep{Benjamin401k, 401KData, GRFDistributionalCausalEffects}, we consider 9 covariates $\Xbf$ that include age, income, education, family size, marital status, and IRA participation, as detailed in Table~\ref{table:variables} in Appendix~\ref{sec:401(k)table}. Finally, the treatment $W$ is taken to be 401(k) eligibility. While 401(k) eligibility is not randomly assigned, prior studies (see, e.g., \citep{401KData, GRFDistributionalCausalEffects} and the references therein) have commonly assumed unconfoundedness conditional on these nine covariates, supported by subject-matter expertise. The three wealth measures we consider are the same as in \citep{401KData}. While they use separate quantile regressions, we utilize the ability of Causal-DRF to estimate the CKTE of the three wealth measures jointly. For this reason, after taking the first 10 households as a test set, we train one Causal-DRF on the remaining 9,905 observations, taking the three variables of financial wealth as one response vector $\Ybf$ in $\R^3$. As one would expect, the three measures are highly skewed. We begin by scaling the three measures of wealth by subtracting the mean and dividing by the standard deviation. We then fit Causal-DRF with $N=4,000$ trees with $100$ trees per group ($B=40$). We focus on two test points $\xbf_1$, $\xbf_2$ out of the test set of 10 households. The two households were chosen to be relatively dissimilar to each other; the reference persons in each household have an age gap of 19 years and the two households have a large difference in family size and income. A more detailed description of the characteristics of the two households is given in the captions of Figures \ref{fig:401KResults1} and \ref{fig:401KResults2}.
Using the test statistic in \eqref{Test}, a test of $H_0: \Prob_{\Ybf \mid \Xbf=\xbf_j}^0=\Prob_{\Ybf \mid \Xbf=\xbf_j}^1$ is rejected for both points $j=1,2$. This indicates a significant impact of eligibility on wealth for these two points, though we did not account for multiple testing. To better understand the effect, Figures \ref{fig:401KResults1} and \ref{fig:401KResults2} show the (univariate) witness function for Net Financial Assets ($Y_1$), Net Non-401(k) Financial Assets ($Y_2$) and Total Wealth ($Y_3$). First, the significant result appears to stem mostly from the influence of eligibility on Net Financial Assets, while the other two wealth measures appear not significantly impacted by eligibility. This is in particular true for the older household, where there is barely an effect of eligibility on Net Non-401(k) Financial Assets. We thus focus on the effect on Net Financial Assets and Total Wealth. For Net Financial Assets the negative values for small amounts of assets indicate that we are more likely to obtain small values for people who are not eligible for 401(k) (the control group). On the other hand, high values of Net Financial Assets are more likely to come from people eligible for 401(k). The picture is not perfectly symmetric, which indicates a higher skewness of the distribution for $W=1$. As skewness in wealth distributions is usually oriented towards higher values, this is not entirely surprising. Interestingly, though there is a slightly different magnitude, the picture for Net Financial Assets is roughly the same for both test points. Though these are only two test households, it might be that for this variable alone, there is less heterogeneity over different households. On the other side of the spectrum, the picture for Total Wealth of the two test points is quite different. While for the older household, the witness function looks similar to the one of Net Financial assets, albeit with a higher skewness in the eligible group, Total wealth for the younger household in Figure~\ref{fig:401KResults1} has another dip in the witness function for higher values of income. This indicates there is a certain spectrum of higher incomes where it is more likely to observe people who are not eligible. 
These discussions should illustrate the more comprehensive analysis of the treatment effect that are possible with Causal-DRF.

\begin{figure}[t]
    \begin{minipage}{0.48\textwidth}
        \centering
        \includegraphics[width=0.95\linewidth]{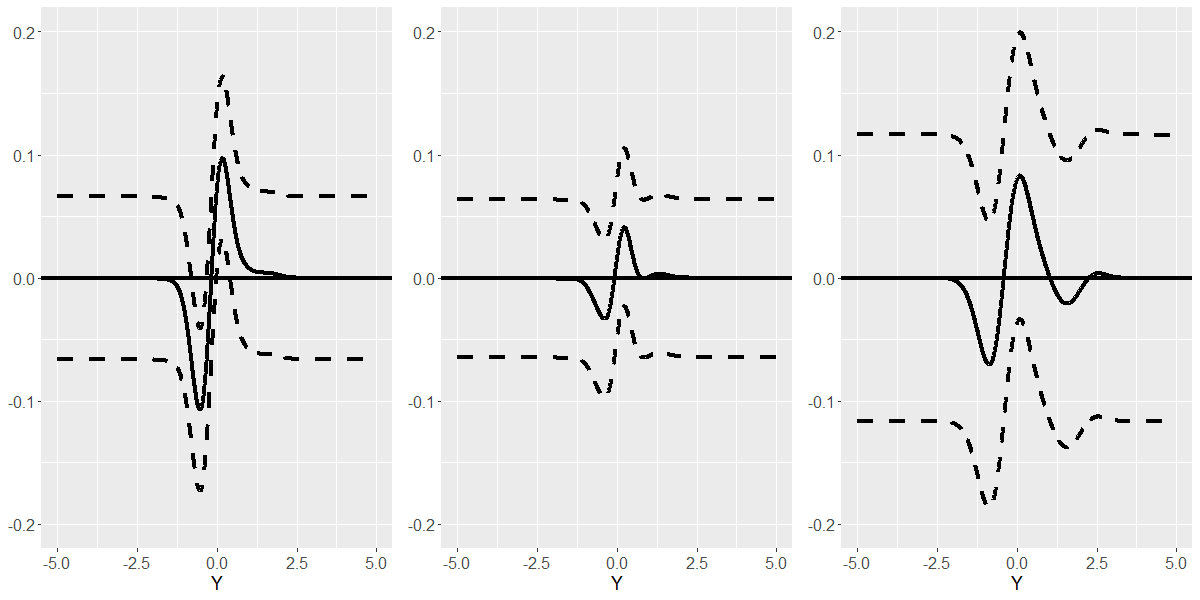}
        \caption{Witness functions with confidence bands for Net Financial Assets (Left), Net Non-401(k) Financial Assets (Middle), and Total Wealth (Right) for a person of age 31, 28'146 income, family size of 5, 12 years of education, married, single earner, homeowner, no defined benefit pension status, and no IRA participation. }
        \label{fig:401KResults1}
    \end{minipage}
    \hfill
    \begin{minipage}{0.48\textwidth}
        \centering
        \includegraphics[width=0.95\linewidth]{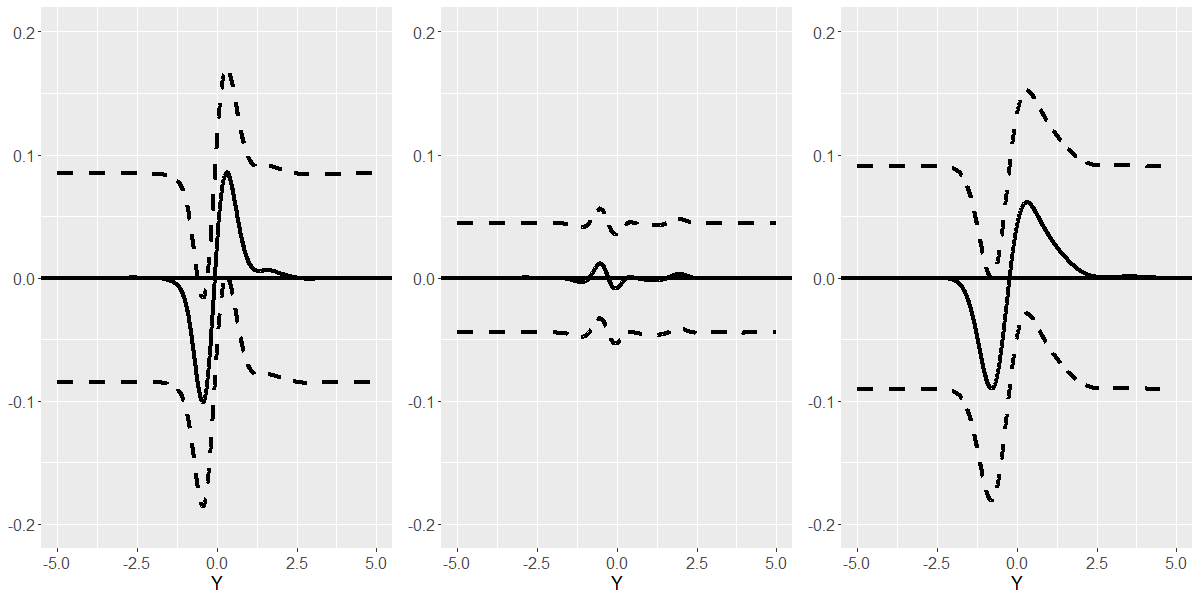}
        \caption{Witness functions with confidence bands for Net Financial Assets (Left), Net Non-401(k) Financial Assets (Middle), and Total Wealth (Right) for a person of age 50, 5'766 income, family size of 1, 14 years of education, unmarried, single earner, no defined benefit pension status, and IRA participation.}
        \label{fig:401KResults2}
    \end{minipage}
\end{figure}

% > pension[testid, c("age", "inc", "fsize", "educ", "marr", "twoearn", "hown", "db", "pira")]
%    age   inc fsize educ marr twoearn hown db pira
% 1   31 28146     5   12    1       0    1  0    0
% 2   52 32634     5   16    0       0    1  0    0
% 3   50 52206     3   11    1       1    1  0    1
% 4   28 45252     4   15    1       1    0  0    0
% 5   42 33126     3   12    0       0    1  1    0
% 6   49 76860     6   15    1       1    1  1    0
% 7   40 57477     4   17    1       1    1  1    0
% 8   58 14637     1   14    0       0    0  0    0
% 9   29  6573     4   12    0       0    0  0    0
% 10  50  5766     1   14    0       0    0  0    1

% \begin{figure}[t]
%     \centering
%     \includegraphics[width=0.8\linewidth]{401KResults1.png}
%     \caption{Witness functions with confidence bands for Net Financial Assets (Left), Net Non-401(k) Financial Assets (Middle), and Total Wealth (Right) for a person of age 50, 5'766 income, family size of 1, 14 years of education, unmarried, single earner, no defined benefit pension status, and IRA participation.}
%     \label{fig:401KResults2}
% \end{figure}

%   marr     twoearn       fsize        educ        pira          db        hown         age         inc 
% 0.005609555 0.025170062 0.045244476 0.068790764 0.105226348 0.146164023 0.176661574 0.340193496 1.413656977

\section{Conclusion}

In this paper, we introduced Causal-DRF, a forest-based method to obtain the conditional kernel treatment effect under the assumption of no confounding and overlap. We discussed how to estimate the CKTE using a single estimated forest fit. On the theoretical side, we derive the consistency and asymptotic normality of the Hilbert space-valued estimate for fixed $\xbf$. A computationally efficient method to approximate the asymptotic sampling distribution was also proposed. The latter was used to build tests of $H_0\colon\PYgXz=\PYgXo$ and confidence bands around $\tau(\xbf)(\ybf)$. We also discussed how the shared-partition, one-fit design differs from using two separate kernel mean embedding estimates, and why it can be beneficial when the two groups share learnable structure.

Our work also motivates several interesting research directions. First, while we obtain elegant asymptotically valid uncertainty quantification for fixed $\xbf$, uniform in $\ybf$, we lack uniform guarantees over $\xbf$. The same limitation holds for causal random forests. As such, when testing $H_0\colon\PYgXz=\PYgXo$ for several test points $\xbf$ one quickly encounters a multiple testing issue. A uniform bound in $\xbf$ would allow for a test of $H_0\colon \Prob_{\Ybf \mid \Xbf}^{0}=\Prob_{\Ybf \mid \Xbf}^{1}$ with asymptotic guarantees. This is already possible, in the sense that rejecting $H_0\colon\PYgXz=\PYgXo$ for a test point $\xbf$ implies that $\Prob_{\Ybf \mid \Xbf}^{0}=\Prob_{\Ybf \mid \Xbf}^{1}$ does not hold. On the other hand, even if $\PYgXz=\PYgXo$ holds, it could be because they are the same or because we chose a test point $\xbf$ where they coincide. Second, since we were mainly interested in $\tau_k(\xbf)$, we did not discuss potential estimators that could be derived from the estimate $\htauk$, such as the CATE, and what guarantees could be derived for this estimator. Finally, we note that the methodology might also be used to study the difference in two conditional distributions more generally, not necessarily with a causal interpretation.

\subsubsection*{Acknowledgements}
This work is part of the DIGPHAT project which was supported by a grant from the French government, managed by the National Research Agency (ANR), under the France 2030 program, with reference ANR-22-PESN-0017. The computational requirements for this work were supported in part by the NYU Langone High Performance Computing (HPC) Core's resources and personnel. JP was supported by SNSF Grant 218343.

\bibliography{bibfile}

\begin{thebibliography}{37}
\providecommand{\natexlab}[1]{#1}
\providecommand{\url}[1]{\texttt{#1}}
\expandafter\ifx\csname urlstyle\endcsname\relax
  \providecommand{\doi}[1]{doi: #1}\else
  \providecommand{\doi}{doi: \begingroup \urlstyle{rm}\Url}\fi

\bibitem[Athey et~al.(2019)Athey, Tibshirani, and Wager]{athey2019generalized}
S.~Athey, J.~Tibshirani, and S.~Wager.
\newblock {Generalized random forests}.
\newblock \emph{The Annals of Statistics}, 47\penalty0 (2):\penalty0
  1148--1178, 2019.

\bibitem[B\'{e}nard et~al.(2024)B\'{e}nard, N\"{a}f, and
  Josse]{DRFVariableimportance}
C.~B\'{e}nard, J.~N\"{a}f, and J.~Josse.
\newblock {MMD}-based variable importance for distributional random forest.
\newblock In \emph{Proceedings of The 27th International Conference on
  Artificial Intelligence and Statistics}, volume 238, pages 1324--1332, 2024.

\bibitem[Benjamin(2003)]{Benjamin401k}
D.~J. Benjamin.
\newblock Does 401(k) eligibility increase saving?: Evidence from propensity
  score subclassification.
\newblock \emph{Journal of Public Economics}, 87\penalty0 (5):\penalty0
  1259--1290, 2003.

\bibitem[Breiman(2001)]{breiman2001random}
L.~Breiman.
\newblock Random forests.
\newblock \emph{Machine learning}, 45\penalty0 (1):\penalty0 5--32, 2001.

\bibitem[{\'C}evid et~al.(2022){\'C}evid, Michel, N\"{a}f, Meinshausen, and
  B\"{u}hlmann]{DRF-paper}
D.~{\'C}evid, L.~Michel, J.~N\"{a}f, N.~Meinshausen, and P.~B\"{u}hlmann.
\newblock Distributional random forests: Heterogeneity adjustment and
  multivariate distributional regression.
\newblock \emph{Journal of Machine Learning Research}, 23\penalty0
  (333):\penalty0 1--79, 2022.

\bibitem[Chen and White(1998)]{HilbertspaceCLTs}
X.~Chen and H.~White.
\newblock Central limit and functional central limit theorems for
  hilbert-valued dependent heterogeneous arrays with applications.
\newblock \emph{Econometric Theory}, 14\penalty0 (2):\penalty0 260--284, 1998.

\bibitem[Chernozhukov and Hansen(2004)]{401KData}
V.~Chernozhukov and C.~Hansen.
\newblock {The Effects of 401(K) Participation on the Wealth Distribution: An
  Instrumental Quantile Regression Analysis}.
\newblock \emph{The Review of Economics and Statistics}, 86\penalty0
  (3):\penalty0 735--751, 2004.

\bibitem[Chernozhukov et~al.(2013)Chernozhukov, Fernández-Val, and
  Melly]{Chernozhukov2013}
V.~Chernozhukov, I.~Fernández-Val, and B.~Melly.
\newblock Inference on counterfactual distributions.
\newblock \emph{Econometrica}, 81\penalty0 (6):\penalty0 2205--2268, 2013.

\bibitem[Chernozhukov et~al.(2020)Chernozhukov, Fernandez-Val, and
  Weidner]{Chernozhukov2020}
V.~Chernozhukov, I.~Fernandez-Val, and M.~Weidner.
\newblock Network and panel quantile effects via distribution regression.
\newblock \emph{Journal of Econometrics}, 2020.

\bibitem[Curth and van~der Schaar(2021)]{AI4}
A.~Curth and M.~van~der Schaar.
\newblock Nonparametric estimation of heterogeneous treatment effects: From
  theory to learning algorithms.
\newblock In \emph{Proceedings of The 24th International Conference on
  Artificial Intelligence and Statistics}, volume 130, pages 1810--1818, 13--15
  Apr 2021.

\bibitem[Dudley(2002)]{dudley}
R.~M. Dudley.
\newblock \emph{Real Analysis and Probability}.
\newblock Cambridge Studies in Advanced Mathematics. Cambridge University
  Press, 2002.

\bibitem[Fawkes et~al.(2022)Fawkes, Hu, Evans, and
  Sejdinovic]{fawkes2022doubly}
J.~Fawkes, R.~Hu, R.~J. Evans, and D.~Sejdinovic.
\newblock Doubly robust kernel statistics for testing distributional treatment
  effects even under one sided overlap.
\newblock \emph{Preprint arXiv:1510.04342}, 2022.

\bibitem[Gonz\'{a}lez-Rodr\'{i}guez and Colubi(2017)]{B3}
G.~Gonz\'{a}lez-Rodr\'{i}guez and A.~Colubi.
\newblock On the consistency of bootstrap methods in separable {H}ilbert
  spaces.
\newblock \emph{Econometrics and Statistics}, 1:\penalty0 118--127, 2017.

\bibitem[Gretton et~al.(2007)Gretton, Borgwardt, Rasch, Sch{\"o}lkopf, and
  Smola]{gretton2007kernel}
A.~Gretton, K.~Borgwardt, M.~Rasch, B.~Sch{\"o}lkopf, and A.~J. Smola.
\newblock A kernel method for the two-sample-problem.
\newblock In \emph{Advances in Neural Information Processing Systems}, pages
  513--520, 2007.

\bibitem[Gretton et~al.(2012)Gretton, Sejdinovic, Strathmann, Balakrishnan,
  Pontil, Fukumizu, and Sriperumbudur]{gretton2012optimal}
A.~Gretton, D.~Sejdinovic, H.~Strathmann, S.~Balakrishnan, M.~Pontil,
  K.~Fukumizu, and B.~K. Sriperumbudur.
\newblock Optimal kernel choice for large-scale two-sample tests.
\newblock In \emph{Advances in Neural Information Processing Systems}, pages
  1205--1213, 2012.

\bibitem[Hohberg et~al.(2020)Hohberg, Pütz, and Kneib]{AI2}
M.~Hohberg, P.~Pütz, and T.~Kneib.
\newblock Treatment effects beyond the mean using distributional regression:
  Methods and guidance.
\newblock \emph{PLOS ONE}, 15\penalty0 (2):\penalty0 1--29, 02 2020.

\bibitem[Holland(1986)]{Holland01121986}
P.~W. Holland.
\newblock Statistics and causal inference.
\newblock \emph{Journal of the American Statistical Association}, 81\penalty0
  (396):\penalty0 945--960, 1986.

\bibitem[Hsing and Eubank(2015)]{hilbertspacebook}
T.~Hsing and R.~Eubank.
\newblock \emph{Theoretical Foundations of Functional Data Analysis, with an
  Introduction to Linear Operators}.
\newblock Wiley Series in Probability and Statistics. Wiley, 2015.

\bibitem[Kallus and Oprescu(2023)]{GRFDistributionalCausalEffects}
N.~Kallus and M.~Oprescu.
\newblock Robust and agnostic learning of conditional distributional treatment
  effects.
\newblock In \emph{Proceedings of The 26th International Conference on
  Artificial Intelligence and Statistics}, volume 206, pages 6037--6060, 2023.

\bibitem[Kosorok(2003)]{B2}
M.~R. Kosorok.
\newblock Bootstraps of sums of independent but not identically distributed
  stochastic processes.
\newblock \emph{Journal of Multivariate Analysis}, 84\penalty0 (2):\penalty0
  299--318, 2003.

\bibitem[Kosorok(2008)]{kosorok2008introduction}
M.~R. Kosorok.
\newblock \emph{Introduction to Empirical Processes and Semiparametric
  Inference}.
\newblock Springer Series in Statistics. Springer New York, 2008.

\bibitem[Martinez~Taboada et~al.(2023)Martinez~Taboada, Ramdas, and
  Kennedy]{martineztaboada2023efficient}
D.~Martinez~Taboada, A.~Ramdas, and E.~Kennedy.
\newblock An efficient doubly-robust test for the kernel treatment effect.
\newblock In \emph{Advances in Neural Information Processing Systems},
  volume~36, pages 59924--59952, 2023.

\bibitem[Muandet et~al.(2021)Muandet, Kanagawa, Saengkyongam, and
  Marukatat]{Counterfactualmeanembeddings}
K.~Muandet, M.~Kanagawa, S.~Saengkyongam, and S.~Marukatat.
\newblock Counterfactual mean embeddings.
\newblock \emph{Journal of Machine Learning Research}, 22\penalty0
  (162):\penalty0 1--71, 2021.

\bibitem[Näf et~al.(2023)Näf, Emmenegger, Bühlmann, and
  Meinshausen]{näf2023confidence}
J.~Näf, C.~Emmenegger, P.~Bühlmann, and N.~Meinshausen.
\newblock Confidence and uncertainty assessment for distributional random
  forests.
\newblock \emph{Journal of Machine Learning Research}, 24\penalty0
  (366):\penalty0 1--77, 2023.

\bibitem[Park et~al.(2021)Park, Shalit, Sch{\"o}lkopf, and
  Muandet]{CATEGeneralization}
J.~Park, U.~Shalit, B.~Sch{\"o}lkopf, and K.~Muandet.
\newblock Conditional distributional treatment effect with kernel conditional
  mean embeddings and {U}-statistic regression.
\newblock In \emph{Proceedings of 38th International Conference on Machine
  Learning (ICML)}, volume 139, pages 8401--8412, 2021.

\bibitem[Pisier(2016)]{pisier_2016}
G.~Pisier.
\newblock \emph{Martingales in Banach Spaces}.
\newblock Cambridge Studies in Advanced Mathematics. Cambridge University
  Press, 2016.

\bibitem[Praestgaard and Wellner(1993)]{B1}
J.~Praestgaard and J.~A. Wellner.
\newblock {Exchangeably Weighted Bootstraps of the General Empirical Process}.
\newblock \emph{The Annals of Probability}, 21\penalty0 (4):\penalty0
  2053--2086, 1993.

\bibitem[Rahimi and Recht(2008)]{rahimi2008random}
A.~Rahimi and B.~Recht.
\newblock Random features for large-scale kernel machines.
\newblock In \emph{Advances in Neural Information Processing Systems}, pages
  1177--1184, 2008.

\bibitem[Rubin(2005)]{RubinPotentialoutcome}
D.~B. Rubin.
\newblock Causal inference using potential outcomes.
\newblock \emph{Journal of the American Statistical Association}, 100\penalty0
  (469):\penalty0 322--331, 2005.

\bibitem[Simon-Gabriel et~al.(2023)Simon-Gabriel, Barp, Sch{\"o}lkopf, and
  Mackey]{simon2020metrizing}
C.-J. Simon-Gabriel, A.~Barp, B.~Sch{\"o}lkopf, and L.~Mackey.
\newblock Metrizing weak convergence with maximum mean discrepancies.
\newblock \emph{Journal of Machine Learning Research}, 24\penalty0
  (184):\penalty0 1--20, 2023.

\bibitem[Sriperumbudur(2016)]{optimalestimationofprobabilitymeasures}
B.~Sriperumbudur.
\newblock {On the optimal estimation of probability measures in weak and strong
  topologies}.
\newblock \emph{Bernoulli}, 22\penalty0 (3):\penalty0 1839--1893, 2016.

\bibitem[Umegaki and Bharucha-Reid(1970)]{UMEGAKI197049}
H.~Umegaki and A.~Bharucha-Reid.
\newblock Banach space-valued random variables and tensor products of banach
  spaces.
\newblock \emph{Journal of Mathematical Analysis and Applications}, 31\penalty0
  (1):\penalty0 49--67, 1970.

\bibitem[Wager and Athey(2017)]{wager2017estimation}
S.~Wager and S.~Athey.
\newblock Estimation and inference of heterogeneous treatment effects using
  random forests.
\newblock \emph{Preprint arXiv:1510.04342}, 2017.

\bibitem[Wager and Athey(2018)]{wager2018estimation}
S.~Wager and S.~Athey.
\newblock Estimation and inference of heterogeneous treatment effects using
  random forests.
\newblock \emph{Journal of the American Statistical Association}, 113\penalty0
  (523):\penalty0 1228--1242, 2018.

\bibitem[Wager and Taylor(2023)]{wager2023causalinference}
S.~Wager and J.~Taylor.
\newblock \emph{Causal Inference: A Statistical Learning Approach}.
\newblock Draft, 2023.
\newblock URL \url{https://web.stanford.edu/~swager/causal_inf_book.pdf}.
\newblock Draft version available online.

\bibitem[Wendland(2004)]{wendland_2004}
H.~Wendland.
\newblock \emph{Scattered Data Approximation}.
\newblock Cambridge Monographs on Applied and Computational Mathematics.
  Cambridge University Press, 2004.
\newblock \doi{10.1017/CBO9780511617539}.

\bibitem[Zhao and Meng(2015)]{zhao2015fastmmd}
J.~Zhao and D.~Meng.
\newblock Fast{MMD}: Ensemble of circular discrepancy for efficient two-sample
  test.
\newblock \emph{Neural computation}, 27\penalty0 (6):\penalty0 1345--1372,
  2015.

\end{thebibliography}
\bibliographystyle{abbrvnat}

%%%%%%%%%%%%%%%%%%%%%%%%%%%%%%%%%%%%%%%%%%%%%%%%%%%%%%%%%%%%
\section*{Checklist}

\begin{enumerate}

  \item For all models and algorithms presented, check if you include:
  \begin{enumerate}
    \item A clear description of the mathematical setting, assumptions, algorithm, and/or model. [Yes] See Section 3 and Appedix C.
    \item An analysis of the properties and complexity (time, space, sample size) of any algorithm. [Yes] Appendix C.
    \item (Optional) Anonymized source code, with specification of all dependencies, including external libraries. [Yes/No/Not Applicable]
  \end{enumerate}

  \item For any theoretical claim, check if you include:
  \begin{enumerate}
    \item Statements of the full set of assumptions of all theoretical results. [Yes] See Section 4. 
    \item Complete proofs of all theoretical results. [Yes] See Appendices D and E. 
    \item Clear explanations of any assumptions. [Yes] See Section 4 and Appendix D.      
  \end{enumerate}

  \item For all figures and tables that present empirical results, check if you include:
  \begin{enumerate}
    \item The code, data, and instructions needed to reproduce the main experimental results (either in the supplemental material or as a URL). [Yes/No/Not Applicable]
    \item All the training details (e.g., data splits, hyperparameters, how they were chosen). [Yes] See Section 5 and Appendices A and B. 
    \item A clear definition of the specific measure or statistics and error bars (e.g., with respect to the random seed after running experiments multiple times). [Yes] See Section 5 and Appendices A and B. 
    \item A description of the computing infrastructure used. (e.g., type of GPUs, internal cluster, or cloud provider). [Not Applicable]
  \end{enumerate}

  \item If you are using existing assets (e.g., code, data, models) or curating/releasing new assets, check if you include:
  \begin{enumerate}
    \item Citations of the creator If your work uses existing assets. [Not Applicable]
    \item The license information of the assets, if applicable. [Not Applicable]
    \item New assets either in the supplemental material or as a URL, if applicable. [Not Applicable]
    \item Information about consent from data providers/curators. [Not Applicable]
    \item Discussion of sensible content if applicable, e.g., personally identifiable information or offensive content. [Not Applicable]
  \end{enumerate}

  \item If you used crowdsourcing or conducted research with human subjects, check if you include:
  \begin{enumerate}
    \item The full text of instructions given to participants and screenshots. [Not Applicable]
    \item Descriptions of potential participant risks, with links to Institutional Review Board (IRB) approvals if applicable. [Not Applicable]
    \item The estimated hourly wage paid to participants and the total amount spent on participant compensation. [Not Applicable]
  \end{enumerate}

\end{enumerate}

\clearpage
\appendix
\thispagestyle{empty}

% Supplementary material: To improve readability, you must use a single-column format for the supplementary material.
\onecolumn
\aistatstitle{Supplementary Materials}
\section{Variables used in the 401(k) eligibility study}\label{sec:401(k)table}
In table~\ref{table:variables}, we summarize the covariates used in the 401(k) eligibility study, investigated in Section~\ref{sec:Experimental}. 
\begin{table}[h]
\centering
\begin{tabular}{lll}
\hline
\textbf{Name} & \textbf{Description} & \textbf{Type} \\
\hline
inc & income & continuous \\
age & age & continuous \\
hown & home ownership & binary \\
db & defined benefit pension status & binary \\
pira & IRA participation & binary \\
educ & years of completed education & continuous \\
fsize & family size & continuous \\
two\_earn & whether dual-earning household & binary \\
marr & marital status & binary \\
\hline
\end{tabular}
\caption{Description of Variables in $\Xbf$. %The variables are ordered from top to bottom according to their importance given by the MMD variable importance in Section \ref{Sec_VarImportance}.
}
\label{table:variables}
\end{table}

\section{Simulation Study}\label{Sec_Simstudy}
Inspired by \citep{wager2018estimation,DRF-paper, näf2023confidence} we define a general simulation framework which we use to define data generating processes that combine conditional treatment effects and confounding.
Let $p = 5$ be the dimension of the baseline covariates. The generic variable $O = (\Xbf, W, Y)$ is drawn according to the following laws:
\begin{align*}
    \mathbf{X} &\sim \mathrm{Unif}(0, 1)^K, \\
    W | \mathbf{X} &\sim \mathrm{Bernoulli}(e(\mathbf{X})), \\
    Y | W, \mathbf{X} &\sim \mathrm{Normal}(2X_3 - 1 + (W - 0.5) t(\Xbf), 1),
\end{align*}
where $e$ is a function defining the probability of treatment and $t$ defines the conditional treatment effect. Next we define four regimes, with the choices of $e$ and $t$ reflecting \cite{näf2023confidence}:
\begin{enumerate}
    \item \label{simulation:nothing} No confounding and no conditional treatment effect: $e(\mathbf{X}) = 0.5$ and $t(\mathbf{X}) = 0$.
    \item \label{simulation:confounding} Confounding: $e(\mathbf{X}) = 0.25(1 + \beta_{2,4}(X_3))$, where $\beta_{a,b}$ is the density of the Beta distribution with parameters $a$ and $b$, and $t(\mathbf{X}) = 0$.
    \item \label{simulation:effect} Effect: $e(\mathbf{X}) = 0.5$ and $t(\mathbf{X}) = \eta(X_1) \cdot \eta(X_2)$, where $\eta(x) = 1 + (1 + \exp(-20(x - 1/3)))^{-1}$.
    \item \label{simulation:both} Confounding and effect: $e(\mathbf{X})$ is defined as in the confounding regime and $t(\mathbf{X})$ is defined as in the effect regime.
\end{enumerate}
In addition, we investigated an alternative outcome model in which there was no shared structure between treatment arms, and the errors were heavy-tailed:
\begin{align}
    \label{eq:heavy-tailed-response}
    Y \mid W,  \mathbf{X} &= 2X_3 - 1 - 0.5 \eta(X_1) \cdot \eta(X_2) + \epsilon, \text{ when } W = 0, \nonumber \\
    Y \mid W, \mathbf{X} &= X_3^2 - 2X_1 \cdot X_2 + 2\epsilon, \text{ when } W = 1, \nonumber\\
    \text{and } \epsilon &\sim \mathrm{StudentT}(3),
\end{align}
where $\mathrm{StudentT}(3)$ denotes Student's $t$ distribution with $3$ degrees of freedom. The alternative outcome model was used to define alternative versions of simulation regimes \ref{simulation:effect} and \ref{simulation:both}.

The simulation datasets were constructed by drawing $n$ iid copies of $O$ from one of the above regimes. We present empirical results based on 500 simulation datasets constructed for each regime and for $n \in \{ 250, 500, 1000, 5000 \}$. Causal DRF was run with $N=2500$ trees with $50$ trees per group ($B=50$). As a comparison, we apply the method described by \citep{näf2023confidence} as a benchmark in which two DRFs are fit separately to treatment and control observations, again each with $2500$ trees and $50$ trees per group. 

Both methods were used to estimate the conditional witness function evaluated at a test point $(0.7, 0.3, 0.5, 0.68, 0.43)^T$. In addition, we draw $8'000$ observations from the conditional distribution $\Ybf \mid \Xbf=\xbf$ to obtain an approximation to the true underlying witness function $\tauk$. The mean absolute error 
\begin{align*}
    \text{MAE}= \frac{1}{n} \sum_{i=1}^{n} |\htauk(\Ybf_i) -  \tauk(\Ybf_i)|
\end{align*}
and the empirical coverage of the estimated 95\% confidence bands were then calculated across all simulations for each sample size. The results are shown in Table~\ref{tab:simulation_results}. For the case of the shared-strcture, light tail outcome model, both Causal-DRF and the benchmark method achieve similar MAE across Regime \ref{simulation:nothing}, although under the confounding regime (Regime \ref{simulation:confounding}), fitting two DRFs appears to be more accurate in terms of MAE. Although this is surprising, the reverse is true as soon as there is an effect, such that in Regimes \ref{simulation:effect} and \ref{simulation:both}, the MAE is up to 30\% lower for Causal-DRF. Crucially, while the empirical coverage of both methods tends to be conservative in scenarios without a treatment effect (Regimes \ref{simulation:nothing} and \ref{simulation:confounding}), in regimes with a treatment effect (Regimes \ref{simulation:effect} and \ref{simulation:both}), Causal-DRF has near-optimal 95\% empirical coverage at all sample sizes considered, while the benchmark method significantly undercovers for sample sizes less than 5000. Thus, while there is a small loss in accuracy for the two cases without effect, Causal-DRF provides accurate coverage for all cases and even for small sample sizes. Figure \ref{fig:simulationIllustration} provides intuition for these results in the case of $n=1000$ and with confounding and treatment effect. It can be seen that the witness function estimated by fitting two DRFs, though also relatively accurate, is further away from the true line in gray overall. In addition, the confidence bands are narrower. This combination of less accurate estimation and narrower bands leads to a higher likelihood of the confidence bands not capturing the true witness function.

For the heavy-tailed case with no shared structure in \eqref{eq:heavy-tailed-response}, DRF somewhat outperforms Causal-DRF, as one might expect. However, the difference is relatively small in terms of MAE and mostly leads to an overcoverage of Causal-DRF, while the coverage for DRF is very close to the nominal level.

The performance of Causal-DRF with regard to the CATE depended on the outcome model. For the settings with a null outcome model, Causal-DRF exhibited performance comparable to Causal-RF. Interestingly, Causal-DRF struggled in the first (Gaussian) outcome model setting for CATE estimation, exhibiting poor coverage and higher MAE than Causal-RF; however, it performed better than Causal-RF for the second (heavy-tailed) outcome model.

\begin{landscape}
\begin{table}[t]
    \centering
    \begin{tabular}{|lrrrr|rrrrrr|}
    \hline
    & \multicolumn{4}{c|}{Witness Function} & \multicolumn{6}{c|}{CATE}  \\
    & \multicolumn{2}{c}{MAE} & \multicolumn{2}{c|}{95\% Coverage} & \multicolumn{3}{c}{MAE} & \multicolumn{3}{c|}{95\% Coverage} \\
    $n$ & DRF & Causal-DRF & DRF & Causal-DRF & DRF & Causal-DRF & Causal-RF & DRF & Causal-DRF & Causal-RF \\
    \hline
    \multicolumn{5}{|l|}{} & \multicolumn{6}{l|}{} \\
    \multicolumn{5}{|l|}{no confounding, no effect (\ref{simulation:nothing})} & \multicolumn{6}{l|}{} \\
    250 & 0.035 & 0.041 & 100.0\% & 100.0\% & 0.134 & 0.147 & 0.125 & 97.1\% & 97.0\% & 94.9\%\\
    500 & 0.030 & 0.033 & 99.4\% & 100.0\% & 0.120 & 0.125 & 0.103 & 97.5\% & 97.5\% & 94.8\%\\
    1000 & 0.027 & 0.029 & 100.0\% & 100.0\% & 0.103 & 0.108 & 0.089 & 97.8\% & 97.8\% & 94.8\%\\
    5000 & 0.019 & 0.022 & 100.0\% & 100.0\% & 0.072 & 0.083 & 0.065 & 98.9\% & 98.9\% & 97.1\%\\
    \multicolumn{5}{|l|}{} & \multicolumn{6}{l|}{} \\
    \multicolumn{5}{|l|}{confounding (\ref{simulation:confounding})} & \multicolumn{6}{l|}{}  \\
    250 & 0.058 & 0.060 & 94.8\% & 99.6\% & 0.298 & 0.224 & 0.136 & 71.3\% & 87.5\% & 94.7\%\\
      500 & 0.035 & 0.048 & 99.6\% & 99.2\% & 0.194 & 0.172 & 0.109 & 85.3\% & 92.3\% & 94.6\%\\
      1000 & 0.027 & 0.039 & 99.8\% & 100.0\% & 0.144 & 0.143 & 0.091 & 91.0\% & 94.3\% & 95.3\%\\
      5000 & 0.020 & 0.027 & 100.0\% & 100.0\% & 0.082 & 0.101 & 0.067 & 97.5\% & 97.7\% & 97.0\%\\
    \multicolumn{5}{|l|}{} & \multicolumn{6}{l|}{} \\
    \multicolumn{5}{|l|}{effect (\ref{simulation:effect})} & \multicolumn{6}{l|}{}  \\
    250 & 0.066 & 0.065 & 78.2\% & 97.0\% & 0.858 & 0.705 & 0.483 & 14.8\% & 27.2\% & 55.3\%\\
    500 & 0.060 & 0.059 & 80.2\% & 96.4\% & 0.726 & 0.643 & 0.307 & 16.9\% & 26.4\% & 74.6\%\\
    1000 & 0.054 & 0.053 & 84.4\% & 97.4\% & 0.594 & 0.590 & 0.214 & 23.4\% & 27.6\% & 82.0\%\\
    5000 & 0.040 & 0.041 & 95.8\% & 96.0\% & 0.270 & 0.484 & 0.118 & 59.9\% & 33.1\% & 90.7\%\\
    \multicolumn{5}{|l|}{} & \multicolumn{6}{l|}{} \\
    \multicolumn{5}{|l|}{confounding and effect (\ref{simulation:both})} & \multicolumn{6}{l|}{} \\
    250 & 0.072 & 0.070 & 77.6\% & 97.4\% & 0.851 & 0.748 & 0.485 & 27.3\% & 31.9\% & 57.6\%\\
     500 & 0.062 & 0.061 & 86.0\% & 96.8\% & 0.709 & 0.688 & 0.315 & 28.7\% & 30.0\% & 74.7\%\\
     1000 & 0.055 & 0.053 & 91.8\% & 96.8\% & 0.583 & 0.641 & 0.220 & 32.6\% & 29.1\% & 82.7\%\\
     5000 & 0.039 & 0.037 & 97.6\% & 99.2\% & 0.279 & 0.529 & 0.124 & 56.8\% & 32.8\% & 90.4\%\\
     \multicolumn{5}{|l|}{} & \multicolumn{6}{l|}{} \\
    \multicolumn{5}{|l|}{effect (\ref{simulation:effect}), heavy-tailed (\ref{eq:heavy-tailed-response})} & \multicolumn{6}{l|}{}  \\
    250 & 0.043 & 0.048 & 97.6\% & 99.2\% & 0.410 & 0.437 & 0.435 & 86.9\% & 92.7\% & 86.9\%\\
      500 & 0.039 & 0.042 & 98.8\% & 99.6\% & 0.325 & 0.376 & 0.370 & 91.1\% & 93.9\% & 86.8\%\\
      1000 & 0.033 & 0.037 & 99.0\% & 100.0\% & 0.270 & 0.331 & 0.321 & 93.4\% & 94.7\% & 87.5\%\\
      5000 & 0.026 & 0.031 & 99.8\% & 99.6\% & 0.203 & 0.263 & 0.246 & 96.5\% & 96.8\% & 91.1\%\\
    \multicolumn{5}{|l|}{} & \multicolumn{6}{l|}{} \\
    \multicolumn{5}{|l|}{confounding and effect (\ref{simulation:both}), heavy-tailed (\ref{eq:heavy-tailed-response})} & \multicolumn{6}{l|}{} \\
    250 & 0.050 & 0.054 & 95.0\% & 99.2\% & 0.440 & 0.452 & 0.438 & 84.3\% & 93.8\% & 88.2\%\\
      500 & 0.041 & 0.044 & 97.2\% & 100.0\% & 0.346 & 0.396 & 0.378 & 88.6\% & 94.8\% & 87.6\%\\
      1000 & 0.034 & 0.038 & 99.0\% & 100.0\% & 0.296 & 0.355 & 0.331 & 91.2\% & 95.3\% & 87.6\%\\
      5000 & 0.027 & 0.032 & 99.2\% & 100.0\% & 0.217 & 0.290 & 0.251 & 95.8\% & 96.7\% & 90.7\%\\
     \hline
    \end{tabular}
    \caption{Simulation study results showing the Mean Absolute Error (MAE) and 95\% empirical coverage for the conditional witness function evaluated at the test point $(0.7, 0.3, 0.5, 0.68, 0.43)^{\top}$.}
    \label{tab:simulation_results}
\end{table}
\end{landscape}

\begin{table}[t]
    \centering
    \begin{tabular}{|ll|rrrrr|}
    \hline
    & & \multicolumn{5}{c|}{Runtime (seconds)} \\
    Simulation Regime & $n$ & \multicolumn{2}{c}{DRF} & \multicolumn{2}{c}{Causal-DRF} & Causal-RF \\
    \hline
    no confounding, no effect (\ref{simulation:nothing}) & 250 & 1.8 & (4.6) & 1.1 & (2.7) & 0.4\\
     & 500 & 3.1 & (4.5) & 2.6 & (3.8) & 0.7\\
     & 1000 & 10.2 & (7.0) & 6.3 & (4.4) & 1.4\\
     & 5000 & 219.4 & (21.7) & 60.4 & (6.0) & 10.1\\
    confounding (\ref{simulation:confounding}) & 250 & 1.9 & (4.5) & 1.1 & (2.5) & 0.4\\
     & 500 & 3.2 & (4.7) & 2.5 & (3.7) & 0.7\\
     & 1000 & 10.4 & (7.4) & 6.2 & (4.4) & 1.4\\
     & 5000 & 217.8 & (23.6) & 57.0 & (6.2) & 9.2\\
    effect (\ref{simulation:effect}) & 1.7 & (4.1) & 1.1 & (2.5) & 0.4\\
     & 500 & 3.3 & (4.9) & 2.6 & (3.9) & 0.7\\
     & 1000 & 10.2 & (7.5) & 6.6 & (4.8) & 1.4\\
     & 5000 & 214.8 & (23.2) & 57.5 & (6.2) & 9.2\\
    confounding and effect (\ref{simulation:both}) & 250 & 1.7 & (4.0) & 1.1 & (2.5) & 0.4\\
     & 500 & 3.4 & (4.7) & 2.5 & (3.6) & 0.7\\
     & 1000 & 10.3 & (7.3) & 6.5 & (4.6) & 1.4\\
     & 5000 & 268.9 & (26.5) & 66.7 & (6.6) & 10.1\\
    effect (\ref{simulation:effect}), heavy tailed \eqref{eq:heavy-tailed-response} & 
    250 & 1.8 & (4.0) & 1.1 & (2.5) & 0.4\\
     & 500 & 3.2 & (4.3) & 2.6 & (3.5) & 0.7\\
     & 1000 & 10.5 & (6.9) & 6.6 & (4.3) & 1.5\\
     & 5000 & 274.4 & (25.6) & 65.9 & (6.2) & 10.7\\
    confounding and effect (\ref{simulation:both}), heavy tailed \eqref{eq:heavy-tailed-response} & 250 & 1.7 & (3.9) & 1.1 & (2.6) & 0.4\\
 & 500 & 3.4 & (4.5) & 2.6 & (3.5) & 0.7\\
 & 1000 & 10.8 & (7.5) & 6.5 & (4.6) & 1.4\\
 & 5000 & 316.1 & (29.7) & 71.1 & (6.7) & 10.7\\

     \hline
    \end{tabular}
    \caption{Simulation study results comparing the average runtime in seconds of each algorithm. The ratios of the runtimes of the DRF and Causal-DRF algorithms compared to that of the Causal-RF algorithm are shown in parentheses.}
    \label{tab:simulation_results2}
\end{table}

\begin{figure}[t]
    \centering
    \includegraphics[width=0.6\linewidth]{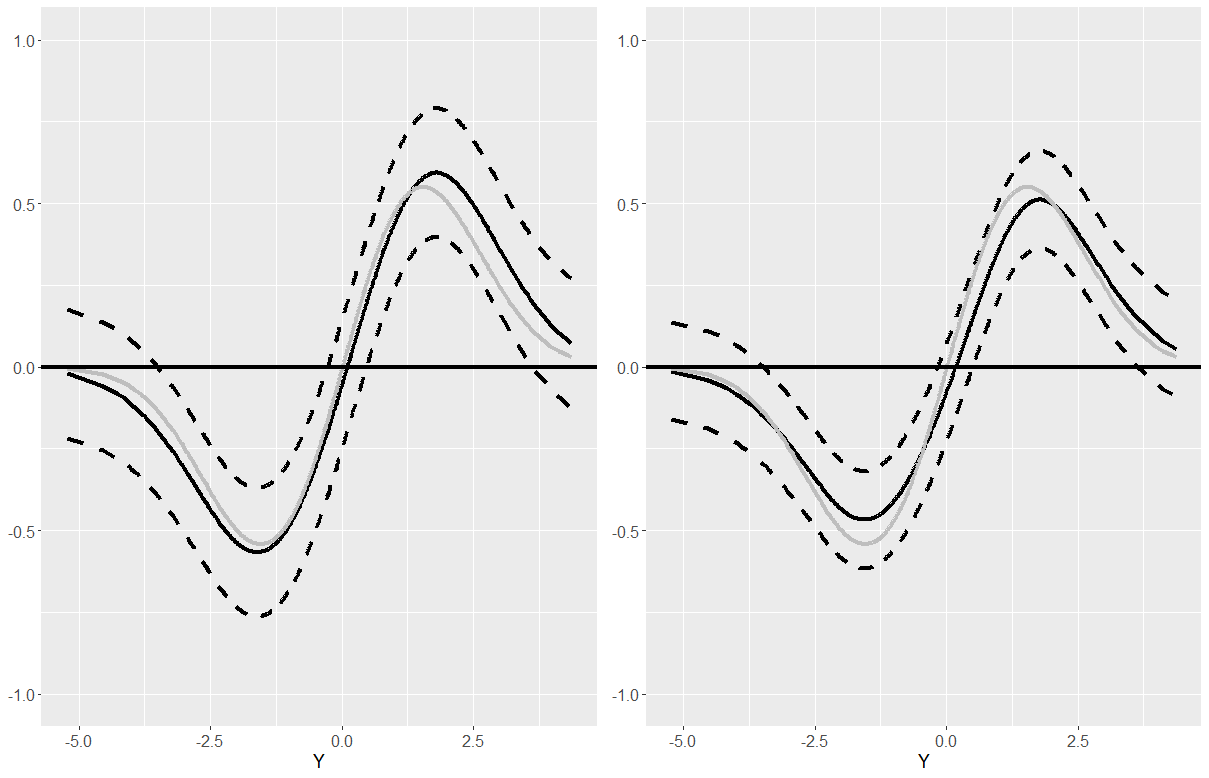}
    \caption{Simulation setting of both confounding and effect for $n=1000$. The black solid lines are the estimated witness functions, while the dashed lines are the estimated 95$\%$ confidence bands. The true witness function is given in gray. Left: Causal-DRF, Right: DRF.}
    \label{fig:simulationIllustration}
\end{figure}

\section{Algorithm and Details on the Splitting Criterion}\label{appendix:splitting}

Following \citep{DRF-paper}, we adopt an approximate kernel $\tilde{k}$ that is a fast random approximation of the MMD statistic \citep{zhao2015fastmmd}. Such an approximation allows for the computational efficiency necessary to allow the number of kernel computations that are required when constructing a large forest. 

By Bochner's theorem (see e.g.\ \citep[Theorem 6.6]{wendland_2004}), any bounded shift-invariant kernel can be written as 
\begin{equation} \label{eq: bochner}
k(\mathbf{u}, \mathbf{v}) = \int_{\R^d} e^{i\boldsymbol{\omega}^T(\mathbf{u}-\mathbf{v})}d\nu(\boldsymbol{\omega}).
\end{equation}
This can be interpreted as a Fourier transform of some measure $\nu$. By randomly sampling the frequency vectors $\boldsymbol{\omega}_1, \ldots, \boldsymbol{\omega}_B$ from normalized $\nu$, we can approximate our kernel $k$ up to a scaling factor by the approximate kernel $\tilde{k}$ via:
$$k(\mathbf{u}, \mathbf{v}) = \int_{\R^d} e^{i\boldsymbol{\omega}^T(\mathbf{u}-\mathbf{v})}d\nu(\boldsymbol{\omega}) \approx \frac{1}{B} \sum_{b=1}^B e^{i\boldsymbol{\omega}_b^T(\mathbf{u}-\mathbf{v})} = \tilde{k}(\mathbf{u}, \mathbf{v}),$$
where $\tilde{k}(\mathbf{u}, \mathbf{v}) = \langle \boldsymbol{\widetilde{\varphi}}(\mathbf{u}), \boldsymbol{\widetilde{\varphi}}(\mathbf{v})\rangle_{\mathbb{C}^B}$ is the kernel function with the feature map given by $$\boldsymbol{\widetilde{\varphi}}(\mathbf{u}) = \frac{1}{\sqrt{B}}\left(\tilde{\varphi}_{\boldsymbol{\omega}_1}(\mathbf{u}), 
\ldots, \tilde{\varphi}_{\boldsymbol{\omega}_B}(\mathbf{u})\right)^T = \frac{1}{\sqrt{B}}\left(e^{i\boldsymbol{\omega}_1^T\mathbf{u}}, 
\ldots, e^{i\boldsymbol{\omega}_B^T\mathbf{u}}\right)^T,$$
a random vector of the Fourier features $\widetilde{\varphi}_{\boldsymbol{\omega}}(\mathbf{u}) = e^{i\boldsymbol{\omega}^T\mathbf{u}} \in \mathbb{C}$ \citep{rahimi2008random}.
Our default choice of kernel $k$ is the Gaussian kernel with bandwidth $\sigma$, as there is a straightforward expression for $\nu$. We can then sample $\boldsymbol{\omega}_1, \ldots, \boldsymbol{\omega}_B \sim N_{d}(\mathbf{0}, \sigma^{-2}I_{d})$. The bandwidth $\sigma$ is chosen as the median pairwise distance between all training responses $\Ybf_1, \ldots, \Ybf_n$(the 'median heuristic'; \citealt{gretton2012optimal}). 

Our splitting criterion seeks to maximize 
\begin{align}\label{splittingcriterionweighted}
    &\frac{|\mathcal{I}_L| |\mathcal{I}_R|}{(|\mathcal{I}_R| + |\mathcal{I}_L|)^2}\Big\lVert
    \sum_{i\in\mathcal{I}_L}   \nu_{i,R} \kYi  
    - \sum_{i\in\mathcal{I}_R}   \nu_{i,R} \kYi    \Big\rVert_{\H}^2= \nonumber \\
    &\frac{|\mathcal{I}_L| |\mathcal{I}_R|}{(|\mathcal{I}_R| + |\mathcal{I}_L|)^2}\Big( \sum_{i \in \mathcal{I}_L, j \in \mathcal{I}_L } v_{i,L} v_{j,L} k(\Ybf_i, \Ybf_j) +     \sum_{i \in \mathcal{I}_R, j \in \mathcal{I}_R } v_{i,R} v_{j,R} k(\Ybf_i, \Ybf_j) \nonumber \\
    &- 2 \sum_{i \in \mathcal{I}_R, j \in \mathcal{I}_L } v_{i,R} v_{j,L} k(\Ybf_i, \Ybf_j)\Big),
\end{align}
which is a weighted version of the MMD statistic used in \citep{DRF-paper}, with weights
\[
\nu_{i,L}= \frac{W_i }{|\mathcal{I}_L, W_i=1|}  - \frac{  1- W_i }{|\mathcal{I}_L, W_i=0|}
\]
and similarly $\nu_{i,R}$. Using the approximate $\tilde{k}$ in \eqref{splittingcriterionweighted} results in:
\begin{align} \label{eq: MMD splitcrit}
\frac{1}{B}\sum_{b=1}^B \frac{|\mathcal{I}_L| |\mathcal{I}_R|}{(|\mathcal{I}_R| + |\mathcal{I}_L|)^2} \left\lvert \sum_{i\in\mathcal{I}_L} \nu_{i,L}\tilde{\varphi}_{\boldsymbol{\omega}_b}(\Ybf_i) -  \sum_{i\in \mathcal{I}_R} \nu_{i,R} \tilde{\varphi}_{\boldsymbol{\omega}_b}(\Ybf_i) \right\rvert^2,
\end{align}
which is the splitting criterion used in the forest.

The overall algorithm is summarized in Algorithm \ref{pseudocodeU}.

\begin{algorithm}[htp]
\caption{Pseudocode for Causal-DRF. The functions \textsc{BuildSubForest} and \textsc{GetWeights} are defined in Algorithm~\ref{pseudocode}.} \label{pseudocodeU}
\begin{algorithmic}[1]

\Procedure{BuildForest2}{set of samples $\mathcal{D} = \{(\mathbf{x}_i,w_i, \mathbf{y}_i)\}_{i=1}^n$, number of trees $N$, number of groups $B$}
 \State  $L \gets$ \textsc{round}(N/B)
    \For{$b=1,\ldots, B$}
        \State $\HS \gets $  Random Subsample from $\mathcal{D}$
        \State $\mathcal{F}_b \gets$ \textsc{BuildSubForest}($\HS$, $L$) \Comment{Build $b$th forest}
    \EndFor
    \State \textbf{return} $\mathcal{F} = \{\mathcal{F}_1, \ldots, \mathcal{F}_B\}$
\EndProcedure

\item[]

\Procedure{GetWeights2}{forests $\mathcal{F}$, test point $\mathbf{x}$} \Comment{Computes the weighting function with uncertainty}
    \For{$b = 1,\ldots,|\mathcal{F}|$}
        \State $w_b \gets $ \textsc{GetWeights}($\mathcal{F}_b$, $\xbf$)
    \EndFor
    \State $w = \frac{1}{B} \sum_{b=1}^B w_b $
    \State \textbf{return} $w, w_1, \ldots, w_B$
\EndProcedure

\end{algorithmic}
\end{algorithm}

\begin{algorithm}[htp]
%\vspace{-0.1cm}
\caption{Pseudocode for building Subforests} \label{pseudocode}
\begin{algorithmic}[1]
\Procedure{BuildSubForest}{set of samples $\mathcal{S} = \{(\mathbf{x}_i, w_i, \mathbf{y}_i)\}_{i=1}^n$, number of trees $N$}
    \For{$i=1,\ldots, N$}
        \State $\mathcal{S}_\text{subsample} \gets$  Subsample $\mathcal{S}$ as in \forestass{5}
        \State $\mathcal{S}_\text{build}, \mathcal{S}_\text{populate}\gets$ Split $\mathcal{S}_\text{subsample}$ as in \ref{forestass1} \Comment{$\mathcal{S}_\text{build}$ to determine tree splits, $\mathcal{S}_\text{populate}$ to populate the leaves}
        \State $\mathcal{N}_i \gets$  initialize root node using $\mathcal{S}_\text{build}$ %\corinne{das und das nächste verstehe ich nicht. Meinst du mit first node den root node, der noch nicht gesplittet wurde? Und wieso brauchts da eine recursion? Kann man da nicht einfach die funktion buildTree mit den daten build aufrufen?} %\textsc{CreateNewTree}($\mathcal{S}_\text{build}$) \Comment{Samples $\mathcal{S}_\text{build}$ used for building the tree}
        \State $\mathcal{T}_i \gets$ \textsc{BuildTree}($\mathcal{N}_i$) \Comment{Start recursion from the root node}
        %\State \textsc{PopulateLeaves}($\mathcal{T}_i, \mathcal{S}_\text{populate})$ \Comment{Samples $\mathcal{S}_\text{populate}$ used for computing $w_\mathbf{x}(\cdot)$}
                \State Populate leaves with $\mathcal{S}_\text{populate}$
    \EndFor
    \State \textbf{return} $\mathcal{F} = \{\mathcal{T}_1, \ldots, \mathcal{T}_N\}$
\EndProcedure

\item[]

%\corinne{Diese procedure verstehe ich leider überhaupt nicht. was ist z.B. StoppingCriterion() Funktion? oder InitializeSplits?}
\Procedure{BuildTree}{current node $\mathcal{N}$} \Comment{Recursively constructs the trees} 
    %\If{\textsc{StoppingCriterion}($\mathcal{N}$)} %\Comment{E.g.\ if only a few samples left}
    %    \State \textbf{return}
    %\EndIf
    \State $\mathcal{S} \gets$ Extract the samples in the node
    \State $\mathcal{I} \gets$ Random set of candidate variables to perform a split on
    \State $\mathcal{C}$ $\gets$  Initialize list
    \Comment{Here we store info %\corinne{what info?} 
    about candidate splits}
    \For{idx $\in \mathcal{I}$, level $l$} \Comment{$l$ iterates over all values of variable $X_{\text{idx}}$}
        \State $\mathcal{S}_L, \mathcal{S}_R \gets$ Split samples based on $(\mathbf{x}_i)_{i \in \text{idx}} \leq l$ 
        \State $v \gets $ Calculate splitting criterion in~\eqref{eq: splitting} using \textsc{($\mathcal{S}_L, \mathcal{S}_R$)} 
        \State Add ($v$, $\mathcal{S}_L, \mathcal{S}_R, \text{idx}, l$) to $\mathcal{C}$ 
    \EndFor
    \State $\mathcal{S}_L, \mathcal{S}_R, \text{idx}, l \gets$ find the best split in $\mathcal{C}$ 
    \State $\mathcal{N}_L \gets$  Create new node with set of samples $\mathcal{S}_L$
    \State $\mathcal{N}_R \gets$ Create new node with set of samples $\mathcal{S}_R$
    \State \textsc{BuildTree}($\mathcal{N}_L$), \textsc{BuildTree}($\mathcal{N}_R$) \Comment{Proceed building recursively}
    \State \textsc{Children}($\mathcal{N}) \gets \mathcal{N}_L, \mathcal{N}_R$ 
    \State \textsc{Split}($\mathcal{N}) \gets \text{idx}, l$ \Comment{Store the split} 
    \State \textbf{return} \textsc{Children}($\mathcal{N}$), \textsc{Split}($\mathcal{N}$) 
\EndProcedure

\item[]

\Procedure{GetWeights}{forest $\mathcal{F}$, test point $\mathbf{x}$} \Comment{Computes the weighting function}
    \State vector of weights $w$ = \textsc{Zeros}($n$) \Comment{$n$ is the training set size}
    \For{$i = 1,\ldots,|\mathcal{F}|$}
        \State $\mathcal{L} \gets $ Get indices of training samples in same leaf as $\mathbf{x}$%\textsc{GetLeafSamples}($\mathcal{T}_i, \mathbf{x}$) \Comment{indices of training samples in same leaf as $\mathbf{x}$}
        \For{$\text{idx} \in \mathcal{L}$}
            \State $w[\text{idx}] = w[\text{idx}]$ + $1/(|\mathcal{L}|\cdot|\mathcal{F}|)$
        \EndFor
    \EndFor
    \State \textbf{return} $w$
\EndProcedure
\end{algorithmic}
\end{algorithm}

\section{Preliminaries}\label{App_Preliminaries}

The following assumptions are placed on the choice of kernel.
\begin{enumerate}[label=(\textbf{K\arabic*})]
    \item\label{kernelass1} $k$ is bounded, i.e., $\sup_{\ybf_1, \ybf_2} k(\ybf_1, \ybf_2) < \infty$.
   \item\label{kernelass2} $(\ybf_1,\ybf_2) \mapsto k(\ybf_1,\ybf_2)$ is (jointly) continuous.
    \item\label{kernelass3} The kernel is translation-invariant, i.e. $k(\ybf_1,\ybf_2)=k_0(\ybf_1-\ybf_2)$ for some function $k_0$.
   \item\label{kernelass4} $k$ is integrally strictly positive definite (denoted by $\int$spd), that is
   \begin{align*}
      \| \Phi(Q_1) - \Phi(Q_2) \|_{\H} = 0 \implies Q_1=Q_2  ,\text{ for all } Q_1, Q_2 \in \mathcal{M}_{b}(\R^d);
   \end{align*}
   see for instance~\citep{simon2020metrizing,optimalestimationofprobabilitymeasures}.
\end{enumerate}
All of these assumptions are met by the Gaussian kernel \begin{align}\label{Gaussiankernel}
    k(\ybf, \cdot) = \exp \left( - \frac{\norm{\ybf - \cdot}}{2\sigma^2} \right).
\end{align}

The forest assumptions with a refinement are discussed below in Section \ref{Sec_Forestass}.

First, we reproduce notation, definitions, and elementary results from \citep{näf2023confidence}.  
Let $\left(\Omega, \mathcal{A}, \mathbb{P}\right)$ denote the underlying probability space.
Throughout, let $(\H, \langle, \cdot, \rangle)$ denote the RKHS associated with the kernel $k$. 
We assume that $k$ is \emph{bounded} and \emph{continuous} in its two arguments. Boundedness of $k$ ensures that $\mu$ is indeed defined on all of $\mathcal{M}_{b}(\R^d)$, and continuity of $k\colon \R^d \times \R^d \to \R$ ensures that $\H$ is \emph{separable}. Thus, we assume throughout that \kernelass{2} holds, such that measurability issues can be avoided. Let us denote by $\xi\colon (\Omega, \mathcal{A}) \to (\H,\mathcal{B}(\H))$ a map from $\Omega$ to $\H$. Separability implies that such a map $\xi$
is measurable if and only if $\langle \xi,f\rangle$ is measurable for all $f \in \H$.  
Moreover, it can easily be checked that $\Phi(P)$ is linear on $\mathcal{M}_{b}(\R^d)$. Separability of $\H$ and $\E[\| \xi \|_{\H}] < \infty$ mean that the integral
\[
\E[\xi] = \int_{\Omega} \xi d \P,
\]
is well defined and that 
\[
F(\E[\xi])=\E[F(\xi)],
\]
for any continuous linear function $F\colon\H \to \R$.\footnote{Here and below, $F(\xi)$ is meant to denote $F(\xi(\omega))$ for all $\omega \in \Omega$. } In particular, $\E[\langle \xi, f  \rangle]= \langle \E[\xi], f \rangle$ for all $f \in \H$. Moreover, for $q \geq 1$, denote by
\begin{align*}
    \mathcal{L}^q(\Omega, \mathcal{A}, \H) &= \{\xi\colon (\Omega, \mathcal{F}) \to (\H,\mathcal{B}(\H)) \text{ measurable, with } \E[\|\xi\|^q] < \infty]  \}\\
    \mathbb{L}^q(\Omega, \mathcal{A}, \H) &= \text{Set of equivalence classes in $\mathcal{L}^q(\Omega, \mathcal{A}, \H)$}\\
    \Var(\xi)&=\E[\|\xi - \E[\xi]\|_{\H}^2]=\E[\|\xi\|_{\H}^2] - \|\E[\xi]\|_{\H}^2, \ \ \xi \in \mathcal{L}^2(\Omega, \mathcal{A}, \H)\\
    \Cov(\xi_1,\xi_2)&= \E[ \langle \xi_1 - \E[\xi_1],\xi_2 - \E[\xi_2] \rangle].%=\E[\langle k(\Ybf_1,\cdot) ,\xi_2  \rangle ] - \langle \E[k(\Ybf_1,\cdot)], \E[\xi_2] \rangle, \ \ k(\Ybf_1,\cdot),\xi_2 \in \mathcal{L}^2(\Omega, \mathcal{A}, \H).
\end{align*}
Furthermore, it is well-known that $(\mathbb{L}^q, \| \cdot \|_{\mathbb{L}^q(\H)})$ is a Banach space with 
\[
\| \xi \|_{\mathbb{L}^q(\H)}= \E[\| \xi \|_{\H}^q]^{1/q}.
\]
This allows us to
also define \emph{conditional} expectations. For a sub $\sigma$-algebra $\mathcal{F} \subset \mathcal{A}$ and an element $\xi \in \mathcal{L}^1(\Omega, \mathcal{A}, \H) $, the conditional expectation $\E[\xi\mid \mathcal{F}]$ is the (a.s.) unique element  such that
\begin{itemize}
    \item[(CE1)]  $\E[\xi\mid \mathcal{F}] \colon (\Omega, \mathcal{F}) \to (\H,\mathcal{B}(\H))$ is measurable and $\E[\xi\mid \mathcal{F}]\in   \mathbb{L}^1(\Omega, \mathcal{F}, \H)$,
    \item[(CE2)] $\E[ \xi \1_{F} ] = \E[ \E[\xi\mid \mathcal{F}]\1_{F} ]$ for all $F \in \mathcal{F}$;
\end{itemize} 
see for instance~\citep{UMEGAKI197049} or~\citep[Chapter 1]{pisier_2016}.
Particularly, condition (CE2) implies that $\E[ \E[\xi\mid \mathcal{F}] ] = \E[ \E[\xi\mid \mathcal{F}]\1_{\Omega} ] = \E[\xi] $ due to $\Omega \in \F$ for any $\sigma$-algebra. It can also be shown that $F(\E[\xi\mid \mathcal{F}])=\E[F(\xi)\mid\mathcal{F}]$ for all linear and continuous $F\colon \H \to \R$ and that $\|\E[\xi\mid \mathcal{F}]  \|_{\H}\leq \E[ \|\xi \|_{\H} \mid \mathcal{F}]$~\citep[Chapter 1]{pisier_2016}. Moreover, it can be shown that
\begin{itemize}
    \item[(CE3)] For $\xi \in \mathbb{L}^2(\Omega, \mathcal{A}, \H)$, $\E[\xi \mid \mathcal{F}]$ is the orthogonal projection onto $\mathbb{L}^2(\Omega, \mathcal{F}, \H)$;
\end{itemize}
see~\citep{UMEGAKI197049}. 
Although the conditional expectation $\E[\xi\mid \mathcal{F}]$, similarly to real-valued conditional expectations, is only defined a.s., we do not explicitly state this in our developments below. Denote
$\E[\xi\mid\Xbf]=\E[\xi\mid \sigma(\Xbf)]$. Proposition 6 and 7 in \citep{DRF-paper} show that this notion is well-defined and establishes further properties of Hilbert space-valued conditional expectations. 

For two functions $f$ and $g$ with $\liminf_{s \to \infty} g(s) > 0$, we denote 
$f(s)= \mathcal{O}(g(s))$ if 
\[
\limsup_{s \to \infty} \frac{|f(s)|}{g(s)} \leq C
\]
for some $C > 0$. If $C=1$, we write $f(s) \precsim g(s)$. For a sequence of random variables $X_n\colon \Omega \to \R$ and a sequence of real numbers $a_n \in (0,+\infty)$, $n \in \N$, we write  $X_n=\O_p(a_n)$ if 
\[
\lim_{M \to \infty} \sup_{n} \P(a_n^{-1} |X_n| > M) = 0,
\]
that is, $X_n$ is bounded in probability. We write $X_n=\o_p(a_n)$ if $a_n^{-1} X_n$ converges to zero in probability. Similarly, for $(S,d)$ a separable metric space, $\Xbf_n\colon (\Omega, \mathcal{A}) \to (S, \mathcal{B}(S))$, $n \in \N$ and $\Xbf\colon (\Omega, \mathcal{A}) \to (S, \mathcal{B}(S))$ measurable, we write $\Xbf_n \stackrel{p}{\to} \Xbf$, if $d(\Xbf_n, \Xbf)=o_{p}(1)$. By separability, every random element $\xi$ with values in $\H$ is tight~\citep[Chapter 7.1]{dudley}. That is, for all $\epsilon > 0$, there is a compact $K_{\epsilon} \subset \H$ such that $\P(\xi \in K_{\epsilon}) \geq 1-\epsilon$. More generally, uniform tightness of a sequence $\xi_n$, $n \in \mathbb{N}$ means that for all $\epsilon > 0$, there is a compact $K_{\epsilon} \subset \H$ such that
\begin{align*}
     \inf_n \P(\xi_{n} \in K_{\epsilon}) \geq 1-\epsilon.
\end{align*}

Finally, let $\Xbf \in \mathcal{L}^2(\Omega, \mathcal{A}, \H_1)$ and $\xi \in \mathcal{L}^2(\Omega, \mathcal{A}, \H_2)$, and assume that $A \subset \Omega$ depends on $\Xbf$, $A=A(\Xbf)$. Thus, for $\Xbf$ fixed to a certain value, $A$ is a fixed set. If $ \P(A \mid \Xbf) > 0$ almost everywhere, we define
\[
\E[\xi \mid  \Xbf, A ] = \frac{\E[\xi \1_A \mid \Xbf ]}{\P( A \mid \Xbf)} \in \mathcal{L}^2(\Omega, \sigma(\Xbf), \H_2).
\]
Then, we have by construction that
\begin{align}
    \E[\xi \1_A \mid \Xbf ] = \E[\xi \mid  \Xbf, A ] \cdot \P( A \mid \Xbf).
\end{align}
For a more compact notation in the following Lemma, let $N=\{1,\ldots,n\}$, and let for $A \subset N$ and $k \leq |A|$, let $C_{k}(A)$ be the set of all subsets of size $k$ drawn from $A$ without replacement, with $C_0=\emptyset$. The following lemma presents a U-statistic expansion that we afterwards apply to an individual tree of our DRF forest.

% \corinne{changed all X's to Z's}
\begin{lemma}[Lemma 9 in~\citep{DRF-paper}]\label{Hdecomposition}

Let $\left(\H_1, \langle\cdot,\cdot\rangle_1\right)$ and $\left(\H_2, \langle\cdot,\cdot\rangle_2\right)$ be two separable Hilbert spaces, and let $\Zbf_1, \ldots, \Zbf_n$ be i.i.d.\ copies of a random element $\Zbf\colon (\Omega, \A) \to (\H_1, \mathcal{B}(\H_1))$. Write $\mathcal{Z}_n=(\Zbf_1, \ldots, \Zbf_n)$, and let $T\colon(\H_1^n, \mathcal{B}(\H_1^n)) \to  (\H_2, \mathcal{B}(\H_2))$ measurable with 
$\E[\| T(\mathcal{Z}_n) \|^2_{\H_2} ] < \infty$. If $T$ is symmetric, there exist functions $T_j$, $j=1,\ldots,n$, such that
\begin{equation}\label{ANOVA}
    T(\mathcal{Z}_n) = \E[T(\mathcal{Z}_n)] + \sum_{i=1}^n T_1(\Zbf_i)  + \sum_{i_1 < i_2} T_2(\Zbf_{i_1}, \Zbf_{i_2}) + \cdots + T_n(\mathcal{Z}_n),
\end{equation}
and it holds that
\begin{align}\label{vardecomp}
    \Var(T(\mathcal{Z}_n))=\sum_{i=1}^n \binom{n}{i} \Var(T_i(\Zbf_1, \ldots, \Zbf_i) )
\end{align}
and
\[
T_{1}(\Zbf_i)=\E[T(\mathcal{Z}_n)\mid\Zbf_i] - \E[T(\mathcal{Z}_n)].
\]

\end{lemma}

\noindent

Recall from the main text that we use $\Zbf_i=\left(Y_i, W_i, \Xbf_i \right)$, for $i=1,\ldots,n$ and $\Zcal_n=\{\Zbf_1, \ldots, \Zbf_n\}$ and similarly $\Zcal_{s_n}=\{\Zbf_1, \ldots, \Zbf_{s_n}\}$.

\subsection{Forest Assumptions and a Refinment}\label{Sec_Forestass}

We first list the forest assumptions taken in \citep{wager2018estimation}:

\begin{enumerate}[label=(\textbf{F\arabic*})]
    \item\label{forestass1} (\textit{Honesty}) The data used for constructing each tree is split into two halves; the first is used for determining the splits and the second for populating the leaves and thus for estimating the response. In addition, $(W_i, \Xbf_i)$ from the second sample may be used to define the splits.  
    \item\label{forestass2} (\textit{Random-split}) At every split point and for all feature dimensions $j=1,\ldots,p$, the probability that the split occurs along the feature $X_j$ is bounded from below by $\pi/p$ for some $\pi > 0$.
    \item\label{forestass3} (\textit{Symmetry}) The (randomized) output of a tree does not depend on the ordering of the training samples.
    \item\label{forestass4} (\textit{$\alpha$-regularity}) 
   For the second sample in \ref{forestass1}, after splitting a parent node, each child node contains at least a fraction $\alpha \leq 0.2$ of the parent's training samples. Moreover, the trees are grown until every leaf contains between $\kappa$ and $2\kappa - 1$ many observations \emph{from each treatment group}, for some fixed tuning parameter $\kappa \in \N$. 
    \item\label{forestass5}(\textit{Data sampling}) 
    To grow a tree, a subsample of size $s_n$ out of the $n$ training data points is sampled. 
    We consider $s_n=n^{\beta}$ with 
    \[
    1 > \beta >  \left( 1 +  \frac{\log( (1-\alpha)^{-1})}{\log(\alpha^{-1})} \frac{\pi}{p} \right)^{-1},
    \]
   where $\alpha$ is chosen in~\ref{forestass4}.
\end{enumerate}
These conditions can be ensured by forest construction, although, as mentioned in \citep{wager2018estimation}, care must be taken for different test points $\xbf$, as \ref{forestass4} must be ensured for each test point.

In the following, denote the indices of the first sample used to build the trees in \ref{forestass1} as $\I_1$, and the second sample used to populate the leaves as $\I_2$. \citep{näf2023confidence} assumed an equivalent of \ref{forestass1}--\ref{forestass5} to derive analogous results for DRF. Unfortunately, a crucial proof in \citep{näf2023confidence}, corresponding to Theorem \ref{thm:varianceassumption} below, only goes through under the assumption that for each $i \in \I_2$, $\Xbf_i$ is independent of $\Lcal(\xbf)$, the leaf $\xbf$ falls into. However, while this is approximately true, since an independent sample is used for each tree to build $\Lcal(\xbf)$, there can be a certain dependence, as $\Lcal(\xbf)$ might be adjusted to ensure \forestass{4}. That is, while a first $\Lcal^0 (\xbf)$ is built completely independently of $(\Xbf_i,W_i)$ from the second sample, $\Lcal(\xbf)$ is then adjusted based on $(\Xbf_i, W_i)$ from the second sample to meet $\forestass{4}$. Although the nature of this dependence is hard to control, we now put some mild assumptions to limit the dependence between $(\Xbf_i, W_i)$ and $\Lcal(\xbf)$ in any forest construction meeting \ref{forestass1}--\ref{forestass5}. In particular, we assume: 
\begin{enumerate}[label=(\textbf{F1}*)]
    \item\label{forestass1starCausalDRF} (\textit{Honesty*}) The second sample is used such that: for all $i\in \I_2$, for all $f$ bounded,
    \begin{align}\label{forestass1starCausalDRF_1}
        \Var( \E[f(\Lcal(\xbf)) \mid \Xbf_i, W_i] ) = \o(s_n^{-(1+\epsilon)}),
    \end{align}  
    and 
    \begin{align}\label{forestass1starCausalDRF_2}
    \E &[ \|\E[ f( (\Xbf_j, W_j, \1\{\Xbf_j \in \Lcal(\xbf)\})_{j \in \I_2, j \neq i} ) \mid \1\{\Xbf_i \notin \Lcal(\xbf)\}, \Lcal(\xbf), \Xbf_i, W_i] - \nonumber \\
    &\E[ f( (\Xbf_j, W_j, \1\{\Xbf_j \in \Lcal(\xbf)\})_{j \in \I_2, j \neq i} )  \mid \1\{\Xbf_i \notin \Lcal(\xbf)\}, \Lcal(\xbf)]\|_{\H}^2]= \o(s_n^{-(1+\epsilon)}),
\end{align}
for some $\epsilon > 0$.
    % for all $i, j \in \I_2$,
    % \begin{align}\label{forestass1starCausalDRF_2old}
    %     (\Xbf_j, W_j) \indep (\Xbf_i, W_i) \mid \Lcal(\xbf), \1\{\Xbf_i \in \Lcal(\xbf)\}.
    % \end{align}
\end{enumerate}

Finally, for Causal-DRF we in addition need an assumption of the behavior of the weight in each class: 
\begin{enumerate}[label=(\textbf{F6})]
    \item\label{forestass6CausalDRF} (\textit{Group Variance}) 
     There exists some $c \in [0, \infty]$, such that for all $i \in \I_2$,
     \begin{align}
         \frac{\Var(\E[S_i \mid \Xbf_i, W_i=1])}{\Var(\E[S_i \mid \Xbf_i, W_i=0])} \to c.
     \end{align}
   % There exists some sequence $\nu_{s_n}$ and constants $c_1,c_2 \in [0, \infty)$, such that for all $i \in \I_2$,
   % \begin{align}
    %    \frac{\Var(\E[S_i \mid \Xbf_1, W_1=1])}{\nu_{s_n}} \to c_1, \frac{\Var(\E[S_i \mid \Xbf_1, W_1=0])}{\nu_{s_n}} \to c_2.
    %\end{align}
\end{enumerate}

\begin{remark}
    Both Assumption \ref{forestass1starCausalDRF} and \ref{forestass6CausalDRF} are needed to prove Theorem \ref{thm:varianceassumption}, which considers the asymptotic ratio of variances of the projected trees:
    \begin{align*}
    \lim_{n \to \infty }\frac{\Var(\langle \Gamma_{n}(\Zbf_1), f \rangle)}{\Var(\Gamma_{n}(\Zbf_1))}
    %\lim_{n \to \infty }\frac{\Var(\langle \Gamma_{n}(\Zbf_1), f \rangle)}{\Var(\Gamma_{n}(\Zbf_1))} = \frac{\Var(\langle k(\Ybf, \cdot) ,f  \rangle|\Xbf=\xbf)}{\Var(k(\Ybf, \cdot)|\Xbf=\xbf)}  =\sigma^2(f) > 0.
\end{align*}
With \eqref{forestass1starCausalDRF_1}, we assume that the adaptation of $\Lcal(\xbf)$ by $(\Xbf_i, W_i)$, $i \in \I_2$ is done in such a way that the dependence is negligible for each individual $(\Xbf_i, W_i)$. Intuitively, one might expect this variance to be $\o(s_n^{-1})$, even if the second sample was freely used to build $\Lcal(\xbf)$, as $(\Xbf_i, W_i)$ is one of $s_n$ points used to obtain $\Lcal(\xbf)$. The additional $\epsilon$ encodes that $\Lcal(\xbf)$ is originally build from an independent set of observations and $(\Xbf_i, W_i)$, $i \in \I_2$ is only used to adapt $\Lcal(\xbf)$ such that \forestass{4} holds. Similarly, \eqref{forestass1starCausalDRF_2} demands that, given $\Lcal(\xbf)$ and $\1\{\Xbf_i \in \Lcal(\xbf)\}$, the dependence of $(\Xbf_i, W_i)$ on all other observations be negligible. Although $(\Xbf_i, W_i)_{i \in \I_2}$ are independent, conditioning on $\Lcal(\xbf)$ and $\1\{\Xbf_i \in \Lcal(\xbf)\}$ will induce (conditional) dependence in general. With \eqref{forestass1starCausalDRF_2} we assume that this dependence decreases sufficiently fast as $n \to \infty$.

As we are only interested in the ratio of variances with identical trees, the denominator and numerator deal with the same dependence of $(\Xbf_i, W_i)$ on $\Lcal(\xbf)$. Thus, it might be possible to further weaken assumption \ref{forestass1starCausalDRF}. However, \ref{forestass1starCausalDRF} is much more plausible for a forest construction satisfying \ref{forestass1}--\ref{forestass5} than the full independence of $\Lcal(\xbf)$ and the second sample implicitly assumed by \citep{näf2023confidence}. Both assumptions (without $W_i$) are also needed for the original DRF, as we discuss in a correction of a key proof argument of \citep{näf2023confidence} in Section \ref{DRFproofcorrection}.

Assumption \ref{forestass6CausalDRF} is the same as Assumption (30) in \citep{näf2023confidence}, ensuring the convergence of the ratio by restricting how the variances can behave in each group. Without it can be shown that there exists constants $0 \leq c_1 \leq c_2 < \infty$, such that
\[
c_1 \leq \lim_{n \to \infty }\frac{\Var(\langle \Gamma_{n}(\Zbf_1), f \rangle)}{\Var(\Gamma_{n}(\Zbf_1))} \leq c_2.
\]
\ref{forestass6CausalDRF} is a sufficient condition that ensures that the ratio does not indefinitely fluctuate between these values as $n \to \infty$. 
\end{remark}

\begin{remark}
    Owing to Assumption \ref{forestass1}, one should write:
    \begin{align*}
        \sum_{i \in \mathcal{I}_2} S_i k(\Ybf_{i}, \cdot)
    \end{align*}
    instead of 
      \begin{align*}
        \sum_{i=1}^{s_n} S_i k(\Ybf_{i}, \cdot).
    \end{align*}
    However, following \citep{wager2018estimation} and \citep{näf2023confidence} we only acknowledge the sample splitting in the tree when it is necessary, such as in the proof of Theorem \ref{thm:varianceassumption}.
\end{remark}

Before we continue with the main proofs, it is instructive to present the corrected argument of a crucial part of the proof of Theorem 5 in \citep{näf2023confidence}. These arguments can then directly be reused to obtain important claims to ultimately proof Theorem \ref{thm:varianceassumption} below.

\subsection{Improving upon the proof of Theorem 5 in [N\"af et al., 2023]}\label{DRFproofcorrection}

We restate here the assumptions of \citep{näf2023confidence}: We refer here to the assumptions \ref{forestass1}--\ref{forestass5} from \citep{näf2023confidence} as (\textbf{DRF-F1})--(\textbf{DRF-F5}). These are essentially the assumptions \ref{forestass1}--\ref{forestass5} shown here when $W_i$ is seen as part of $\Xbf_i$. We repeat here the assumptions on the data taken in \citep{näf2023confidence}:

\begin{enumerate}[label=(\textbf{DRF-D\arabic*})]
    \item\label{dataass1DRF} The covariates $\mathbf{X}_1,\ldots, \mathbf{X}_n$ are independent and identically distributed on $[0,1]^p$ with a density bounded away from $0$ and infinity.
    \item\label{dataass2DRF} The mapping $\xbf \mapsto \mu(\xbf)=\E[ k(\Ybf,\cdot)  \mid \Xbf \myeq \xbf] \in \H $ is Lipschitz.
    \item\label{dataass3DRF} The mapping $\xbf \mapsto \E[ \|k(\Ybf,\cdot) \|_{\H}^2 \mid \Xbf \myeq \xbf] $ is Lipschitz.
    \item\label{dataass4DRF} $\Var(k(\Ybf,\cdot)\mid \Xbf=\xbf) = \E[\| k(\Ybf,\cdot) \|_{\H}^2 \mid \Xbf=\xbf] - \|\E[ k(\Ybf,\cdot)  \mid \Xbf=\xbf]\|_{\H}^2 > 0$.
    \item\label{dataass5DRF} $\E[\left\| k(\Ybf, \cdot) - \mu(\xbf)\right \|_{\H}^{2+\delta}\mid \Xbf=\xbf] \leq M$, for some constants $\delta, M$ uniformly over $\xbf\in [0,1]^d$. 
    \item\label{dataass6DRF} For all $f \in \H \setminus \{0\}$, $\Var( \langle k(\Ybf,\cdot), f \rangle\mid \Xbf=\xbf)= \Var( f(\Ybf)\mid \Xbf=\xbf) > 0$.
    \item\label{dataass7DRF} For all $f \in \H\setminus \{0\}$, $\xbf \mapsto \E[\left| f(\Ybf)\right|^{2}\mid \Xbf=\xbf]$ is Lipschitz. 
\end{enumerate}
Instead of demanding (finite-sample) independence, we ask that the dependence between $\Xbf_i$, $i \in \I_2$ and the leaf $\Lcal(\xbf)$ decays at a certain rate, namely:
\begin{enumerate}[label=(\textbf{DRF-F1}*)]
    \item\label{forestass1starDRF} (\textit{Honesty*}) The second sample is used such that: for all $i\in \I_2$, for all $f$ bounded,
    \begin{align}\label{forestass1starDRF_1}
        \Var( \E[f(\Lcal(\xbf)) \mid \Xbf_i] ) = \o(s_n^{-(1+\epsilon)}).
    \end{align}  
 and
\begin{align}\label{forestass1starDRF_2}
    \E &[ \|\E[ f( (\Xbf_j, \1\{\Xbf_j \in \Lcal(\xbf)\})_{j \in \I_2, j \neq i} ) \mid \1\{\Xbf_i \notin \Lcal(\xbf)\}, \Lcal(\xbf), \Xbf_i] - \nonumber \\
    &\E[ f( (\Xbf_j, \1\{\Xbf_j \in \Lcal(\xbf)\})_{j \in \I_2, j \neq i} )  \mid \1\{\Xbf_i \notin \Lcal(\xbf)\}, \Lcal(\xbf)]\|_{\H}^2]= \o(s_n^{-(1+\epsilon)})
\end{align}
\end{enumerate}
With \eqref{forestass1starDRF_1}, we assume that the adaptation of $\Lcal(\xbf)$ by $\Xbf_i$, $i \in \I_2$ is done in such a way that the dependence is negligible for each individual $\Xbf_i$. Intuitively we might expect this variance to be $\o(s_n^{-1})$, even if the second sample was freely used to build $\Lcal(\xbf)$, as $\Xbf_i$ is one of $s_n$ points used to obtain $\Lcal(\xbf)$. The additional $\epsilon$ encodes that $\Lcal(\xbf)$ is originally build from an independent set of observations and $\Xbf_i$, $i \in \I_2$ is only used to adapt $\Lcal(\xbf)$ such that \forestass{4} holds. This is a convenient and mild assumption that allows for an elegant correction of the proof in \citep{näf2023confidence}, although it might be possible to weaken it.
This assumption holds in particular, if $\Lcal(\xbf)$ is independent of $\Xbf_i$, $i \in \I_2$.

Similarly \eqref{forestass1starDRF_2} is a mild assumption that holds in particular in the following: Assume the sample in $\I_2$ is used to adapt an original $\Lcal^{0}(\xbf)$ build from independent data in $\I_1$, based solely on whether $\Xbf_i \in \Lcal(\xbf)$, as for instance shown in Figure \ref{fig:possibletreebuildingprocess}. Then it holds that,
\begin{align}
        (\Xbf_j, \1\{\Xbf_j \in \Lcal(\xbf)\})_{j \in \I_2, j \neq i} \indep \Xbf_i \mid \Lcal(\xbf), \1\{\Xbf_i \in \Lcal(\xbf)\}.
\end{align}
Thus in this plausible scenario, a condition much stronger than \eqref{forestass1starDRF_2} holds.

\begin{figure}[t]
    \centering
 \begin{tikzpicture}[node distance=2cm, auto]
    % Define styles
    \tikzstyle{block} = [rectangle, draw, text width=5em, text centered, rounded corners, minimum height=2em]
    \tikzstyle{line} = [draw, -latex']

    % Place nodes
    \node [block] (sample1) {Sample $\I_1$};
    \node [block, below of=sample1] (L0) {$\Lcal^{0}(\xbf)$};
    \node [block, left of=L0, node distance=4cm] (X1) {$\Xbf_1$};
    \node [block, right of=L0, node distance=4cm] (X2) {$\Xbf_2$};
    %\node [block, below of=L0, node distance=3cm] (L) {$L(x)$};
    \node [block, below left of=L0, node distance=3cm] (L0X1) {$1\{\Xbf_1 \in \Lcal^{0}(\xbf)\}$};
    \node [block, below right of=L0, node distance=3cm] (L0X2) {$1\{\Xbf_2 \in \Lcal^{0}(\xbf)\}$};
    \node [block, below of=L0, node distance=4.5cm] (L) {$\Lcal(\xbf)$};
    \node [block, below of=L, node distance=2.5cm] (diamL) {$\text{diam}(\Lcal(\xbf))$};
    %\node [block, below of=L, node distance=2.5cm] (diamL) {$\text{diam}(\Lcal(\xbf))$};
    \node [block, below left of=L, node distance=4.5cm] (LX1) {$1\{\Xbf_1 \in \Lcal(\xbf)\}$};
    \node [block, below right of=L, node distance=4.5cm] (LX2) {$1\{\Xbf_2 \in \Lcal(\xbf)\}$};

    % Draw edges
    \path [line] (sample1) -- (L0);
    \path [line] (X1) -- (L0X1);
    \path [line] (X2) -- (L0X2);
    \path [line] (L0) -- (L0X1);
    \path [line] (L0) -- (L0X2);
    \path [line] (L) -- (LX1);
    \path [line] (L) -- (LX2);
    \path [line] (L) -- (diamL);
    \path [line] (L0X1) -- (L);
    \path [line] (L0X2) -- (L);
    \path [line] (X1) -- (LX1);
    \path [line] (X2) -- (LX2);
\end{tikzpicture}

    \caption{Potential Tree building approach for $n=2$ that satisfies \eqref{forestass1starDRF_2}. Crucially, $\Lcal(\xbf)$ is an adapted version of $\Lcal^{0}(\xbf)$ based only on whether $\Xbf_i$, $i \in \I_2$ is included in $\Lcal^{0}(\xbf)$.}
    \label{fig:possibletreebuildingprocess}
\end{figure}
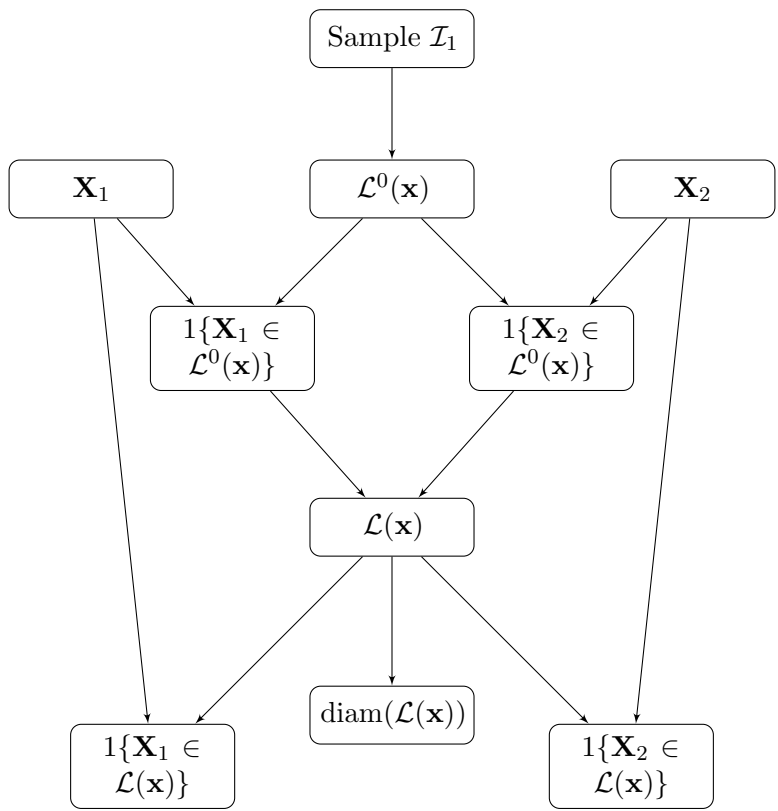

To establish the setting of \citep{näf2023confidence}, define in the following the number of data points belonging to the same leaf as $\xbf$ as $N_{\xbf}=|\{j\colon \Xbf_j \in \Lcal(\xbf)  \}|$ and let 
\begin{align}\label{Sdef}
    S_i=\frac{\1\{ \Xbf_i \in \Lcal(\xbf) \}}{N_{\xbf}},
\end{align}
be the weight associated with each observation $i$ in a tree $T(\mathcal{Z}_{s_n})$ with $\mathcal{E}$ integrated out, i.e.  $T(\mathcal{Z}_{s_n})=\E[T(\xbf, \mathcal{E}; \Zcal_{s_n})$, such that
\[
 T(\mathcal{Z}_{s_n})=\sum_{i=1}^{s_n} S_i k(\mathbf{Y}_i, \cdot).
\]
We will make use the following property of the $S_i$:
\begin{align}\label{eq:expectS1}
    1=\E\left[\sum_{i=1}^{s_{n}} S_i \right] = \sum_{i=1}^{s_{n}}\E[ S_i] = s_n \E[S_1].
\end{align}
In particular,
\begin{align}\label{eq:varDecompSprime2}
\Var(\E[S_1|\Xbf_1])\leq \E[\E[S_1|\Xbf_1]^2] \leq \E[\E[S_1|\Xbf_1]]=\E[S_1]=\O(s_n^{-1})
\end{align}
We define
\begin{align}
    T'(\Zcal_{s_n})&=T(\Zcal_{s_n})\1\{\text{diam}(\Lcal(\xbf))\leq s_n^{-w}\}, \label{Tdashdef}\\
    S_i'&=S_i\1\{\text{diam}(\Lcal(\xbf))\leq s_n^{-w}\}, \text{ where } w=\frac{1}{2} \frac{\pi}{p} \frac{\log \left((1-\alpha)^{-1} \right)}{\log(\alpha^{-1})} \label{wdef},
\end{align}
so that
$T'(\Zcal_{s_n})=\sum_{i=1}^{s_n} S_i' k(\Ybf_i, \cdot)$. We now formulate the corrected version of Theorem 5 in \citep{näf2023confidence} with \ref{forestass1starDRF} added:

\begin{theorem}
    Assume conditions~(\textbf{DRF-F1})--(\textbf{DRF-F5}),~\ref{forestass1starDRF},~\ref{dataass1DRF}--\ref{dataass7DRF},~\ref{kernelass1}, and~\ref{kernelass2} hold. Then, for all $f \in \H\setminus\{0\}$, we have
\begin{align}\label{varianceassumption_DRF}
    \lim_{n \to \infty }\frac{\Var(\langle T_{n}(\Zbf_1), f \rangle)}{\Var(T_{n}(\Zbf_1))} = \frac{\Var(\langle k(\Ybf, \cdot) ,f  \rangle|\Xbf=\xbf)}{\Var(k(\Ybf, \cdot)|\Xbf=\xbf)}  =\sigma^2(f) > 0.
\end{align}
\end{theorem}

\begin{proof}
%%%%%%%%%%%%%Importantnotes donotdelete######################
%- I am a bit sloppy with the \epsilon here in \o(s_n^{-(1+\epsilon)}). It
% always just means a number > 0, but can change in the context of the proof.
%%%%%%%%%%%%%Importantnotes donotdelete######################

We only consider and correct the proof of the crucial Claim (90) in \citep{näf2023confidence}, that is we prove:
\begin{claim}
\begin{align}\label{A_Varboundnew}
    \Var(\E[ \langle T'(\Zcal_{s_n}), f \rangle \mid \Xbf_1 ]) = \o(s_n^{-(1+\epsilon)}) \nonumber\\
\Var( \E[T'(\Zcal_{s_n})| \Xbf_1 ])= \o(s_n^{-(1+\epsilon)}).
\end{align}

\end{claim}

\begin{claimproof}
First, due to honesty (\textbf{DRF-F1}), we have
\begin{align*}
   \E[ \langle T'(\Zcal_{s_n}),f \rangle \mid \Xbf_1 ] &=\E[S_1'\mid \Xbf_1]  \E[ \langle k(\Ybf_1, \cdot),f \rangle\mid  \Xbf_1] + \sum_{i=2}^{s_n} \E[S_i' \langle k(\Ybf_i, \cdot),f \rangle\mid \Xbf_1]
\end{align*}
and
\begin{align*}
  \E[ T'(\Zcal_{s_n}) \mid \Xbf_1 ] &=\E[S_1'\mid \Xbf_1]  \E[ k(\Ybf_1, \cdot)\mid \Xbf_1] + \sum_{i=2}^{s_n} \E[S_i' k(\Ybf_i, \cdot)\mid \Xbf_1].
\end{align*}
Subsequently, we consider the variance of the two terms and their covariance individually. First, we study the variance of the first terms.
\begin{claim}
For some $\epsilon > 0$,
\begin{align}\label{A_variance1term_1}
    \Var(\E[S_1'|\Xbf_1]  \E[ k(\Ybf_1, \cdot)|\Xbf_1])&= \Var(\E[S_1'|\Xbf_1]) \|\E[k(\Ybf_1, \cdot)|\Xbf=\xbf]\|_{\H}^2 + \o(s_n^{-(1+\epsilon)})
\end{align}
and
\begin{align}\label{A_variance1term_2}
        \Var(\E[S_1'|\Xbf_1]  \E[ \langle k(\Ybf_1, \cdot),f \rangle|\Xbf_1]) &= \Var(\E[S_1'|\Xbf_1]) \E[\langle k(\Ybf_1, \cdot), f \rangle|\Xbf=\xbf]^2 + \o(s_n^{-(1+\epsilon)}).
\end{align}
\end{claim}

\begin{claimproof}

We only show~\eqref{A_variance1term_1} because~\eqref{A_variance1term_2} follows analogously. We have
\begin{align}
    \Var(\E[S_1' |\Xbf_1]   \E[ k(\Ybf_1, \cdot)&|\Xbf_1]) =\Var(\E[S_1' |\Xbf_1] (  \E[ k(\Ybf_1, \cdot)|\Xbf_1] - \mu(\xbf)+ \mu(\xbf) )) \nonumber \\
    &=  \Var(\E[S_1' |\Xbf_1] \mu(\xbf)     ) + \Var(\E[S_1' |\Xbf_1]( \E[ k(\Ybf_1, \cdot)|\Xbf_1] - \mu(\xbf) )) \nonumber \\
    &\quad+ \Cov\left(\E[S_1' |\Xbf_1] \mu(\xbf), \E[S_1' |\Xbf_1]( \E[ k(\Ybf_1, \cdot)|\Xbf_1] - \mu(\xbf) ) \right) . 
\end{align}
Because 
\begin{align*}
    \Var(\E[S_1' |\Xbf_1] \mu(\xbf)     )
    =&\E[ \|(\E[S_1' |\Xbf_1] - \E[S_1']) \mu(\xbf)\|_{\H}^2 ]\\
    =&\E[ (\E[S_1' |\Xbf_1] - \E[S_1'])^2] \|\mu(\xbf)\|_{\H}^2\\
    =&\Var(\E[S_1' |\Xbf_1])\|\mu(\xbf)\|_{\H}^2,
\end{align*}
it follows that 
\begin{align*}
    \Var(\E[S_1' |\Xbf_1]   \E[ k(\Ybf_1, \cdot)|\Xbf_1]) &= \Var(\E[S_1' |\Xbf_1])\|\mu(\xbf)\|_{\H}^2\\
    &\quad+ \Var(\E[S_1' |\Xbf_1]( \E[ k(\Ybf_1, \cdot)|\Xbf_1] - \mu(\xbf) )) \nonumber \\
    &\quad+ \Cov\left(\E[S_1' |\Xbf_1] \mu(\xbf), \E[S_1' |\Xbf_1]( \E[ k(\Ybf_1, \cdot)|\Xbf_1] - \mu(\xbf) ) \right).
\end{align*}
Because $\E[S_1' |\Xbf_1]$ maps into $\R_{\geq 0}$, we have
\begin{align*}
    \Var(\E[S_1' |\Xbf_1]( \E[ k(\Ybf_1, \cdot)|\Xbf_1]- \mu(\xbf))) & \leq \E[ \|\E[S_1' |\Xbf_1]( \E[ k(\Ybf_1, \cdot)|\Xbf_1]) - \mu(\xbf))   \|_{\H}^2 ] \\
    &= \E[\E[S_1' |\Xbf_1]^2 \|( \E[ k(\Ybf_1, \cdot)|\Xbf_1]- \mu(\xbf)    ) \|_{\H}^2 ]\\
    & \leq \E[\E[S_1'^2 |\Xbf_1] C^2\|\Xbf_1 -\xbf \|_{\R^p}^2 ]\\
    & \leq \E[S_1'^2  ] C^2 s_n^{-2w},
\end{align*}
where the last step followed because $\E[S_1'^2 |\Xbf_1]=0$, for $\|\Xbf_1 -x \|_{\R^p} > s_n^{-w}$ by definition of $S_1'=S_1\1\{ \text{diam}(\Lcal(\xbf))\leq s_n^{-w}\} $. Since $\E[S_1'^2  ]\leq \E[S_1'] \leq \E[S_1]=\O(s_n^{-1})$, we have 
\begin{align*}
      \Var(\E[S_1' |\Xbf_1]( \E[ k(\Ybf_1, \cdot)|\Xbf_1]- \mu(\xbf)))=\O(s_n^{-(1+2w)}).
\end{align*}
Finally, we infer from the Cauchy-Schwarz inequality,
\begin{align*}
    &\left| \Cov\left(\E[S_1' |\Xbf_1] \mu(\xbf), \E[S_1' |\Xbf_1]( \E[ k(\Ybf_1, \cdot)|\Xbf_1] - \mu(\xbf) ) \right)\right| \\
    \leq& \sqrt{\Var(\E[S_1' |\Xbf_1] \mu(\xbf)     )} \sqrt{\Var(\E[S_1' |\Xbf_1]( \E[ k(\Ybf_1, \cdot)|\Xbf_1]- \mu(\xbf) )) }\\
    =& \O(s_n^{-(1+w)}), 
\end{align*}
due to
\begin{align}\label{A_whatweneed3}
    \Var(\E[S_1' |\Xbf_1] \mu(\xbf)     )&= \Var(\E[S_1' |\Xbf_1]) \cdot \O(1)= \O(s_n^{-1}).
\end{align}
Thus, choosing $0< \epsilon < w$, Claim~\eqref{A_variance1term_1} holds.
\end{claimproof}

Before we continue proving the theorem, we note that, due to honesty (\textbf{DRF-F1}), we have
\begin{align}\label{honestchange}
 \sum_{i=2}^{s_n} \E[S_i' k(\Ybf_i, \cdot)\mid \Xbf_1] &= \sum_{i=2}^{s_n} \E[ \E[S_i' k(\Ybf_i, \cdot) \mid  \Xbf_i, \Xbf_1]  \mid \Xbf_1] \nonumber \\
 &=\sum_{i=2}^{s_n} \E[ \E[S_i' \mid  \Xbf_i, \Xbf_1] \E[ k(\Ybf_i, \cdot) \mid  \Xbf_i, \Xbf_1]  \mid \Xbf_1] \nonumber \\
  &=\sum_{i=2}^{s_n} \E[ \E[S_i' \E[ k(\Ybf_i, \cdot) \mid  \Xbf_i]   \mid  \Xbf_i, \Xbf_1] \mid \Xbf_1]\nonumber \\
  &=\sum_{i=2}^{s_n} \E[ S_i' \E[  k(\Ybf_i, \cdot) \mid  \Xbf_i]  \mid \Xbf_1].
\end{align}

Now, we consider the variance of the sum in~\eqref{honestchange}:

\begin{claim}
For some $\epsilon > 0$,
\begin{align}\label{A_variance2term_1}
    \Var\left( \sum_{i=2}^{s_n} \E[S_i' k(\Ybf_i, \cdot)\mid \Xbf_1] \right)&= \Var(\E[S_1'|\Xbf_1]) \|\E[k(\Ybf_1, \cdot)|\Xbf=\xbf]\|_{\H}^2 + \o(s_n^{-(1+\epsilon)})
\end{align}
and
\begin{align}\label{A_variance2term_2}
        \Var\left( \sum_{i=2}^{s_n} \E[S_i' \langle k(\Ybf_i, \cdot),f \rangle\mid \Xbf_1] \right) &= \Var(\E[S_1'|\Xbf_1]) \E[\langle k(\Ybf_1, \cdot), f \rangle|\Xbf=\xbf]^2 + \o(s_n^{-(1+\epsilon)}).
\end{align}

\end{claim}

\begin{claimproof}

We only show~\eqref{A_variance2term_1}, because~\eqref{A_variance2term_2} follows analogously. First we note that, using the definition of $S_i'$, it holds that  
\begin{align*}
\sum_{i=1}^{s_n} \E[S_i'\mid \Xbf_1]&=\E[\sum_{i=1}^{s_n} S_i' \mid \Xbf_1] \nonumber \\
&=\P\left(\text{diam}(\Lcal(\xbf))\leq s_n^{-w} \mid \Xbf_1\right).
%&=\P\left(\text{diam}(\Lcal(\xbf))\leq s_n^{-w}\right),
\end{align*}
It follows that
\begin{align}\label{A_Sdproperty}
   \sum_{i=2}^{s_n} \E[S_i'\mid \Xbf_1]=\P\left(\text{diam}(\Lcal(\xbf))\leq s_n^{-w} \mid \Xbf_1\right)-\E[S_1' \mid \Xbf_1]
\end{align}

By~\eqref{A_Sdproperty} 
\begin{align*}
    \Var\left(\mu(\xbf) \sum_{i=2}^{s_n} \E[S_i' \mid \Xbf_1]  \right)&= \|\mu(\xbf)\|_{\H}^2\Var\left( \P\left(\text{diam}(\Lcal(\xbf))\leq s_n^{-w} \mid \Xbf_1\right)-\E[S_1' \mid \Xbf_1]\right)\\
    &=\|\mu(\xbf)\|_{\H}^2\Var(\E[S_1' \mid \Xbf_1]) \\
    &\quad+ \|\mu(\xbf)\|_{\H}^2\Var( \P\left(\text{diam}(\Lcal(\xbf))\leq s_n^{-w} \mid \Xbf_1\right)) \\
    &\quad- 2 \|\mu(\xbf)\|_{\H}^2\Cov(\P\left(\text{diam}(\Lcal(\xbf))\leq s_n^{-w} \mid \Xbf_1\right), \E[S_1' \mid \Xbf_1])
\end{align*}
Using Assumption \ref{forestass1starDRF}, it holds that
\begin{align*}
    \Var( \P\left(\text{diam}(\Lcal(\xbf))> s_n^{-w} \mid \Xbf_1\right))=\o(s_n^{-(1+\epsilon)}),
\end{align*}
where $\text{diam}(\Lcal(\xbf))$ is bounded because $[0,1]^p$ is a bounded set. Combining this with the fact that $\Var(\E[S_1' \mid \Xbf_1])=\O(s_n^{-1})$, we get
\begin{align*}
    &\Cov(\P\left(\text{diam}(\Lcal(\xbf))\leq s_n^{-w} \mid \Xbf_1\right), \E[S_1' \mid \Xbf_1])  \\
    &\leq(\Var(\P\left(\text{diam}(\Lcal(\xbf))\leq s_n^{-w} \mid \Xbf_1\right)))^{1/2} (\Var(\E[S_1' \mid \Xbf_1]))^{1/2}\\
    &=\o(s_n^{-(1+\epsilon/2)}).
\end{align*}
Thus,
\begin{align}\label{ConnectionbetweenS1andSi}
    \Var\left(\mu(\xbf) \sum_{i=2}^{s_n} \E[S_i' \mid \Xbf_1]  \right)
    &= \Var(\E[S_1'|\Xbf_1]) \|\mu(\xbf)\|_{\H}^2 + \o(s_n^{-(1+\epsilon)}).
\end{align}
Thus, we need to show that
\begin{claim}
    For some $\epsilon > 0$,
    \begin{align}\label{ConnectionbetweenSkandS}
     \Var\left( \sum_{i=2}^{s_n} \E[S_i' k(\Ybf_i, \cdot)\mid \Xbf_1] \right)=\Var\left(\mu(\xbf) \sum_{i=2}^{s_n} \E[S_i' \mid \Xbf_1]  \right)  + \o(s_n^{-(1+\epsilon)}),
\end{align}
\end{claim}

\begin{claimproof}
Claim \eqref{ConnectionbetweenSkandS} is implied by
\begin{align}\label{A_whatweneed}
    \Var\left( \sum_{i=2}^{s_n} \E[S_i' k(\Ybf_i, \cdot)\mid \Xbf_1] -  \mu(\xbf) \sum_{i=2}^{s_n} \E[S_i' \mid \Xbf_1]  \right) =\O(s_n^{-(1+2w)}),
\end{align}
\begin{align}\label{A_whatweneed2}
    \Var\left( \mu(\xbf) \sum_{i=2}^{s_n} \E[S_i' \mid \Xbf_1]  \right) =\mathcal{O}(s_n^{-1}),
\end{align}
Indeed it holds that, 
\begin{align*}
  &\Var\left( \sum_{i=2}^{s_n} \E[S_i' k(\Ybf_i, \cdot)\mid \Xbf_1] \right)  \\
  &=\Var\left( \sum_{i=2}^{s_n} \E[S_i' k(\Ybf_i, \cdot)\mid \Xbf_1] -  \mu(\xbf) \sum_{i=2}^{s_n} \E[S_i' \mid \Xbf_1] +\mu(\xbf) \sum_{i=2}^{s_n} \E[S_i' \mid \Xbf_1] \right)\\
  &= \Var(\mu(\xbf) \sum_{i=2}^{s_n} \E[S_i' \mid \Xbf_1] )+ \Var(  \sum_{i=2}^{s_n} \E[S_i' k(\Ybf_i, \cdot)\mid \Xbf_1] -  \mu(\xbf) \sum_{i=2}^{s_n} \E[S_i' \mid \Xbf_1]) \\
  &+ 2 \Cov \left( \sum_{i=2}^{s_n} \E[S_i' k(\Ybf_i, \cdot)\mid \Xbf_1] -  \mu(\xbf) \sum_{i=2}^{s_n} \E[S_i' \mid \Xbf_1],\mu(\xbf) \sum_{i=2}^{s_n} \E[S_i' \mid \Xbf_1] \right).
\end{align*}
Thus \eqref{A_whatweneed} and \eqref{A_whatweneed2}, combined with the fact that
\begin{align*}
    &\Cov \left( \sum_{i=2}^{s_n} \E[S_i' k(\Ybf_i, \cdot)\mid \Xbf_1] -  \mu(\xbf) \sum_{i=2}^{s_n} \E[S_i' \mid \Xbf_1],\mu(\xbf) \sum_{i=2}^{s_n} \E[S_i' \mid \Xbf_1] \right)\leq \\
    &\left( \Var(\sum_{i=2}^{s_n} \E[S_i' k(\Ybf_i, \cdot)\mid \Xbf_1] -  \mu(\xbf) \sum_{i=2}^{s_n} \E[S_i' \mid \Xbf_1]) \right)^{1/2} \left(\Var \left( \mu(\xbf) \sum_{i=2}^{s_n} \E[S_i' \mid \Xbf_1] \right)\right)^{1/2},
\end{align*}
gives \eqref{ConnectionbetweenSkandS}. Condition \eqref{A_whatweneed2} is an immediate consequence of \eqref{ConnectionbetweenS1andSi} and the fact that $\Var(\E[S_1'|\Xbf_1])=\O(s_n^{-1})$. Subsequently, we establish~\eqref{A_whatweneed}. 
Now, with~\eqref{honestchange}, we have 
\begin{align}\label{A_whatweneed22}
      \Var&\left( \sum_{i=2}^{s_n} \E[S_i' k(\Ybf_i, \cdot)\mid \Xbf_1] -  \mu(\xbf) \sum_{i=2}^{s_n} \E[S_i' \mid \Xbf_1]  \right)\nonumber\\
      &=\Var\left( \sum_{i=2}^{s_n} \E[ S_i' \E[  k(\Ybf_i, \cdot) \mid  \Xbf_i]  \mid \Xbf_1] -   \E[S_i' \mu(\xbf) \mid \Xbf_1]  \right) \nonumber \\
      &=\Var\left( \sum_{i=2}^{s_n} \E[ S_i'( \E[  k(\Ybf_i, \cdot) \mid  \Xbf_i] - \mu(\xbf) )   \mid \Xbf_1]  \right)\nonumber \\
      &=\Var\left( \sum_{i=2}^{s_n} \E[ S_i'\Delta(\Xbf_i)  \mid \Xbf_1]  \right),
\end{align}
with $\Delta(\Xbf_i)=\E[  k(\Ybf_i, \cdot) \mid  \Xbf_i] - \mu(\xbf) $. 
%%%%%%%%%%%%%%%%%%%%%%%%%%%%%%%%%%%%%%%%%%%%%%%%%%%%%%%%%%%%%%
\begin{claim}
\begin{align}\label{crucialclaimforsumSi}
        \E[\|\E[ \sum_{i=2}^{s_n} S_i'\Delta(\Xbf_i)  \mid \Lcal(\xbf), \Xbf_1] - \E[\sum_{i=2}^{s_n} S_i'\Delta(\Xbf_i)  \mid \Lcal(\xbf)]\|_{\H}^2 ]= \o(s_n^{-(1+\epsilon)}).
\end{align}
\end{claim}
\begin{claimproof}
% %%%%%%%%%%%%%%%%%%%%%%% old %%%%%%%%%%%%%%%%%
%  First, we note that, since $\1\{\Xbf_1 \notin \Lcal(\xbf)\}$ is a deterministic function of $\Xbf_1$, $\Lcal(\xbf)$, we have that
%     \begin{align*}
%         \E[ S_i'\Delta(\Xbf_i)  \mid \Lcal(\xbf), \Xbf_1] &= \E[ S_i'\Delta(\Xbf_i)  \mid \1\{\Xbf_1 \notin \Lcal(\xbf)\}, \Lcal(\xbf), \Xbf_1]\\
%         &=\E[ S_i'\Delta(\Xbf_i)  \mid \1\{\Xbf_1 \notin \Lcal(\xbf)\}, \Lcal(\xbf)]\\
%         &=\E[ S_i'\Delta(\Xbf_i)  \mid \{\Xbf_1 \notin \Lcal(\xbf)\}, \Lcal(\xbf)]\1\{\Xbf_1 \notin \Lcal(\xbf)\} + \\
%         &\E[ S_i'\Delta(\Xbf_i)  \mid \{\Xbf_1 \in \Lcal(\xbf)\}, \Lcal(\xbf)]\1\{\Xbf_1 \in \Lcal(\xbf)\},
%     \end{align*}
%     where in the second step we have used \eqref{forestass1starDRF_2}.
% %%%%%%%%%%%%%%%%%%%%%%%%%%%%%%%%%%%%%%%%%%%%%%%%%%
Before we go on, note that for any conditioning set $\F$,
\begin{align}\label{A_100d2}
  \left \|\sum_{i=2}^{s_n}\E[ S_i'\Delta(\Xbf_i)  \mid \F] \right \|_{\H}   &\leq   \E\left[ \sum_{i=2}^{s_n} S_i' \left \| \Delta(\Xbf_i)\right\|_{\H} \mid \F \right] \nonumber\\
  &\leq  \E\left[\sum_{i=2}^{s_n} S_i' C \left \| \Xbf_i -  \xbf\right\|_{\R^p} \mid \F \right] \nonumber\\
  &\leq C s_n^{-w} \E\left[\sum_{i=2}^{s_n} S_i' \mid \F \right]  \nonumber\\
  &\leq  C s_n^{-w},
\end{align}

We now show that:
\begin{claim}
For some $\epsilon > 0$,
    \begin{align}\label{A_firstpartofthepuzzle}
%\E[ \| \E[\sum_{i=2}^{s_n} S_i'\Delta(\Xbf_i)  \mid \1\{\Xbf_1 \notin \Lcal(\xbf)\}, \Lcal(\xbf)] -  \E[ \sum_{i=2}^{s_n} S_i'\Delta(\Xbf_i)  \mid \{\Xbf_1 \notin \Lcal(\xbf)\}, \Lcal(\xbf)]\|_{\H}^2 ]=\o(s_n^{-(1+\epsilon)}).
     \E[ \| \E[ \sum_{i=2}^{s_n} S_i'\Delta(\Xbf_i)  \mid \Lcal(\xbf), \Xbf_1] -  \E[ \sum_{i=2}^{s_n} S_i'\Delta(\Xbf_i)  \mid \{\Xbf_1 \notin \Lcal(\xbf)\}, \Lcal(\xbf)]\|_{\H}^2 ]=\o(s_n^{-(1+\epsilon)}).
\end{align} 
\end{claim}

\begin{claimproof}
We first note that
\begin{align}\label{newextendedconnection}
    \E[ \| \E[\sum_{i=2}^{s_n} S_i'\Delta(\Xbf_i)  \mid \Lcal(\xbf), \Xbf_1] -  \E[ \sum_{i=2}^{s_n}S_i'\Delta(\Xbf_i)  \mid \1\{\Xbf_1 \notin \Lcal(\xbf)\}, \Lcal(\xbf)] \|_{\H}^2 ] = \o(s_n^{-(1+\epsilon)}),
\end{align}
for some $\epsilon > 0$.
Indeed, since $\1\{\Xbf_1 \notin \Lcal(\xbf)\}$ is a deterministic function of $\Xbf_1$, $\Lcal(\xbf)$, we have that 
\begin{align}
    \E[\sum_{i=2}^{s_n} S_i'\Delta(\Xbf_i)  \mid \Lcal(\xbf), \Xbf_1] &= \E[\sum_{i=2}^{s_n} S_i'\Delta(\Xbf_i)  \mid \1\{\Xbf_1 \notin \Lcal(\xbf)\}, \Lcal(\xbf), \Xbf_1],
\end{align}
and since $f( (\Xbf_j, \1\{\Xbf_j \in \Lcal(\xbf)\})_{j \in \I_2, j \neq i} )= \sum_{i=2}^{s_n} S_i' \Delta(\Xbf_i)$ is bounded, we can use \eqref{forestass1starDRF_2} in \ref{forestass1starDRF}:
\begin{align*}
    &\E[ \| \E[ \sum_{i=2}^{s_n} S_i'\Delta(\Xbf_i)  \mid \Lcal(\xbf), \Xbf_1] -  \E[ \sum_{i=2}^{s_n}S_i'\Delta(\Xbf_i)  \mid \1\{\Xbf_1 \notin \Lcal(\xbf)\}, \Lcal(\xbf)] \|_{\H}^2 ]=\\
    &\E[ \| \E[\sum_{i=2}^{s_n} S_i'\Delta(\Xbf_i)  \mid \1\{\Xbf_1 \notin \Lcal(\xbf)\}, \Lcal(\xbf), \Xbf_1]-  \E[\sum_{i=2}^{s_n} S_i'\Delta(\Xbf_i)  \mid \1\{\Xbf_1 \notin \Lcal(\xbf)\}, \Lcal(\xbf)] \|_{\H}^2 ]\\
    &=\o(s_n^{-(1 + \epsilon)}),
\end{align*}
showing \eqref{newextendedconnection}. Starting from \eqref{newextendedconnection}, we note that
 \begin{align*}
     \E[\sum_{i=2}^{s_n} S_i'\Delta(\Xbf_i)  \mid \1\{\Xbf_1 \notin \Lcal(\xbf)\}, \Lcal(\xbf)]
     &=\E[\sum_{i=2}^{s_n} S_i'\Delta(\Xbf_i)  \mid \{\Xbf_1 \notin \Lcal(\xbf)\}, \Lcal(\xbf)]\1\{\Xbf_1 \notin \Lcal(\xbf)\} + \\
        &\E[ \sum_{i=2}^{s_n}S_i'\Delta(\Xbf_i)  \mid \{\Xbf_1 \in \Lcal(\xbf)\}, \Lcal(\xbf)]\1\{\Xbf_1 \in \Lcal(\xbf)\}
 \end{align*}
  and thus,
\begin{align}\label{A_firstpartofthepuzzle_1}
    &\E[ \| \E[\sum_{i=2}^{s_n} S_i'\Delta(\Xbf_i)  \mid \1\{\Xbf_1 \notin \Lcal(\xbf)\}, \Lcal(\xbf)] -  \E[ \sum_{i=2}^{s_n} S_i'\Delta(\Xbf_i)  \mid \{\Xbf_1 \notin \Lcal(\xbf)\}, \Lcal(\xbf)]\1\{\Xbf_1 \notin \Lcal(\xbf)\}\|_{\H}^2 ] \nonumber \\
    &=\E[\|\E[ \sum_{i=2}^{s_n} S_i'\Delta(\Xbf_i)  \mid \{\Xbf_1 \in \Lcal(\xbf)\}, \Lcal(\xbf)]\1\{\Xbf_1 \in \Lcal(\xbf)\}\|_{\H}^2] \nonumber \\
    &\leq \E[\E[ \sum_{i=2}^{s_n} S_i'\|\Delta(\Xbf_i)\|_{\H}  \mid \{\Xbf_1 \in \Lcal(\xbf)\}, \Lcal(\xbf)]^2  \1\{\Xbf_1 \in \Lcal(\xbf)\} ]\nonumber\\
      &\leq C s_{n}^{-2w} \Prob(\Xbf_1 \in \Lcal(\xbf)) \nonumber\\
      &=\O(s_{n}^{-(1+2w)}),
\end{align}
utilizing \eqref{A_100d2}. Similarly, 
\begin{align}\label{A_firstpartofthepuzzle_2}
    &\E[ \|  \E[ \sum_{i=2}^{s_n} S_i'\Delta(\Xbf_i)  \mid \{\Xbf_1 \notin \Lcal(\xbf)\}, \Lcal(\xbf)]\1\{\Xbf_1 \notin \Lcal(\xbf)\}  - \E[ \sum_{i=2}^{s_n} S_i'\Delta(\Xbf_i)  \mid \{\Xbf_1 \notin \Lcal(\xbf)\}, \Lcal(\xbf)]\|_{\H}^2 ]\nonumber\\
    &=\E[ \|  \E[ \sum_{i=2}^{s_n} S_i'\Delta(\Xbf_i)  \mid \{\Xbf_1 \notin \Lcal(\xbf)\}, \Lcal(\xbf)]\1\{\Xbf_1 \in \Lcal(\xbf)\} \|_{\H}^2]\nonumber\\
        &\leq C s_{n}^{-2w} \Prob(\Xbf_1 \in \Lcal(\xbf))\nonumber\\
      &=\O(s_{n}^{-(1+2w)}).
\end{align}
Repeatedly using $\| \xi_1 + \xi_2 \|_{\H}^2 \leq 2 \| \xi_1 \|_{\H}^2 +  2 \|\xi_2\|_{\H}^2$ and then utilizing \eqref{newextendedconnection}, \eqref{A_firstpartofthepuzzle_1} and \eqref{A_firstpartofthepuzzle_2}, we obtain
\begin{align*}
    &\E[ \| \E[ \sum_{i=2}^{s_n} S_i'\Delta(\Xbf_i)  \mid \Lcal(\xbf), \Xbf_1] -  \E[ \sum_{i=2}^{s_n} S_i'\Delta(\Xbf_i)  \mid \{\Xbf_1 \notin \Lcal(\xbf)\}, \Lcal(\xbf)]\|_{\H}^2 ] \leq \\
    & 2 \E[ \| \E[\sum_{i=2}^{s_n} S_i'\Delta(\Xbf_i)  \mid \Lcal(\xbf), \Xbf_1] -  \E[ \sum_{i=2}^{s_n}S_i'\Delta(\Xbf_i)  \mid \1\{\Xbf_1 \notin \Lcal(\xbf)\}, \Lcal(\xbf)] \|_{\H}^2 ] + \\
    & 4\E[ \| \E[ \sum_{i=2}^{s_n}S_i'\Delta(\Xbf_i)  \mid \1\{\Xbf_1 \notin \Lcal(\xbf)\}, \Lcal(\xbf)] -  \E[ \sum_{i=2}^{s_n} S_i'\Delta(\Xbf_i)  \mid \{\Xbf_1 \notin \Lcal(\xbf)\}, \Lcal(\xbf)]\1\{\Xbf_1 \notin \Lcal(\xbf)\}\|_{\H}^2 ] + \\
    & 4\E[ \|  \E[ \sum_{i=2}^{s_n} S_i'\Delta(\Xbf_i)  \mid \{\Xbf_1 \notin \Lcal(\xbf)\}, \Lcal(\xbf)]\1\{\Xbf_1 \notin \Lcal(\xbf)\}  - \E[ \sum_{i=2}^{s_n} S_i'\Delta(\Xbf_i)  \mid \{\Xbf_1 \notin \Lcal(\xbf)\}, \Lcal(\xbf)]\|_{\H}^2 ]\\
    &=\O(s_{n}^{-(1+2w)}),
\end{align*}
showing \eqref{A_firstpartofthepuzzle}. 
\end{claimproof}

Similarly, we have that
\begin{align*} 
    \E[ \sum_{i=2}^{s_n} S_i'\Delta(\Xbf_i)  \mid \Lcal(\xbf)]&=\E[ \sum_{i=2}^{s_n} S_i'\Delta(\Xbf_i)  \mid \{\Xbf_1 \notin \Lcal(\xbf)\}, \Lcal(\xbf)]\Prob(\Xbf_1 \notin \Lcal(\xbf) \mid \Lcal(\xbf)) + \nonumber\\
        &\E[ \sum_{i=2}^{s_n} S_i'\Delta(\Xbf_i)  \mid \{\Xbf_1 \in \Lcal(\xbf)\}, \Lcal(\xbf)]\Prob(\Xbf_1 \in \Lcal(\xbf) \mid \Lcal(\xbf))
\end{align*}
and, since $\E[\Prob(\Xbf_1 \in \Lcal(\xbf) \mid \Lcal(\xbf))^2]\leq \E[\Prob(\Xbf_1 \in \Lcal(\xbf) \mid \Lcal(\xbf))]$, we can use the same approach as in \eqref{A_firstpartofthepuzzle_1} and \eqref{A_firstpartofthepuzzle_2} to get:
\begin{align}\label{A_secondpartofthepuzzle}
    \E[ \| \E[ \sum_{i=2}^{s_n} S_i'\Delta(\Xbf_i)  \mid \Lcal(\xbf)] -  \E[ \sum_{i=2}^{s_n} S_i'\Delta(\Xbf_i)  \mid \{\Xbf_1 \notin \Lcal(\xbf)\}, \Lcal(\xbf)]\|_{\H}^2 ]=\o(s_n^{-(1+\epsilon)}).
\end{align}
% Indeed similarly as before, as $\E[\Prob(\Xbf_1 \in \Lcal(\xbf) \mid \Lcal(\xbf))^2]\leq \E[\Prob(\Xbf_1 \in \Lcal(\xbf) \mid \Lcal(\xbf))]$,
% \begin{align*}
%     &\E[ \| \E[ \sum_{i=2}^{s_n} S_i'\Delta(\Xbf_i)  \mid \Lcal(\xbf)] -  \E[ \sum_{i=2}^{s_n} S_i'\Delta(\Xbf_i)  \mid \{\Xbf_1 \notin \Lcal(\xbf)\}, \Lcal(\xbf)]\Prob(\Xbf_1 \notin \Lcal(\xbf) \mid \Lcal(\xbf))\|_{\H}^2 ]\\
%     &=\E[ \| \E[ \sum_{i=2}^{s_n} S_i'\Delta(\Xbf_i)  \mid \{\Xbf_1 \in \Lcal(\xbf)\}, \Lcal(\xbf)]\Prob(\Xbf_1 \in \Lcal(\xbf) \mid \Lcal(\xbf))\|_{\H}^2 ]\\
%     &\leq C s_n^{-(1+\epsilon)},
% \end{align*}
% and

Using \eqref{A_firstpartofthepuzzle} and \eqref{A_secondpartofthepuzzle}, we get:
\begin{align*}
    &\E[\|\E[ \sum_{i=2}^{s_n} S_i'\Delta(\Xbf_i)  \mid \Lcal(\xbf), \Xbf_1] - \E[ \sum_{i=2}^{s_n} S_i'\Delta(\Xbf_i)  \mid \Lcal(\xbf)]\|_{\H}^2 ] \\
    &\leq 2 \E[ \| \E[ \sum_{i=2}^{s_n} S_i'\Delta(\Xbf_i)  \mid \Lcal(\xbf), \Xbf_1] -  \E[ 
 \sum_{i=2}^{s_n} S_i'\Delta(\Xbf_i)  \mid \{\Xbf_1 \notin \Lcal(\xbf)\}, \Lcal(\xbf)]\|_{\H}^2 ]\\
    &+2\E[ \| \E[ \sum_{i=2}^{s_n} S_i'\Delta(\Xbf_i)  \mid \Lcal(\xbf)] -  \E[ \sum_{i=2}^{s_n} S_i'\Delta(\Xbf_i)  \mid \{\Xbf_1 \notin \Lcal(\xbf)\}, \Lcal(\xbf)]\|_{\H}^2 ]\\
    &=\o(s_n^{-(1+\epsilon)}),
\end{align*}
proving Claim \eqref{crucialclaimforsumSi}.
\end{claimproof}
%%%%%%%%%%%%%%%%%%%%%%%%%%%%%%%%%%%%%%%%%%%%%%%%%%%%%%%%%%%%%%

Thus it follows that
\begin{align*}
    &\Var(\E[\sum_{i=2}^{s_n} S_i'\Delta(\Xbf_i)  \mid \Xbf_1] ) \\
    &= \Var( \E[ \E[\sum_{i=2}^{s_n} S _i'\Delta(\Xbf_i) \mid \Lcal(\xbf)]  \mid \Xbf_1] )+ \Var( \E[ \E[ \sum_{i=2}^{s_n} S_i'\Delta(\Xbf_i) \mid \Lcal(\xbf)]  \mid \Xbf_1]-\E[ \sum_{i=2}^{s_n} S_i'\Delta(\Xbf_i)  \mid \Xbf_1] )\\
    &+ 2 \Cov(\E[ \E[\sum_{i=2}^{s_n} S _i'\Delta(\Xbf_i) \mid \Lcal(\xbf)]  \mid \Xbf_1],\E[ \E[ \sum_{i=2}^{s_n} S _i'\Delta(\Xbf_i) \mid \Lcal(\xbf)]  \mid \Xbf_1]-\E[\sum_{i=2}^{s_n} S_i'\Delta(\Xbf_i)  \mid \Xbf_1]  )
\end{align*}
First, note that by assumption \ref{forestass1starDRF}, and the boundedness of $\E[ \sum_{i=2}^{s_n} S _i'\Delta(\Xbf_i) \mid \Lcal(\xbf)]$,
\begin{align*}
    \Var( \E[ \E[ \sum_{i=2}^{s_n} S _i'\Delta(\Xbf_i) \mid \Lcal(\xbf)]  \mid \Xbf_1] ) = \o(s_n^{-(1+\epsilon)}).
\end{align*}
Moreover, using Jensen's inequality and \eqref{crucialclaimforsumSi},
\begin{align*}
    &\Var( \E[ \E[\sum_{i=2}^{s_n} S _i'\Delta(\Xbf_i) \mid \Lcal(\xbf)]  \mid \Xbf_1]-\E[ \sum_{i=2}^{s_n} S_i'\Delta(\Xbf_i)  \mid \Xbf_1] )  \\
    &\leq\E[  \|\E[ \E[\sum_{i=2}^{s_n}  S _i'\Delta(\Xbf_i) \mid \Lcal(\xbf)] - \E[\sum_{i=2}^{s_n} S _i'\Delta(\Xbf_i) \mid \Xbf_1, \Lcal(\xbf)]  \mid \Xbf_1] \|_{\H}^2 ]\\
    &\leq \E[  \E[\| \E[\sum_{i=2}^{s_n} S _i'\Delta(\Xbf_i) \mid \Lcal(\xbf)] - \E[ \sum_{i=2}^{s_n} S _i'\Delta(\Xbf_i) \mid \Xbf_1, \Lcal(\xbf)]\|_{\H}^2  \mid \Xbf_1]  ]\\
    &= \E[\| \E[\sum_{i=2}^{s_n} S _i'\Delta(\Xbf_i) \mid \Lcal(\xbf)] - \E[\sum_{i=2}^{s_n} S _i'\Delta(\Xbf_i) \mid \Xbf_1, \Lcal(\xbf)]\|_{\H}^2    ]\\
    &=\o(s_n^{-(1+\epsilon)}).
\end{align*}
Adding the Cauchy-Schwarz inequality to bound the covariances, we obtain
\begin{align*}
    \Var(\E[\sum_{i=2}^{s_n} S_i'\Delta(\Xbf_i)  \mid \Xbf_1] )=\o(s_n^{-(1+\epsilon)}),
\end{align*}
showing \eqref{A_whatweneed}. As outlined above, \eqref{A_whatweneed} and \eqref{A_whatweneed2} imply \eqref{ConnectionbetweenSkandS}.
\end{claimproof}

In turn, \eqref{ConnectionbetweenSkandS} combined with \eqref{ConnectionbetweenS1andSi} proves Claim \eqref{A_variance2term_1}.
\end{claimproof}

Finally, we consider the covariance between $\E[S_1' k(\Ybf_i, \cdot)\mid \Xbf_1]$ and $\sum_{i=2}^{s_n} \E[S_i' k(\Ybf_i, \cdot)\mid \Xbf_1]$.

\begin{claim}
For some $\epsilon > 0$, we have
\begin{align}\label{A_variance3term_1}
        &\Cov\left( \sum_{i=2}^{s_n} \E[S_i' k(\Ybf_i, \cdot)\mid \Xbf_1], \E[S_1'|\Xbf_1]  \E[ k(\Ybf_1, \cdot)|\Xbf_1] \right)\nonumber \\
        &= -\Var(\E[S_1'|\Xbf_1]) \|\E[k(\Ybf_1, \cdot)|\Xbf=\xbf]\|_{\H}^2 + \o(s_n^{-(1+\epsilon)})
\end{align}
and
\begin{align}\label{A_variance3term_2}
    &\Cov\left( \sum_{i=2}^{s_n} \E[S_i' \langle k(\Ybf_i, \cdot),f \rangle\mid \Xbf_1], \E[S_1'|\Xbf_1]  \E[ \langle k(\Ybf_1, \cdot), f \rangle|\Xbf_1] \right) \nonumber\\
    &= -\Var(\E[S_1'|\Xbf_1]) \E[\langle k(\Ybf_1, \cdot), f \rangle|\Xbf=\xbf]^2 + \o(s_n^{-(1+\epsilon)}).
\end{align}
\end{claim}

\begin{claimproof}
Again, we only show~\eqref{A_variance3term_1}, because~\eqref{A_variance3term_2} follows analogously. Using~\eqref{honestchange}, we can subtract and add $\mu(\xbf)\sum_{i=2}^{s_n} \E[S_i'|\Xbf_1]$ and $\E[S_1'|\Xbf_1] \mu(\xbf)$ to obtain
\begin{align*}
     &\Cov\left( \sum_{i=2}^{s_n} \E[S_i' k(\Ybf_i, \cdot)\mid \Xbf_1], \E[S_1'|\Xbf_1]  \E[ k(\Ybf_1, \cdot)|\Xbf_1] \right)\\
     =&  \Cov\left( \sum_{i=2}^{s_n} \E[ S_i' \E[  k(\Ybf_i, \cdot) \mid  \Xbf_i]  \mid \Xbf_1], \E[S_1'|\Xbf_1]  \E[ k(\Ybf_1, \cdot)|\Xbf_1] \right)\\
    =&  \Cov\left( \sum_{i=2}^{s_n} \E[ S_i' \Delta(\Xbf_i)  \mid \Xbf_1] + \mu(\xbf) \sum_{i=2}^{s_n} \E[S_i' \mid \Xbf_1 ], \E[S_1'  \Delta(\Xbf_1)|\Xbf_1]  + \mu(\xbf) \E[S_1'|\Xbf_1]  \right)\\
    =&\Cov\left( \sum_{i=2}^{s_n} \E[ S_i' \Delta(\Xbf_i)  \mid \Xbf_1] , \E[S_1' \Delta(\Xbf_1)|\Xbf_1]    \right) +\Cov\left( \sum_{i=2}^{s_n} \E[ S_i' \Delta(\Xbf_i)  \mid \Xbf_1] , \mu(\xbf) \E[S_1'|\Xbf_1] \right)\\
    &\quad+\Cov\left( \mu(\xbf) \sum_{i=2}^{s_n} \E[S_i' \mid \Xbf_1 ] , \E[S_1' \Delta(\Xbf_1)|\Xbf_1]    \right) +  \Cov\left( \mu(\xbf) \sum_{i=2}^{s_n} \E[S_i' \mid \Xbf_1 ], \mu(\xbf) \E[S_1'|\Xbf_1]  \right)\\
    =& (I) + (II) + (III) + (IV),
\end{align*}
where again $\Delta(\Xbf_i)=\E[  k(\Ybf_i, \cdot) \mid  \Xbf_i]  - \mu(\xbf)$. Since from~\eqref{A_Sdproperty},
\begin{align}
    \mu(\xbf)\sum_{i=2}^{s_n} \E[S_i'|\Xbf_1] = \mu(\xbf)(\P\left(\text{diam}(\Lcal(\xbf))\leq s_n^{-w} \mid \Xbf_1\right)-\E[S_1' \mid \Xbf_1]),
\end{align}
it holds that
\begin{align*}
    (IV)&=\Cov\left( \mu(\xbf) (\P\left(\text{diam}(\Lcal(\xbf))\leq s_n^{-w} \mid \Xbf_1\right)-\E[S_1'|\Xbf_1]), \mu(\xbf) \E[S_1'|\Xbf_1]  \right)\\
    &=\| \mu(\xbf) \|_{\H}^2 \left( \Cov(\P\left(\text{diam}(\Lcal(\xbf))\leq s_n^{-w} \mid \Xbf_1\right), \E[S_1'|\Xbf_1]) - \Var(\E[S_1'|\Xbf_1]) \right) \\
    &=-\Var(\E[S_1'|\Xbf_1]) \| \mu(\xbf) \|_{\H}^2 + \o(s_n^{-(1+\epsilon)}),
\end{align*}
where the last step again followed by Cauchy--Schwarz and \ref{forestass1starDRF}:
\begin{align*}
&\Cov(\P\left(\text{diam}(\Lcal(\xbf))\leq s_n^{-w} \mid \Xbf_1\right), \E[S_1'|\Xbf_1]) \\
&\leq \Var(\P\left(\text{diam}(\Lcal(\xbf))\leq s_n^{-w} \mid \Xbf_1\right))^{1/2} \Var( \E[S_1'|\Xbf_1])^{1/2}\\
&=\o(s_n^{-(1+\epsilon)}),%
\end{align*}
as $\Var( \E[S_1'|\Xbf_1])=\O(s_n^{-1})$.

Subsequently, we show that the remaining terms are negligible.
Due to the Cauchy--Schwarz inequality, we have
\begin{align*}
    |(I)| \leq  \Var\left(\sum_{i=2}^{s_n} \E[ S_i' \Delta(\Xbf_i)  \mid \Xbf_1]\right)^{1/2} \Var(\E[S_1' \Delta(\Xbf_1)|\Xbf_1])^{1/2} .
\end{align*}
As proven above (combining~\eqref{A_whatweneed} and~\eqref{A_whatweneed22}), $\Var(\sum_{i=2}^{s_n} \E[ S_i' \Delta(\Xbf_i)  \mid \Xbf_1])=\o(s_n^{-(1+\epsilon)})$, and it holds that
\[
\Var(\E[S_1' \Delta(\Xbf_1)|\Xbf_1]) \leq \E[  \E[S_1' \|\Delta(\Xbf_1)\|_{\H}|\Xbf_1]^2]=\mathcal{O}(s_n^{-(1+2w)})
\]
holds. Consequently, $(I)=\o(s_n^{-(1+\epsilon)})$. Similarly,
\begin{align*}
    |(II)| \leq  \Var\left(\sum_{i=2}^{s_n} \E[ S_i' \Delta(\Xbf_i)  \mid \Xbf_1] \right)^{1/2} \Var(\mu(\xbf) \E[S_1'|\Xbf_1]) ^{1/2}=\o(s_n^{-(1+\epsilon)}),
\end{align*}
as $\Var(\E[S_1'|\Xbf_1])\leq \E[(S_1')^2]=\mathcal{O}(s_n^{-1}).$
Finally, 
\begin{align*}
    |(III)|&=|\Cov\left( \mu(\xbf) (1-\E[S_1' \Delta(\Xbf_1)|\Xbf_1] ) , \E[S_1' \Delta(\Xbf_1)|\Xbf_1]    \right)| \\
    &=|- \|\mu(\xbf)\|_{\H}^2 \Var(\E[S_1' \Delta(\Xbf_1)|\Xbf_1])|\\
    &=\mathcal{O}(s_n^{-(1+2w)})
\end{align*}
as above. Thus (I), (II) and (III) are all $\o(s_n^{-(1+\epsilon)})$, for (potentially different) $\epsilon > 0$. Combining this with (IV), we obtain Claim \eqref{A_variance3term_1}.

\end{claimproof}

Combining~\eqref{A_variance1term_1},~\eqref{A_variance2term_1}, and~\eqref{A_variance3term_1}, we obtain
\begin{align*}
        &\Var( \E[T'(\Zcal_{s_n})| \Xbf_1 ])\\
        &= 2\Var(\E[S_1'|\Xbf_1]) \|\E[k(\Ybf_1, \cdot)|\Xbf=\xbf]\|_{\H}^2 - 2\Var(\E[S_1'|\Xbf_1]) \|\E[k(\Ybf_1, \cdot)|\Xbf=\xbf]\|_{\H}^2  + \o(s_n^{-(1+\epsilon)})\\
        &=\o(s_n^{-(1+\epsilon)})
\end{align*}
proving~\eqref{A_Varboundnew}.
% and analogously
% \begin{align*}
%        \Var( \E[\langle T'(\Zcal_{s_n}), f \rangle \mid \Zbf_1 ])&= \o(s_n^{-(1+\epsilon)}),
% \end{align*}

\end{claimproof}

%We recall that $\Var(\E[S_1'|\Xbf_1]) \sim \Var(\E[S_1|\Xbf_1])=\Var(\E[S_1|\Zbf_1])=\Omega((s_n\log(s_n))^{-1})$, by~\eqref{truncation_S},~\eqref{Honestyconsequence}, and Lemma~\ref{lemma4} respectively. This together with Claim~\eqref{A_Varboundnew} and the expansion in~\eqref{varexpansion} establishes~\eqref{varianceassumption}.

\end{proof}

\section{Main Proofs} \label{appendix: proofs}

We recall from the main text that
\begin{equation*}
    \htauk = \binom{n}{s_n}^{-1}  \sum_{i_1 < i_2 < \cdots < i_{s_n}} \E_{\mathcal{E}} \left[\Gamma(\xbf; \mathcal{E}, \{\mathbf{Z}_{i_1}, \ldots, \mathbf{Z}_{i_{s_n}}\})\right],
\end{equation*}
whereby 
%%% Note: Do not delete, this notation is ok, because we do not specify what i is!!!%%%
\begin{align*}% \{\mathbf{Z}_{1}, \ldots, \mathbf{Z}_{s_n}\}
    \E_{\mathcal{E}} \left[\Gamma(\xbf; \mathcal{E}, \{\mathbf{Z}_{i_1}, \ldots, \mathbf{Z}_{i_{s_n}}\})\right]= &\frac{1}{ | \{j: \Xbf_j \in \Lcal(\xbf), W_j=1 \}|}\sum_{\Xbf_i\in\Lcal(\xbf), W_i=1} k(\Ybf_i,\cdot) - \nonumber \\
    &\frac{1}{ | \{j: \Xbf_i \in \Lcal(\xbf), W_j=0 \}|}\sum_{\Xbf_i\in\Lcal(\xbf), W_i=0} k(\Ybf_i,\cdot).
\end{align*}
To abbreviate this, we will just write 
%%% Note: Do not delete!!
%% As we say below, it should really be \sum_{i \in \mathcal{I}_2} S_i k(\Ybf_{i}, \cdot), but we ignore it to simplify notation!
%%%
\begin{align}\label{singletree}
   \Gamma(\mathcal{Z}_{s_n})= \E_{\mathcal{E}} \left[\Gamma(\xbf; \mathcal{E}, \{\mathbf{Z}_{i_1}, \ldots, \mathbf{Z}_{i_{s_n}}\})\right]=\sum_{i=1}^{s_n} S_i k(\Ybf_{i}, \cdot)
\end{align}
with weights
\begin{align}\label{SdefGamma}
     S_{i1}&= \frac{\1\{\Xbf_{i} \in \mathcal{L}(\xbf), W_i=1\}}{| \{j: \Xbf_j \in \Lcal(\xbf), W_j=1 \}|}\\
     S_{i0}&= -\frac{\1\{\Xbf_{i} \in \mathcal{L}(\xbf), W_i=0\}}{| \{j: \Xbf_j \in \Lcal(\xbf), W_j=0 \}|}\\
     S_i &= S_{i1}+S_{i0}.
\end{align}
% \begin{align}\label{SdefGamma}
%     S_i&= \frac{\1\{\Xbf_{i} \in \mathcal{L}(\xbf), W_i=1\}}{| \{j: \Xbf_j \in \Lcal(\xbf), W_j=1 \}|} - \frac{\1\{\Xbf_{i} \in \mathcal{L}(\xbf), W_i=0\}}{| \{j: \Xbf_j \in \Lcal(\xbf), W_j=0 \}|} \nonumber \\
%     &=W_i S_{i1} - (1-W_i) S_{i0} 
% \end{align}
We note that $\htauk$ takes the sum over all possible subsets of size $s_n$. As such, the chosen subsets of size $s_n$ might not have the indices $\Zbf_1, \ldots, \Zbf_{s_n}$ in general. To simplify notation we usually only consider the indices $i=1,\ldots, s_n$, appealing to the assumption of i.i.d. data.

\begin{lemma} \label{lemma4}
Suppose \dataass{1} holds for $\Xbf_1,\Xbf_2, \ldots$ and \forestass{3}, \forestass{4} for $\Gamma(\Zcal_{s_n})$. Moreover, assume \ref{causalityass1} holds. Then, there is a constant $C_{f,p}$ depending on $f$ and $p$ such that,
\begin{equation}\label{loverboundvarianceSbehavior}
    s_n \Var(\E[S_1 | \Zbf_1 ]) \succsim \frac{\varepsilon}{\kappa} \frac{C_{f,p}}{\log(s_n)^p},
\end{equation}
where $\varepsilon > 0$ is from the overlap condition in \ref{causalityass1}. When $f$ is uniform over $[0,1]^p$, the bound holds with $C_{f,p}=2^{-(p+1)} (p-1)!$. Moreover, it holds that for $j \in \{0,1\}$
\begin{align}\label{crucialSproperties2}
    &\sum_{i=1}^{s_n} |S_{i}|=2, \sum_{i: W_i=j} |S_{ij}|=1 \nonumber\\
    &\E[|S_{1}|]=\O(s_n^{-1}), \E[|S_{1j}|]=\O(s_n^{-1}), \E[|S_{1j}| \mid W_1=j]=\O(s_n^{-1}), \nonumber\\ 
    &\Var( \E[S_{1}\mid \Zbf_1 ])=\O(s_n^{-1}), \E[S_{1}^2]=\O(s_n^{-1}), \Var( \E[S_{1j}\mid \Zbf_1 ])=\O(s_n^{-1}), \E[S_{1j}^2]=\O(s_n^{-1}).
\end{align}
\end{lemma}

\begin{proof}

The proof of \eqref{loverboundvarianceSbehavior} is the same as in \citep[p. 47]{wager2017estimation} when proving incrementally. 

We next note that for $i=1,\ldots,s_n$ 
\begin{align*}
    |S_i| = \begin{cases}
        \frac{\1\{\Xbf_{i} \in \mathcal{L}(\xbf)\}}{| \{j: \Xbf_j \in \Lcal(\xbf), W_j=1 \}|} & \text{ if } W_i=1\\
        \frac{\1\{\Xbf_{i} \in \mathcal{L}(\xbf)\}}{| \{j: \Xbf_j \in \Lcal(\xbf), W_j=0 \}|} & \text{ if } W_i=0,
    \end{cases}
\end{align*}
and that by \ref{forestass4}, $| \{j: \Xbf_j \in \Lcal(\xbf), W_j=l \}|$ is not empty, for $l \in \{0,1\}$. Thus indeed $\sum_{i: W_i=j} |S_{ij}|=1$, so that $\sum_{i=1}^{s_n} |S_{i}|=2$. Thus it follows by the i.i.d assumption that
\begin{align*}
    2 = \E[\sum_{i=1}^{s_n} |S_i|] = s_n \E[|S_1|]
    %1=\E[ \sum_{i: W_i=j} |S_{ij}|] = s_n  \E[|S_{1j}|],
\end{align*}
showing $\E[|S_{1}|]=\O(s_n^{-1})$. Since $ \E[|S_{1}|]=\E[|S_{11}|] + \E[|S_{10}|]$, this also shows that $\E[|S_{1j}|]=\O(s_n^{-1})$. By the overlap conditon \ref{causalityass1}, we moreover have:
\[
\E[|S_{1j}| \mid W_1=j] = \frac{\E[|S_{1j}|]}{\Prob(W_1=j)} \leq \frac{\E[|S_{1j}|]}{\min(\varepsilon,1-\varepsilon)},
\]
showing that $\E[|S_{1j}| \mid W_1=j]=\O(s_n^{-1})$ as well.
Finally, using the fact that $|S_{ij}| \in [0,1]$,
\begin{align*}
    0\leq \Var(\E[S_{1j} \mid \Zbf_1]) \leq \E[ \E[S_{1j} \mid \Zbf_1]^2 ] \leq  \E[ \E[|S_{1j}| \mid \Zbf_1] ] = \E[|S_{1j}|]=\O(s_n^{-1}),
\end{align*}
and similarly for $\E[S_1^2]$. 
\end{proof}

% and $\Var(\E[S_1 | \Xbf_1, W_1 ])=\O(s_n^{-1})$ and $\E[S_1 ]=\O(s_n^{-1})$.
We note that by honesty \ref{forestass1}, $\E[S_1 | \Zbf_1 ]= \E[S_1 | \Xbf_1, W_1 ]$, and so the same properties also hold, for $\E[S_1 | \Xbf_1, W_1 ]$, in particular
\begin{equation}\label{loverboundvarianceSbehavior2}
    s_n \Var(\E[S_1 | \Xbf_1, W_1 ]) \succsim \frac{\varepsilon}{\kappa} \frac{C_{f,p}}{\log(s_n)^p}.
\end{equation}
This also implies that
\begin{align}\label{firsttimevarminusE}
    \Var(\E[S_1 | \Xbf_1, W_1 ])&= \E[\E[S_1 | \Xbf_1, W_1 ]^2]-\E[S_1 ]^2 \nonumber \\
    &=\E[\E[S_1 | \Xbf_1, W_1 ]^2] + \O(s_n^{-2}).
\end{align}
With \eqref{loverboundvarianceSbehavior2}, this means that the difference between $\Var(\E[S_1 | \Zbf_1 ])$ and $\E[\E[S_1 | \Zbf_1 ]^2]$ is negligible. Let 
\begin{align}\label{vardefSij}
    \Var(\E[S_{1j} | \Xbf_1, W_1=j] )=\E[\E[S_{1j} | \Xbf_1, W_1 ]^2 \mid W_1=j] - \E[S_{1j} | W_1=j ]^2.
\end{align}
Then,  
\begin{align}\label{overallvsindividual}
    &\E[\E[S_1 | \Xbf_1, W_1 ]^2] \nonumber\\
    &=\E[\E[ \E[\E[S_1 | \Xbf_1, W_1 ]^2 \mid W_1] ]] \nonumber \\
    &=\E[\E[S_1 | \Xbf_1, W_1 ]^2 \mid W_1=1] \Prob(W_1=1)  + \E[\E[S_1 | \Xbf_1, W_1 ]^2 \mid W_1=0]\Prob(W_1=0)  \nonumber\\
    &=\Var(\E[S_{11} | \Xbf_1, W_1=1 ])  \Prob(W_1=1)  + \Var(\E[S_{10} | \Xbf_1, W_1=0 ])\Prob(W_1=0)+ \O(s_n^{-2})
    %&=\Var(\E[S_{11} | \Xbf_1, W_1 ] |  W_1=1 )  \Prob(W_1=1)  + \Var(\E[S_{10} | \Xbf_1, W_1 ] |  W_1=0)\Prob(W_1=0)+ \O(s_n^{-2}),
\end{align}
as for $j \in \{0,1\}$
\[
%\E[\E[S_1 | \Xbf_1, W_1 ] \mid W_1=j]^2\Prob(W_1=0) \leq 
\E[S_{1j} \mid W_1=j]^2 \Prob(W_1=j)=\O(s_n^{-2}). 
\]
Combining \eqref{firsttimevarminusE} and \eqref{overallvsindividual}, we get
\begin{align}\label{overalliszeroplusone}
\Var(\E[S_1 | \Xbf_1, W_1 ])&=\Var(\E[S_{11} | \Xbf_1, W_1=1 ] )  \Prob(W_1=1)  \nonumber \\
    &+ \Var(\E[S_{10} | \Xbf_1, W_1=0 ] )\Prob(W_1=0)+ \O(s_n^{-2}).
    % \Var(\E[S_1 | \Xbf_1, W_1 ])&=\Var(\E[S_{11} | \Xbf_1, W_1 ] |  W_1=1 )  \Prob(W_1=1)  \nonumber \\
    % &+ \Var(\E[S_{10} | \Xbf_1, W_1 ] |  W_1=0)\Prob(W_1=0)+ \O(s_n^{-2}).
\end{align}
This helps to relate the behavior of $\Var(\E[S_{1j} | \Xbf_1, W_1=j ] )$ to the overall variance $\Var(\E[S_1 | \Xbf_1, W_1 ])$. In particular, either $\Var(\E[S_{11} | \Xbf_1, W_1=1 ] )= \Omega( (\log(s_n)^p s_n)^{-1})$ or $\Var(\E[S_{10} | \Xbf_1, W_1=0 ])= \Omega( (\log(s_n)^p s_n)^{-1})$ or both. These arguments will be crucial for the proof of Theorem \ref{thm:varianceassumption} below. We now turn to the behavior of the bias:

\begin{lemma} \label{lemma2}
Assume \forestass{2}--\forestass{4} 
that is trained on data $\Zcal_{s_n}$. Suppose that assumption~\dataass{1} holds for $\Xbf_1,\ldots, \Xbf_{s_n}$ and that \ref{causalityass1} holds. Then,
\begin{align}
    \P \left( \text{\textup{diam}} (\Lcal(\xbf)) \geq \sqrt{p} \left( \frac{\varepsilon s_n}{2\kappa-1}\right) ^{-0.51 \frac{\log((1-\alpha)^{-1})}{\log(\alpha^{-1})} \frac{\pi}{p}}   \right) \leq p \left( \frac{\varepsilon s_n}{2\kappa-1}\right) ^{-1/2 \frac{\log((1-\alpha)^{-1})}{\log(\alpha^{-1})} \frac{\pi}{p}}.
\end{align}
\end{lemma}

\begin{proof}
   This follows directly from the arguments in \citep[pages 34 and 46]{wager2017estimation}, in addition to a union bound over all dimensions of $\Xbf$. In particular, instead of $s_n$ we need to consider the minimal number of observations of the smaller class $\Gamma$ receives, say $s_{n, min}$. By \ref{causalityass1}, it holds a.s. that $s_{n, min}  \succsim \varepsilon s_n $, such that $s_{n, min}$ can be replaced by $\varepsilon s_n$.
\end{proof}

% \begin{lemma}\label{helperlemma}
% Let $T$ be a tree satisfying \ref{forestass1} and \forestass{5}.
% Then,
% \begin{align}\label{star1star}
%     \E[T^{(-j)}(\Zcal_{s_n})] = \E[ \E[k(\Ybf,\cdot) \mid \Xbf^{(-j)} \in \Lcal^{(-j)}(\xbf^{(-j)}) ] ].
% \end{align}
% \end{lemma}

\begin{corollary}\label{bias}
In addition to the conditions of Lemma~\ref{lemma2}, assume \dataass{2}, \ref{causalityass2} and that the trees $\Gamma(\Zcal_{s_n})$ in the forest satisfy \ref{forestass1}. Then, we have
\begin{equation}\label{biasbound}
    \| \E[\htauk] - \tauk \|_{\H} = \O\left( (\varepsilon s_n)^{-1/2 \frac{\log((1-\alpha)^{-1})}{\log(\alpha^{-1})} \frac{\pi}{p}}\right).
\end{equation}
% and
% \begin{equation}\label{propabiltiyconv1}
% \| \E[\xi \mid \Xbf \in \Lcal(\xbf)]\|_{\H} \stackrel{p}{\to} \| \E[\xi \mid \Xbf=\xbf] \|_{\H}.
% \end{equation}
% If moreover \dataass{3} holds, then we have
% \begin{align}\label{propabiltiyconv2}
%     \E[ \|\xi\|_{\H}^2 \mid \Xbf \in \Lcal(\xbf)]  &\stackrel{p}{\to} \E[\|\xi\|_{\H}^2 \mid \Xbf=\xbf].
% \end{align}
\end{corollary}

\begin{proof}
We recall that
\begin{align*}
     S_i&= \frac{\1\{\Xbf_{i} \in \mathcal{L}(\xbf), W_i=1\}}{| \{j: \Xbf_j \in \Lcal(\xbf), W_j=1 \}|} - \frac{\1\{\Xbf_{i} \in \mathcal{L}(\xbf), W_i=0\}}{| \{j: \Xbf_j \in \Lcal(\xbf), W_j=0 \}|}\\
     &=S_{i1} + S_{i0}
\end{align*}
Analogous to \citep{wager2017estimation}, we may write
\begin{align*}
   & \| \E[\Gamma(\Zcal_{s_n})] - \tauk\|_{\H}  \\
   &=\|\E\left[\sum_{W_i=1} S_{i1} k(\Ybf_i,\cdot) \right]- \E\left[\sum_{ W_i=0} ( - S_{i0}) k(\Ybf_i,\cdot)\right] - \tauk\|_{\H}\\
    &\leq \|\E[\sum_{ W_i=1} S_{i1} k(\Ybf_i,\cdot) ] - \E[k(\Ybf^{(1)}, \cdot) \mid \Xbf=\xbf]\|_{\H}+  \\
    &\| \E[\sum_{ W_i=0} (-S_{i0}) k(\Ybf_i,\cdot)]- \E[k(\Ybf^{(0)}, \cdot) \mid \Xbf=\xbf]\|_{\H}.
\end{align*}
Adapting the arguments of \citep[Lemma 12]{DRF-paper}, using that $\sum_{i: W_i=1} S_{i1}=\sum_{i: W_i=0} (-S_{i0})=1$,
\begin{align*}
    \E[\sum_{ W_j=1} S_{i1} k(\Ybf_i,\cdot)] &= \E[ \E[ k(\Ybf, \cdot) \mid \Xbf \in \Lcal(\xbf), W=1, \Lcal(\xbf) ]   ]\\
    &=\E[ \E[ k(\Ybf^{(1)}, \cdot) \mid \Xbf \in \Lcal(\xbf), \Lcal(\xbf)]   ],
\end{align*}
% $N_{\xbf,0}=|\{j: \Xbf_j \in \Lcal(\xbf), W_j=0 \}|$, $s_0$
% \begin{align*}
%     &\E[\frac{1}{ | \{j: \Xbf_j \in \Lcal(\xbf), W_j=0 \}|}\sum_{\Xbf_j\in\Lcal(\xbf), W_j=0} k(\Ybf_j,\cdot)] = s_0 \E[ \frac{\1\{\Xbf_1 \in \Lcal(\xbf), W_1=0 \}}{N_{\xbf,0}}k(\Ybf_1,\cdot) ]\\
% &=s_0 \E[\E[ \frac{\1\{\Xbf_1 \in \Lcal(\xbf), W_1=0 \}}{N_{\xbf,0}}k(\Ybf_1,\cdot) \mid \Xbf_1 \in \Lcal(\xbf), W_1=0, \Lcal(\xbf) ] ]\\
% \end{align*}
% Now by \ref{forestass4}, $N_{\xbf,0}$ and $k(\Ybf_1,\cdot)$ are independent given the knowledge that $\Xbf_1 \in \Lcal(\xbf)$, so that:
% \begin{align*}
%     &=s_0 \E[\E[ \frac{\1\{\Xbf_1 \in \Lcal(\xbf), W_1=0 \}}{N_{\xbf,0}} \mid  \Xbf_1 \in \Lcal(\xbf), W_1=0, \Lcal(\xbf)] \E[k(\Ybf_1,\cdot) \mid \Xbf_1 \in \Lcal(\xbf), W_1=0, \Lcal(\xbf) ] ]
% \end{align*}
where the last argument follows by \ref{causalityass2}. The same applies to the control group, and with \ref{dataass2} thus we can directly apply the arguments in \citep[Corollary 13]{DRF-paper} for each term separately to obtain 
\begin{align*}
     \| \E[\Gamma(\Zcal_{s_n})] - \tauk\|_{\H} = \O\left( (\varepsilon s_n)^{-1/2 \frac{\log((1-\alpha)^{-1})}{\log(\alpha^{-1})} \frac{\pi}{p}}\right).
\end{align*}
Since the bias of a tree is the same as the bias of the overall forest, we obtain \eqref{biasbound}.
\end{proof}

% In the following, we prove a cornerstone of the theory in \citep{näf2023confidence} for the causal regression tree $\Gamma(\mathcal{Z}_{s_n})$. In addition to the correction made in Section \ref{DRFproofcorrection}, this will require additional care as treatment and control groups need to be considered separately. First we need a helpful auxiliary result: 

% \begin{lemma}\label{varianceboundinglemma}
% Let $\xi_1, \xi_2, \xi \in \mathbb{L}^2(\Omega, \mathcal{A}, \H)$. Then
%  \begin{align}\label{sumbounding}
%      \Var(\xi_1 + \xi_2) \leq 2\Var(\xi_1) + 2\Var(\xi_2 ).
%  \end{align}
% Moreover, assume that $\|\xi\|_{\H} \leq C$, for some constant $C > 0$. Then for any random variable $Z$ with $\Var(Z) < \infty$,
% \[
% \Var(Z \xi) \leq 2 \Var(Z) C^2 + 2 \E[Z]^2 \Var(\xi).
% \]
% \end{lemma}

% \begin{proof}
% We note that for two random elements $\xi_1, \xi_2 \in \mathbb{L}^2(\Omega, \mathcal{A}, \H)$ it holds that
% \begin{align*}
%     0&\leq \Var(\xi_1 + \xi_2) = \Var(\xi_1) + \Var(\xi_2 ) + 2 \Cov(\xi_1, \xi_2)\\
%     0&\leq \Var(\xi_1 - \xi_2) = \Var(\xi_1) + \Var(\xi_2 ) - 2 \Cov(\xi_1, \xi_2),
% \end{align*}
%     so that
%     \begin{align*}
%         |\Cov(\xi_1, \xi_2)| \leq \Var(\xi_1) + \Var(\xi_2 ),
%     \end{align*}
% and \eqref{sumbounding} holds. Applying \eqref{sumbounding} to,
% \begin{align*}
%     \Var(Z \xi) &= \Var(Z \xi - \E[Z] \xi + \E[Z] \xi )\\
%     &\leq 2 \Var(\xi(Z-\E[Z])) + \Var(\E[Z] \xi)\\
%     &\leq 2 \E[\|\xi\|_{\H}^2(Z-\E[Z])^2] + \E[Z]^2\Var( \xi)\\
%     &\leq 2 \Var(Z)C^2 + \E[Z]^2\Var( \xi).
% \end{align*} 
% \end{proof}

Denote by 
\begin{align}
    \Gamma_n(\Zbf_i) = \E[\Gamma(\mathcal{Z}_{s_n})\mid\Zbf_i] - \E[T(\mathcal{Z}_{s_n})],
\end{align}
and by
\begin{align}\label{firstorderapproxT}
    \tilde{\Gamma}(\Zcal_{s_n})=\sum_{i=1}^{s_n}\Gamma_n(\Zbf_i),% \E[\Gamma(\mathcal{Z}_{s_n})\mid\Zbf_i] - \E[T(\mathcal{Z}_{s_n})],
\end{align}
the first order approximation of $\Gamma(\Zcal_{s_n})$ as in Lemma \ref{Hdecomposition}. Recall from the main text that $\mu^{(0)}(\xbf)=\E[k(\Ybf^{0}, \cdot) \mid \Xbf=\xbf]$ and $\mu^{(1)}(\xbf)=\E[k(\Ybf^{1}, \cdot) \mid \Xbf=\xbf]$. Moreover, we define the random element
\begin{align}\label{muWx}
    \mu^{(W)}(\xbf)=\mu^{(0)}(\xbf) (1-W) + \mu^{(1)}(\xbf) W,
\end{align}
and
\begin{align}\label{vardef}
    \sigma^2(\xbf, W)&=\E[\|k(\Ybf, \cdot) - \mu^{(W)}(\xbf) \|_{\H}^2 \mid \Xbf=\xbf, W].
\end{align}

% Note that, using \ref{causalityass2} and \ref{dataass4},
% \begin{align*}
%     \sigma^2(\xbf)&=\Prob(W=0 \mid \Xbf=\xbf) \Var(k(\Ybf^0, \cdot) \mid \Xbf=\xbf) + \Prob(W=1 \mid \Xbf=\xbf) \Var(k(\Ybf^1, \cdot) \mid \Xbf=\xbf) > 0.
% \end{align*}
%In particular, we note that $\sigma^2(\xbf)$ is not the same as $\Var(k(\Ybf,\cdot) \mid \Xbf=\xbf)$. 
Finally, for any $f \in \H\setminus\{0\}$, we also define
\begin{align}\label{fvardef}
    \sigma^2_f(\xbf, W)=\E[|\langle f, k(\Ybf, \cdot) \rangle - \langle f, \mu^{(W)}(\xbf) \rangle |^2 \mid \Xbf=\xbf, W].
\end{align}
%By the same argument, $\sigma^2_f(\xbf) > 0$ by Assumption $\dataass{4}$.

Finally we define
\begin{align}
    \Var(\E[S_1 |\Xbf_1, W_1] \mid  W_1)= \Var(\E[S_{11} |\Xbf_1, W_1=1]) W_1 + \Var(\E[S_{10} |\Xbf_1, W_1=0]) (1-W_1), 
 \end{align}
such that,
\begin{align}
    &\E \Big[ \Var(\E[S_1 |\Xbf_1, W_1] \mid  W_1)  \sigma^2(\xbf, W_1) \Big ]\\
    &=\Var(\E[S_{11} |\Xbf_1, W_1=1]) \sigma^2(\xbf, 1) \Prob(W_1=1) + \Var(\E[S_{10} |\Xbf_1, W_1=0]) \sigma^2(\xbf,0) \Prob(W_1=0),
\end{align}
where $\Var(\E[S_{1j} |\Xbf_1, W_1=j])$ is defined in \eqref{vardefSij}.

In the following, we give our crucial result, Theorem \ref{thm:varianceassumption}, which studies the behavior of $\Var( \E[\langle \Gamma(\Zcal_{s_n}), f \rangle \mid \Zbf_1 ])$, for any $f \in \H$, and $\Var( \E[\Gamma(\Zcal_{s_n})| \Zbf_1 ])$. In particular Theorem \ref{thm:varianceassumption} implies that:
\begin{itemize}
    \item[1.] The ratio of $\Var( \E[\langle \Gamma(\Zcal_{s_n}), f \rangle \mid \Zbf_1 ])$ to $\Var( \E[\Gamma(\Zcal_{s_n})| \Zbf_1 ])$ converges to a real number under \ref{forestass6CausalDRF}. This will be crucial in showing univariate convergence, which in turn is a crucial ingredient in showing total convergence.
    \item[2.] Together with Lemma \ref{lemma4}, \eqref{loverboundvarianceSbehavior2} and \eqref{overalliszeroplusone}, it shows that
    \begin{align}\label{lowervariancebehavior}
        \Var(\tilde{\Gamma}(\Zcal_{s_n})) = s_n \Var(\E[\Gamma(\Zcal_{s_n})| \Zbf_1 ])=\Omega\left( \frac{1}{\log(s_n)^p} \right).
    \end{align}
This will be essential, as it means that the variance of the first-order approximation decays slow enough relative to the variance of the actual tree $\Gamma(\Zcal_{s_n})$. 
\end{itemize}

\begin{remark}\label{rmk:Difficulty}
Proving Theorem \ref{thm:varianceassumption} boils down to controlling the variance of $\E[\Gamma(\Zcal_{s_n})| \Xbf_1, W_1 ]$. In \cite{näf2023confidence} a similar argument had to be made only conditioning on $\Xbf_1$, which allowed to employ successive Lipschitz arguments. However, here the conditioning on $W_1$ adds variance that does not vanish as $n \to \infty$. Thus, compared to DRF for Causal-DRF we need to differentiate the case $W_1=1$ and $W_1=0$. Although this may seem benign, it requires a more careful treatment. This may best be seen in the somewhat awkward looking term $\E \Big[ \Var(\E[S_1 |\Xbf_1, W_1] \mid  W_1) \sigma^2(\xbf, W_1) \Big ] $. Its analogue in the proof of \cite[Theorem 5]{näf2023confidence} was simply $\Var(\E[S_1  \mid \Xbf_1])\Var(k(\Ybf_1, \cdot) \mid \Xbf=\xbf) $.
\end{remark}

\begin{theorem}\label{thm:varianceassumption}
    Assume conditions~\ref{forestass1}--\ref{forestass6CausalDRF},~\ref{forestass1starCausalDRF},~\ref{dataass1}--\ref{dataass6},~\ref{kernelass1}--\ref{kernelass2} and \ref{causalityass1}--\ref{causalityass2} hold. Then, for all $f \in \H\setminus\{0\}$, we have
%     Here the assumption of equal behavior of variance is not yet needed!
    \begin{align}\label{varexpansionfull}
  \Var( \langle \Gamma_{n}(\Zbf_1), f \rangle )&= \E \Big[ \Var(\E[S_1 |\Xbf_1, W_1] \mid  W_1) \sigma_f^2(\xbf, W_1) \Big ] + \o(s_n^{-(1+\epsilon)}), \nonumber \\
\Var( \Gamma_{n}(\Zbf_1))&=  \E \Big[ \Var(\E[S_1 |\Xbf_1, W_1] \mid  W_1)  \sigma^2(\xbf, W_1) \Big ]   + \o(s_n^{-(1+\epsilon)})
\end{align}
for some $\epsilon > 0$. In addition, if \ref{forestass6CausalDRF} holds,
% \begin{align}\label{varianceassumption}
%     \lim_{n \to \infty }\frac{\Var(\langle \Gamma_{n}(\Zbf_1), f \rangle)}{\Var(\Gamma_{n}(\Zbf_1))} = \frac{\Var(\langle k(\Ybf_1, \cdot) ,f  \rangle|\Xbf=\xbf)}{\Var(k(\Ybf_1, \cdot)|\Xbf=\xbf)} > 0,
% \end{align}
\begin{align}\label{varianceassumption}
    &\lim_{n \to \infty }\frac{\Var(\langle \Gamma_{n}(\Zbf_1), f \rangle)}{\Var(\Gamma_{n}(\Zbf_1))}
    %&= \frac{\Var(\E[S_1' |\Xbf_1, W_1=1]) \sigma^2(\xbf, 1) \Prob(W=1) + \Var(\E[S_1' |\Xbf_1, W_1=0]) \sigma^2(\xbf, 0) \Prob(W=0)}{\Var(\E[S_1' |\Xbf_1, W_1=1]) \sigma_f^2(\xbf, 1) \Prob(W=1) + \Var(\E[S_1' |\Xbf_1, W_1=0]) \sigma_f^2(\xbf, 0) \Prob(W=0)} > 0,\\
    = \frac{c \sigma_f^2(\xbf, 1) \Prob(W=1) + \sigma_f^2(\xbf, 0) \Prob(W=0)}{c \sigma^2(\xbf, 1) \Prob(W=1) +  \sigma^2(\xbf, 0) \Prob(W=0)},
\end{align}
with $\Gamma_{n}(\Zbf_1)=\E[\Gamma(\Zcal_{s_n})| \Zbf_1 ] - \E[\Gamma(\Zcal_{s_n})]$.
\end{theorem}

\begin{proof}

Throughout we focus on proving \eqref{varexpansionfull} for $\Var( \E[\Gamma(\Zcal_{s_n}) \mid \Zbf_1 ])$, as the proof for $\Var(\E[\langle \Gamma(\Zcal_{s_n}), f \rangle \mid \Zbf_1 ])$ works analogously.

By the same arguments as in \citep[Theorem 5]{näf2023confidence} and \citep{wager2017estimation}, we can ignore the double-sampling aspect in \ref{forestass1} and condition on a point $\Zbf_1$ used in populating the trees  (i.e. $1 \in \I_2$) and use $s_n$ rather than $|\I_2|$ elements in the tree predictions to simplify notation.

Recall that $\Gamma_{n}(\Zbf_1)=\E[\Gamma(\Zcal_{s_n})| \Zbf_1 ] - \E[\Gamma(\Zcal_{s_n})]$, such that,
\begin{align*}
    \frac{\Var(\langle \Gamma_{n}(\Zbf_1), f \rangle)}{\Var(\Gamma_{n}(\Zbf_1))}= \frac{\Var( \E[\langle \Gamma(\Zcal_{s_n}), f \rangle \mid \Zbf_1 ])}{\Var( \E[\Gamma(\Zcal_{s_n})| \Zbf_1 ])}.
\end{align*}

Let $\1_{w,s_n}=\1\{\text{diam}(\Lcal(\xbf))\leq s_n^{-w}\}$, $ S_i'= S_i \1_{w,s_n}$ and
\begin{align*}
&\Gamma'(\Zcal_{s_n})=\Gamma(\Zcal_{s_n})\1_{w,s_n}=\sum_{i=1}^{s_n} S_i' k(\Ybf_{i}, \cdot)
%&\Gamma'(\Zcal_{s_n})=\Gamma(\Zcal_{s_n})\1_{w,s_n}=\E[ \sum_{i=1}^{s_n} S_i' k(\Ybf_{i}, \cdot) \mid \Zcal_{s_n}].%\1_{w,s_n}\\
\end{align*}
 Crucially, $w$ is chosen such that
 \begin{align}\label{Wagerprobbound}
     \Prob( \text{diam}(\Lcal(\xbf, \Zcal_{s_n})) > s_n^{-w})=\mathcal{O}(s_n^{-w}),
 \end{align}
following from Lemma \ref{lemma2}, similar to~\citep{wager2017estimation}. 

First, using the fact that conditional expectations are projections in $\mathbb{L}^2(\Omega, \mathcal{A}, \H)$, we can write
\begin{align}\label{firstkeyequality}
   \Var( \E[\Gamma'(\Zcal_{s_n})| \Zbf_1 ]) = \Var( \E[\Gamma'(\Zcal_{s_n})| \Xbf_1, W_1 ]) + \Var(  \E[\Gamma'(\Zcal_{s_n})| \Xbf_1, W_1, \Ybf_1 ]-\E[\Gamma'(\Zcal_{s_n})| \Xbf_1, W_1 ])  ,
\end{align}
as in \citep[Claim (69)]{näf2023confidence}. We start by focusing on the second variance in this sum: 
\begin{claim}
    \begin{align}\label{firstnewvarexpansion}
        \Var(  \E[\Gamma'(\Zcal_{s_n})| \Xbf_1, W_1, \Ybf_1 ]-\E[\Gamma'(\Zcal_{s_n})| \Xbf_1, W_1 ])= \E \Big[ \Var(\E[S_1' |\Xbf_1, W_1] \mid  W_1) \sigma^2(\xbf, W_1) \Big ] + \o(s_n^{-(1+\epsilon)}),
    \end{align}
    for some $\epsilon > 0$.
\end{claim}

\begin{claimproof}
We adapt the expansion of \citep{näf2023confidence}: By honesty (i) $\Ybf_i$ is independent of $S_i'$ conditional on $\Xbf_i, W_i$, and more generally, (ii) $\Ybf_i$ is independent of $S_j'$, $j=1,\ldots,n$, conditional on $\Xbf_i, W_i$. Thus, using (i), (ii), and the independence of $\Ybf_1$ from $\Ybf_j$, $j > 1$, we have
\begin{align*}
    \E[ \Gamma'(\Zcal_{s_n}) \mid\Xbf_1, W_1, \Ybf_1 ] &= \E[S_1' k(\Ybf_1, \cdot)\mid\Xbf_1, W_1, \Ybf_1] + \sum_{i=2}^{s_n} \E[S_i' k(\Ybf_i, \cdot)\mid\Xbf_1, W_1, \Ybf_1]\\
    %&=\E[S_1' \mid\Xbf_1, \xi_1] \E[ \xi_1\mid\Xbf_1, \xi_1] + \sum_{i=2}^{n} \E[S_i' k(\Ybf_i, \cdot)\mid\Xbf_1]\\
    &=\E[S_1' \mid \Xbf_1, W_1] k(\Ybf_1, \cdot) + \sum_{i=2}^{s_n} \E[S_i' k(\Ybf_i, \cdot)\mid\Xbf_1, W_1]
\end{align*}
Similarly,
\begin{align*}
  \E[ \Gamma'(\Zcal_{s_n}) \mid \Xbf_1, W_1 ] &=\E[S_1' k(\Ybf_1, \cdot)\mid \Xbf_1, W_1] + \sum_{i=2}^{s_n} \E[S_i' k(\Ybf_i, \cdot)\mid \Xbf_1, W_1]\\
  &=\E[S_1'\mid  \Xbf_1, W_1]  \E[ k(\Ybf_1, \cdot) \mid \Xbf_1, W_1] + \sum_{i=2}^{s_n} \E[S_i' k(\Ybf_i, \cdot)\mid \Xbf_1, W_1]
\end{align*}
and consequently 
\begin{align}\label{21}
   \Var( \E[ T'(\Zcal_{s_n}) \mid \Xbf_1, W_1, \Ybf_1 ]  - \E[ T'(\Zcal_{s_n}) \mid \Xbf_1, W_1 ] )= \Var(\E[S_1' \mid\Xbf_1, W_1] (k(\Ybf_1, \cdot) -   \E[ k(\Ybf_1, \cdot)\mid\Xbf_1,W_1])).
\end{align}
We note that by \ref{causalityass2},
\begin{align*}
    \E[ k(\Ybf_1, \cdot)\mid\Xbf_1,W_1]&=\E[ k(\Ybf_1^{0}, \cdot)\mid\Xbf_1] \1\{W_1=0\} + \E[ k(\Ybf_1^{1}, \cdot)\mid\Xbf_1] \1\{W_1=1\}\\
    &=\mu^{(0)}(\Xbf_1) \1\{W_1=0\} + \mu^{(1)}(\Xbf_1) \1\{W_1=1\}\\
    &=\mu^{(W_1)}(\Xbf_1).
\end{align*}
Note that by the Lipschitz assumption on $\mu^{(0)}(\Xbf_1)$, $\mu^{(1)}(\Xbf_1)$, i.e., \dataass{2}, it holds that
\begin{align}\label{awesomeLipschitz}
    \|\mu^{(W_1)}(\Xbf_1) - \mu^{(W_1)}(\xbf) \|_{\H} &=  \|(\mu^{(0)}(\Xbf_1) - \mu^{(0)}(\xbf) )\1\{W_1=0\} + (\mu^{(1)}(\Xbf_1) - \mu^{(1)}(\xbf) )\1\{W_1=1\} \|_{\H} \nonumber\\
    &\leq \|(\mu^{(0)}(\Xbf_1) - \mu^{(0)}(\xbf) )\|_{\H} \1\{W_1=0\} + \|(\mu^{(1)}(\Xbf_1) - \mu^{(1)}(\xbf) ) \|_{\H}\1\{W_1=1\} \nonumber\\
    &\leq C \| \Xbf_1 - \xbf \|_{\H}.
\end{align}

Then 
\begin{align}\label{22}
    &\Var(\E[S_1' \mid \Xbf_1, W_1] (k(\Ybf_1, \cdot) -   \E[ k(\Ybf_1, \cdot)\mid\Xbf_1, W_1]))  \nonumber \\
    =&\Var(\E[S_1' \mid \Xbf_1, W_1] (k(\Ybf_1, \cdot) -   \mu^{(W_1)}(\Xbf_1) + \mu^{(W_1)}(\xbf)- \mu^{(W_1)}(\xbf) )) \nonumber \\
    %=&\Var(\E[S_1' \mid\Xbf_1, W_1] (k(\Ybf_1, \cdot) - \mu^{(W_1)}(\xbf)  ) - \E[S_1' \mid\Xbf_1]( \E[ k(\Ybf_1, \cdot)\mid\Xbf_1]- \mu(\xbf) )  ) \nonumber \\
    =&\Var(\E[S_1' \mid\Xbf_1, W_1] (k(\Ybf_1, \cdot) - \mu^{(W_1)}(\xbf)  )   ) + \Var(\E[S_1' \mid\Xbf_1, W_1]( \mu^{(W_1)}(\Xbf_1) - \mu^{(W_1)}(\xbf) )) \nonumber \\
    &\quad- \Cov(\E[S_1' \mid\Xbf_1, W_1] (k(\Ybf_1, \cdot) - \mu^{(W_1)}(\xbf)), \E[S_1' \mid\Xbf_1, W_1]( \mu^{(W_1)}(\Xbf_1)- \mu^{(W_1)}(\xbf) ) ). 
\end{align}
For the second variance we have:
\begin{align*}
    &\Var(\E[S_1' \mid\Xbf_1, W_1]( \mu^{(W_1)}(\Xbf_1)- \mu^{(W_1)}(\xbf))) \\
    %& \leq \E[ \|\E[S_1' \mid\Xbf_1, W_1]( \E[ k(\Ybf_1, \cdot)\mid\Xbf_1, W_1]) - \mu^{(W_1)}(\xbf))   \|_{\H}^2 ] \nonumber \\
    &\leq \E[\E[S_1' \mid\Xbf_1, W_1]^2 \| \mu^{(W_1)}(\Xbf_1)- \mu^{(W_1)}(\xbf)     \|_{\H}^2 ] \nonumber \\
    & \leq \E[\E[S_1'^2 \mid\Xbf_1, W_1] C^2\|\Xbf_1 - \xbf \|_{\R^p}^2 ] \nonumber\\
    & \leq \E[S_1'^2  ] C^2 s_n^{-2w},
\end{align*}
where we used assumption \dataass{2} for the second inequality and where the last step followed because $\E[S_1'^2 \mid\Xbf_1, W_1]=0$, for $\|\Xbf_1 - \xbf \|_{\R^p} > s_n^{-w}$ by definition of $S_1'=S_1\1\{ \text{diam}(\Lcal(\xbf))\leq s_n^{-w}\} $. Finally, because $\E[S_1'^2  ]\leq \E[|S_1'|  ]=\O(s_n^{-1})$ by Lemma \ref{lemma4}, we have that
\begin{align}\label{22_1}
    \Var(\E[S_1' \mid\Xbf_1, W_1]( \mu^{(W_1)}(\Xbf_1)- \mu^{(W_1)}(\xbf)))=o(s_n^{-(1+2w)}).
\end{align}
Similarly, for the covariance, it holds that
\begin{align}\label{22_2}
    &\Cov(\E[S_1' \mid\Xbf_1, W_1] (k(\Ybf_1, \cdot) - \mu^{(W_1)}(\xbf)), \E[S_1' \mid\Xbf_1, W_1]( \mu^{(W_1)}(\Xbf_1)- \mu^{(W_1)}(\xbf) ) ) \nonumber\\
    &\leq \Var(\E[S_1' \mid\Xbf_1, W_1] (k(\Ybf_1, \cdot) - \mu^{(W_1)}(\xbf)))^{1/2} \Var(\E[S_1' \mid\Xbf_1, W_1]( \mu^{(W_1)}(\Xbf_1)- \mu^{(W_1)}(\xbf) ))^{1/2}\nonumber\\
    &\leq \E\left[\E[S_1' \mid\Xbf_1, W_1]^2 \|k(\Ybf_1, \cdot) - \mu^{(W_1)}(\xbf)\|_{\H}\right]^{1/2} o(s_n^{-(1/2+w)})\nonumber\\
    %&\leq C \E\left[\E[S_1' \mid\Xbf_1, W_1]^2 \right]^{1/2} o(s_n^{-(1/2+w)})\nonumber\\
    &\leq C \E\left[S_1'  \right]^{1/2} o(s_n^{-(1/2+w)})\nonumber\\
    &= o(s_n^{-(1+w)}),
\end{align}
where the second last step followed because 
\[
\|k(\Ybf_1, \cdot) - \mu^{(W_1)}(\xbf)\|_{\H} \leq \|k(\Ybf_1, \cdot)\|_{\H} + \|\mu^{(W_1)}(\xbf)\|_{\H} \leq C < \infty,
\]
by the boundedness of $k$ \ref{kernelass1}. Combining \eqref{22}, \eqref{22_1} and \eqref{22_2}, we obtain 
\begin{align}\label{23}
    &\Var(\E[S_1' \mid \Xbf_1, W_1] (k(\Ybf_1, \cdot) -   \E[ k(\Ybf_1, \cdot)\mid\Xbf_1, W_1])) \nonumber \\
    &=\Var(\E[S_1' \mid\Xbf_1, W_1] (k(\Ybf_1, \cdot) - \mu^{(W_1)}(\xbf)  )   ) + o(s_n^{-(1+\epsilon)}).
\end{align}
Focussing on the first term, we have:
\begin{align}\label{24}
    &\Var(\E[S_1' \mid\Xbf_1, W_1] (k(\Ybf_1, \cdot) - \mu^{(W_1)}(\xbf)  )   )  \nonumber \\
    =& \E[ \E[S_1' |\Xbf_1, W_1]^2 \|k(\Ybf_1, \cdot) - \mu^{(W_1)}(\xbf) \|_{\H}^2  ] - \E[\E[S_1' |\Xbf_1, W_1] \|k(\Ybf_1, \cdot) - \mu^{(W_1)}(\xbf)  \|_{\H} ]^2 \nonumber \\
    =& \E \Big[ \E[S_1' |\Xbf_1, W_1]^2 \E[\|k(\Ybf_1, \cdot) - \mu^{(W_1)}(\xbf) \|^2 | \Xbf_1, W_1]  \Big] - \E[\E[S_1' |\Xbf_1, W_1] \|k(\Ybf_1, \cdot) - \mu^{(W_1)}(\xbf)  \|_{\H} ]^2.
\end{align}
For the second part, note that because $\|k(\Ybf, \cdot)\|_{\H}$ and $\|\mu^{(W_1)}(\xbf)\|_{\H}$ are bounded by \kernelass{2}, it holds that
\begin{align}\label{24_1}
    \E[\E[S_1' |\Xbf_1, W_1] \|k(\Ybf_1, \cdot) - \mu^{(W_1)}(\xbf)  \|_{\H} ]^2 &\leq \E[\E[S_1' |\Xbf_1, W_1] \|k(\Ybf_1, \cdot)\|_{\H} + \|\mu^{(W_1)}(\xbf)  \|_{\H} ]^2\\
    &\leq C\E[S_1']^2\\
    &=\O(s_n^{2}).
\end{align}
Define $\sigma^2(\xbf, W_1)$ as in \eqref{vardef} and note that by \ref{causalityass2},
\begin{align*}
    \sigma^2(\xbf, W_1) &=\E[\|k(\Ybf_1, \cdot) - \mu^{(W_1)}(\xbf) \|^2 | W_1, \Xbf=\xbf]\\
    &=\Var(k(\Ybf_1^0, \cdot) \mid \Xbf=\xbf) (1-W_1) + \Var(k(\Ybf_1^1, \cdot) \mid \Xbf=\xbf) W_1.
\end{align*}
The first term in~\eqref{24} can be bounded by
\begin{align}\label{24_2}
    &\E \Big[ \E[S_1' |\Xbf_1, W_1]^2 \E[\|k(\Ybf_1, \cdot) - \mu^{(W_1)}(\xbf) \|^2 | \Xbf_1, W_1]  \Big] \nonumber\\
    =&\E \Big[ \E[S_1' |\Xbf_1, W_1]^2 \left(\E[\|k(\Ybf_1, \cdot) - \mu^{(W_1)}(\xbf) \|^2 | \Xbf_1, W_1] - \sigma^2(\xbf, W_1)+ \sigma^2(\xbf, W_1) \right) \Big] \nonumber \\
    =&\E \Big[ \E[S_1' |\Xbf_1, W_1]^2 \left(\E[\|k(\Ybf_1, \cdot) - \mu^{(W_1)}(\xbf) \|^2 | \Xbf_1, W_1] - \sigma^2(\xbf, W_1)+ \sigma^2(\xbf, W_1) \right) \Big] \nonumber\\
    =&\E \Big[ \E[S_1' |\Xbf_1, W_1]^2 \left(\E[\|k(\Ybf_1, \cdot) - \mu^{(W_1)}(\xbf) \|^2 | \Xbf_1, W_1] - \sigma^2(\xbf, W_1) \right) \Big] + \E[\E[S_1' |\Xbf_1, W_1]^2 \sigma^2(\xbf, W_1)  ]
\end{align}
By the same argument, as before, it can be shown that 
\begin{align*}
   &| \E[\|k(\Ybf_1, \cdot) - \mu^{(W_1)}(\xbf) \|^2 | \Xbf_1, W_1] - \sigma^2(\xbf, W_1)| \\
   &\leq |\E[\|k(\Ybf_1^1, \cdot) - \mu^{(1)}(\xbf) \|^2 | \Xbf_1]-\E[\|k(\Ybf_1^1, \cdot) - \mu^{(1)}(\xbf) \|^2 | \Xbf=\xbf]|\\
   &+ |\E[\|k(\Ybf_1^0, \cdot) - \mu^{(0)}(\xbf) \|^2 | \Xbf_1]-\E[\|k(\Ybf_1^0, \cdot) - \mu^{(0)}(\xbf) \|^2 | \Xbf=\xbf]|\\
   &\leq C \| \Xbf_1 - \xbf\|,
\end{align*}
using the fact that $\xbf \mapsto \mu^{(j)}(\xbf)$ and $\xbf \mapsto \E[\|k(\Ybf_1^j, \cdot) - \mu^{(j)}(\xbf) \|^2 | \Xbf=\xbf]$, for $j \in \{0,1\}$, are all Lipschitz (see e.g., the argument in \citep[Equation (78)]{näf2023confidence}). Thus,
\begin{align}\label{24_22}
    &\E \Big[ \E[S_1' |\Xbf_1, W_1]^2 \left(\E[\|k(\Ybf_1, \cdot) - \mu^{(W_1)}(\xbf) \|^2 | \Xbf_1, W_1] - \sigma^2(\xbf, W_1) \right) \Big] \nonumber \\
    &\leq \E \Big[ \E[S_1' |\Xbf_1, W_1]^2 |\E[\|k(\Ybf_1, \cdot) - \mu^{(W_1)}(\xbf) \|^2 | \Xbf_1, W_1] - \sigma^2(\xbf, W_1)| \Big] \nonumber\\
    &=\E[\E[S_1' |\Xbf_1, W_1]^2] C s_n^{w}\nonumber\\
    &=\o(s_n^{-(1+w)}),
\end{align}
using again that $\E[S_1'^2 \mid\Xbf_1, W_1]=0$, for $\|\Xbf_1 - \xbf \|_{\R^p} > s_n^{-w}$ by definition of $S_1'$. Combining~\eqref{24}, \eqref{24_1}, \eqref{24_2} and \eqref{24_22} it holds that
\begin{align}\label{25}
    &\Var(\E[S_1' \mid\Xbf_1, W_1] (k(\Ybf_1, \cdot) - \mu^{(W_1)}(\xbf)  )   )=\E \Big[ \E[S_1' |\Xbf_1, W_1]^2 \sigma^2(\xbf, W_1) \Big ] + \o(s_n^{-(1+w)}).
\end{align}
We can further rewrite this:
\begin{align*}
    &\E \Big[ \E[S_1' |\Xbf_1, W_1]^2 \sigma^2(\xbf, W_1) \Big ]\\
    &=\E\Big[ \E[\E[S_1' |\Xbf_1, W_1]^2 \mid W_1] \sigma^2(\xbf, W_1)  \Big ]\\
    %&=\Var(k(\Ybf_1^0, \cdot) \mid \Xbf=\xbf)\E [ \E[S_1' |\Xbf_1, W_1=0]^2] \Prob(W_1=0) +\\
    %&\Var(k(\Ybf_1^1, \cdot) \mid \Xbf=\xbf)\E [ \E[S_1' |\Xbf_1, W_1=1]^2] \Prob(W_1=1). 
\end{align*}
Moreover, since $\sigma^2(\xbf, W_1)$ is again bounded and $\E[S_{1j}'\mid W_1=j]= \O(s_n^{-1})$ by Lemma \ref{lemma4},
\begin{align}\label{crucialexpectationresult}
\E\Big[\E[\E[S_1' |\Xbf_1, W_1] \mid W_1]^2\sigma^2(\xbf, W_1)\Big]&=\E[\E[\E[S_1' \mid W_1]^2\sigma^2(\xbf, W_1)]\nonumber\\
&=\E[S_1'\mid W_1=1]^2 *P(W_1=1) + \E[S_1'\mid W_1=0]^2 *P(W_1=0) \nonumber\\
%&\leq C \cdot \E[\E[\E[S_1' \mid W_1]^2]\\
&= \O(s_n^{-2}).
\end{align}
Consequently,
\begin{align}\label{25_1}
    \E \Big[ \E[S_1' |\Xbf_1, W_1]^2 \sigma^2(\xbf, W_1) \Big ]=\E \Big[ \Var(\E[S_1' |\Xbf_1, W_1] \mid  W_1) \sigma^2(\xbf, W_1) \Big ] + \O(s_n^{-2}).
\end{align}
Combining \eqref{21}, \eqref{23}, \eqref{25} and \eqref{25_1}, we obtain Claim \eqref{firstnewvarexpansion}. 

% \begin{align*}
%     \E \Big[ \E[S_1' |\Xbf_1, W_1]^2 \sigma^2(\xbf, W_1) \Big ]=& \E \Big[ \E[ \E[S_1' |\Xbf_1, W_1]^2 \mid W_1] \sigma^2(\xbf, W_1) \Big ]
% \end{align*}
% and
% \begin{align*}
%     \E \Big[ \E[S_1' |\Xbf_1, W_1]^2 (\E[\|k(\Ybf_1, \cdot) - \mu^{(W_1)}(\xbf) \|^2 | \Xbf_1, W_1] - \sigma^2(\xbf, W_1))]
% \end{align*}

\end{claimproof}

Combining~\eqref{firstkeyequality} with~\eqref{firstnewvarexpansion}, we have 
\begin{align}\label{varexpansion}
%     \Var( \E[\langle \Gamma'(\Zcal_{s_n}), f \rangle \mid \Zbf_1 ])&= \Var(\E[ \langle \Gamma'(\Zcal_{s_n}), f \rangle \mid \Xbf_1, W_1 ])+ \Var(\E[S_1'  \mid \Xbf_1, W_1])\Var(\langle k(\Ybf_1,\cdot), f \rangle\mid \Xbf=\xbf)  + \mathcal{O}(s_n^{-(1+\epsilon)}), \nonumber \\
% \Var( \E[\Gamma'(\Zcal_{s_n})| \Zbf_1 ])&= \Var(\E[ \Gamma'(\Zcal_{s_n}) |\Xbf_1, W_1 ])+ \Var(\E[S_1'  \mid \Xbf_1, W_1])\Var(k(\Ybf_1,\cdot) \mid \Xbf=\xbf)  + \mathcal{O}(s_n^{-(1+\epsilon)})
    \Var( \E[\langle \Gamma'(\Zcal_{s_n}), f \rangle \mid \Zbf_1 ])&= \Var(\E[ \langle \Gamma'(\Zcal_{s_n}), f \rangle \mid \Xbf_1, W_1 ])+ \E \Big[ \Var(\E[S_1' |\Xbf_1, W_1] \mid  W_1) \sigma^2(\xbf, W_1) \Big ]  + \o(s_n^{-(1+\epsilon)}), \nonumber \\
\Var( \E[\Gamma'(\Zcal_{s_n})| \Zbf_1 ])&= \Var(\E[ \Gamma'(\Zcal_{s_n}) |\Xbf_1, W_1 ])+ \E \Big[ \Var(\E[S_1' |\Xbf_1, W_1] \mid  W_1) \sigma_f^2(W_1,\xbf) \Big ]  + \o(s_n^{-(1+\epsilon)})
\end{align}
for some $\epsilon > 0$. We now show that
\begin{claim}
For some $\epsilon > 0$,
\begin{align}\label{Varboundnew}
    %\Var(\E[ \langle \Gamma'(\Zcal_{s_n}), f \rangle \mid \Xbf_1, W_1 ]) = \mathcal{O}(s_n^{-(1+\epsilon)}) \nonumber\\
\Var( \E[\Gamma'(\Zcal_{s_n})| \Xbf_1, W_1 ])= \o(s_n^{-(1+\epsilon)}).
\end{align}

\end{claim}

\begin{claimproof}
First, due to honesty \ref{forestass1}, we have
\begin{align*}
   \E[ \langle \Gamma'(\Zcal_{s_n}),f \rangle \mid \Xbf_1, W_1 ] &=\E[S_1'\mid \Xbf_1, W_1]  \E[ \langle k(\Ybf_1,\cdot),f \rangle\mid  \Xbf_1, W_1] + \sum_{i=2}^{s_n} \E[S_i' \langle k(\Ybf_i, \cdot),f \rangle\mid \Xbf_1, W_1]
\end{align*}
and
\begin{align*}
  \E[ \Gamma'(\Zcal_{s_n}) \mid \Xbf_1, W_1 ] &=\E[S_1'\mid \Xbf_1, W_1]  \E[ k(\Ybf_1,\cdot)\mid \Xbf_1, W_1] + \sum_{i=2}^{s_n} \E[S_i' k(\Ybf_i, \cdot)\mid \Xbf_1, W_1].
\end{align*}
Subsequently, we consider the variance of the two terms and their covariance individually. That is, we study:
\begin{align}\label{overviewoverview}
    \Var( \E[S_1'\mid \Xbf_1, W_1]  \E[ k(\Ybf_1,\cdot)\mid \Xbf_1, W_1]) + \Var(\sum_{i=2}^{s_n} \E[S_i' k(\Ybf_i, \cdot)\mid \Xbf_1, W_1] ) \nonumber \\
    + 2 \Cov(\E[S_1'\mid \Xbf_1, W_1]  \E[ k(\Ybf_1,\cdot)\mid \Xbf_1, W_1],\sum_{i=2}^{s_n} \E[S_i' k(\Ybf_i, \cdot)\mid \Xbf_1, W_1] ).
\end{align}

First, we study the variance of the first terms. The variances satisfy
\begin{claim}
\begin{align}\label{variance1term_1}
%    \Var(\E[S_1'|\Xbf_1, W_1]  \E[ k(\Ybf_1,\cdot)|\Xbf_1, W_1])&= \Var(\E[S_1'|\Xbf_1, W_1]) \|\E[k(\Ybf_1,\cdot)|\Xbf=\xbf]\|_{\H}^2 + \mathcal{O}(s_n^{-(1+w)})
    \Var(\E[S_1'|\Xbf_1, W_1]  \E[ k(\Ybf_1,\cdot)|\Xbf_1, W_1])&= \E\left[ \Var(\E[S_1'|\Xbf_1, W_1] \mid W_1) \|\mu^{(W_1)}(\xbf)\|_{\H}^2 \right]  + \mathcal{O}(s_n^{-(1+\epsilon)})
\end{align}
%and
%\begin{align}\label{variance1term_2}
%        \Var(\E[S_1'|\Xbf_1, W_1]  \E[ \langle k(\Ybf_1,\cdot),f \rangle|\Xbf_1, W_1]) &= \E\left[ %\Var(\E[S_1'|\Xbf_1, W_1] \mid W_1)  |\langle f, \mu^{(W_1)}(\xbf) \rangle|^2 \right] + \mathcal{O}(s_n^{-%(1+\epsilon)}).
%\end{align}
\end{claim}

\begin{claimproof}

%We only show~\eqref{variance1term_1} because~\eqref{variance1term_2} follows analogously. 
We have
\begin{align}
    &\Var(\E[S_1' |\Xbf_1, W_1]   \E[ k(\Ybf_1,\cdot)|\Xbf_1, W_1]) \\
    &=\Var(\E[S_1' |\Xbf_1, W_1] (  \E[ k(\Ybf_1,\cdot)|\Xbf_1, W_1] - \mu^{(W_1)}(\xbf)+ \mu^{(W_1)}(\xbf) )) \nonumber \\
    &=  \Var(\E[S_1' |\Xbf_1, W_1] \mu^{(W_1)}(\xbf)    ) + \Var(\E[S_1' |\Xbf_1, W_1]( \E[ k(\Ybf_1,\cdot)|\Xbf_1, W_1] - \mu^{(W_1)}(\xbf) )) \nonumber \\
    &\quad+ \Cov\left(\E[S_1' |\Xbf_1, W_1] \mu^{(W_1)}(\xbf), \E[S_1' |\Xbf_1, W_1]( \E[ k(\Ybf_1,\cdot)|\Xbf_1, W_1] - \mu^{(W_1)}(\xbf) ) \right) . 
\end{align}
For the first term it holds that
\begin{align}\label{26}
     \Var(\E[S_1' |\Xbf_1, W_1] \mu^{(W_1)}(\xbf)    ) = \E[\E[S_1' |\Xbf_1, W_1]^2 \|\mu^{(W_1)}(\xbf)\|_{\H}^2    ] + O(s_n^{-2}),
\end{align}
with 
\begin{align*}
    \E[\E[S_1' |\Xbf_1, W_1]^2 \|\mu^{(W_1)}(\xbf)\|_{\H}^2    ] &=\E\Big[ \E[\E[S_1' |\Xbf_1, W_1]^2 \|\mu^{(W_1)}(\xbf)\|_{\H}^2  \mid W_1]\Big] \\
    &=\E\Big[ \E[\E[S_1' |\Xbf_1, W_1]^2  \mid W_1]\|\mu^{(W_1)}(\xbf)\|_{\H}^2 \Big]\\
     &=\E\Big[ \Var(\E[S_1' |\Xbf_1, W_1] \mid W_1)\|\mu^{(W_1)}(\xbf)\|_{\H}^2 \Big] + O(s_n^{-2}),
\end{align*}
since,$\E\Big[ \E[\E[S_1' |\Xbf_1, W_1]  \mid W_1]^2\|\mu^{(W_1)}(\xbf)\|_{\H}^2 \Big]=O(s_n^{-2})$, using the same arguments as in \eqref{crucialexpectationresult}. As was shown in \eqref{awesomeLipschitz},
\begin{align*}
    \|\E[ k(\Ybf_1,\cdot)|\Xbf_1, W_1]- \mu^{(W_1)}(\xbf)    ) \|_{\H} \leq C \| \Xbf-\xbf \|,
\end{align*}
so that,
\begin{align}\label{27}
    \Var(\E[S_1' |\Xbf_1, W_1]( \E[ k(\Ybf_1,\cdot)|\Xbf_1, W_1]- \mu^{(W_1)}(\xbf))) & \leq \E[ \|\E[S_1' |\Xbf_1, W_1]( \E[ k(\Ybf_1,\cdot)|\Xbf_1, W_1]) -\mu^{(W_1)}(\xbf))   \|_{\H}^2 ]\nonumber \\
    &= \E[\E[S_1' |\Xbf_1, W_1]^2 \|( \E[ k(\Ybf_1,\cdot)|\Xbf_1, W_1]- \mu^{(W_1)}(\xbf)    ) \|_{\H}^2 ] \nonumber\\
    & \leq \E[\E[S_1'^2 |\Xbf_1, W_1] C^2\|\Xbf_1 -\xbf \|_{\R^p}^2 ] \nonumber\\
    %& \leq \E[S_1'^2  ] C^2 s_n^{-2w}
    &=\O(s_n^{-(1+2w)}).
\end{align}
where the last step followed because $\E[S_1'^2 |\Xbf_1, W_1]=0$, for $\|\Xbf_1 -\xbf \|_{\R^p} > s_n^{-w}$ by definition of $S_1'$ and the fact that $\E[S_1'^2  ]\leq \E[|S_1'|] \leq \E[|S_1|]=\O(s_n^{-1})$.
Finally, we infer
\begin{align*}
    &\left| \Cov\left(\E[S_1' |\Xbf_1, W_1] \mu^{(W_1)}(\xbf), \E[S_1' |\Xbf_1, W_1]( \E[ k(\Ybf_1,\cdot)|\Xbf_1, W_1] - \mu^{(W_1)}(\xbf) ) \right)\right| \\
    \leq& \sqrt{\Var(\E[S_1' |\Xbf_1, W_1] \mu^{(W_1)}(\xbf)     )} \sqrt{\Var(\E[S_1' |\Xbf_1, W_1]( \E[ k(\Ybf_1,\cdot)|\Xbf_1, W_1]- \mu^{(W_1)}(\xbf) )) }\\
    =& \O(s_n^{-(1+w)}), 
\end{align*}
due to \eqref{27}, and,
\begin{align*}%\label{whatweneed3}
    \Var(\E[S_1' |\Xbf_1, W_1] \mu^{(W_1)}(\xbf)     )&= \E[\E[S_1' |\Xbf_1, W_1]^2 \|\mu^{(W_1)}(\xbf)\|_{\H}^2    ]+ \O(s_n^{-2}) \nonumber \\
    &=\E\Big[ \E[ \E[S_1' |\Xbf_1, W_1]^2  \mid W_1]  \|\mu^{(W_1)}(\xbf)\|_{\H}^2 \Big]+ \O(s_n^{-2})\\
    &=  \E[S_1' |\Xbf_1, W_1=1]^2 \|\mu^{(1)}(\xbf)\|_{\H}^2 \Prob(W_1=1)\\
    &+\E[S_1' |\Xbf_1, W_1=0]^2 \|\mu^{(0)}(\xbf)\|_{\H}^2 \Prob(W_1=0) + \O(s_n^{-2})\\
    &=\O(s_n^{-1}).
\end{align*}
Thus, our Claim~\eqref{variance1term_1} holds.
\end{claimproof}

Before we continue proving the theorem, we note that, due to honesty, we have
\begin{align}\label{honestchange_0}
 \sum_{i=2}^{s_n} \E[S_i' k(\Ybf_i, \cdot)\mid \Xbf_1, W_1] &= \sum_{i=2}^{s_n} \E[ \E[S_i' k(\Ybf_i, \cdot) \mid   \Xbf_i,W_i, \Xbf_1, W_1]  \mid \Xbf_1, W_1] \nonumber \\
 &=\sum_{i=2}^{s_n} \E[ \E[S_i' \mid  \Xbf_i,W_i, \Xbf_1, W_1] \E[ k(\Ybf_i, \cdot) \mid  \Xbf_i,W_i, \Xbf_1, W_1]  \mid \Xbf_1, W_1] \nonumber \\
  &=\sum_{i=2}^{s_n} \E[ \E[S_i' \E[ k(\Ybf_i, \cdot) \mid  \Xbf_i,W_i]   \mid  \Xbf_i, W_i, \Xbf_1, W_1] \mid \Xbf_1, W_1]\nonumber \\
  &=\sum_{i=2}^{s_n} \E[ S_i' \E[  k(\Ybf_i, \cdot) \mid  \Xbf_i, W_i]  \mid \Xbf_1, W_1].
\end{align}

Now, we consider the variance of the sum in \eqref{honestchange_0}:

\begin{claim}
\begin{align}\label{variance2term_1}
    %\Var\left( \sum_{i=2}^{s_n} \E[S_i' k(\Ybf_i, \cdot)\mid \Xbf_1, W_1] \right)&= \Var(\E[S_1'|\Xbf_1, W_1]) \|\E[k(\Ybf_1,\cdot)|\Xbf=\xbf]\|_{\H}^2 + \mathcal{O}(s_n^{-(1+w)})
    \Var\left( \sum_{i=2}^{s_n} \E[S_i' k(\Ybf_i, \cdot)\mid \Xbf_1, W_1] \right)&= \E\left[ \Var(\E[S_1'|\Xbf_1, W_1] \mid W_1) \|\mu^{(W_1)}(\xbf)\|_{\H}^2 \right]  + \o(s_n^{-(1+\epsilon)}) 
\end{align}
% and
% \begin{align}\label{variance2term_2}
%         %\Var\left( \sum_{i=2}^{s_n} \E[S_i' \langle k(\Ybf_i, \cdot),f \rangle\mid \Xbf_1, W_1] \right) &= \Var(\E[S_1'|\Xbf_1, W_1]) \E[\langle k(\Ybf_1,\cdot), f \rangle|\Xbf=\xbf]^2 + \mathcal{O}(s_n^{-(1+w)}).
%         \Var\left( \sum_{i=2}^{s_n} \E[S_i' \langle k(\Ybf_i, \cdot),f \rangle\mid \Xbf_1, W_1] \right) &= \Var(\E[S_1'|\Xbf_1, W_1]) \E[\langle k(\Ybf_1,\cdot), f \rangle|\Xbf=\xbf]^2 + \mathcal{O}(s_n^{-(1+w)}).
% \end{align}

\end{claim}

\begin{claimproof}
% First, we note that, using the definition of $S_i'$, it holds that  
% \begin{align*}
%  \E[\sum_{W_i=1} S_{i1}'\mid \Xbf_1, W_1] &=  \E[\sum_{W_i=1} S_{i1}'\mid \Xbf_1, W_1] \\
%     &=\P\left(\text{diam}(\Lcal(\xbf))\leq s_n^{-w} \mid \Xbf_1, W_1\right).
% \end{align*}
We start by showing the following crucial claim:
\begin{claim}
For some $\epsilon > 0$,
%%%%%% do not delete: Version where everything is separated %%%%%
% To achieve this, we show that:
% \begin{align}\label{whatweneed1}
%     \Var\left(  \E[ \sum_{i \geq 2, W_i=1}  S_{i1}' k(\Ybf_i, \cdot )\mid \Xbf_1, W_1] -   \E[ \sum_{i \geq 2, W_i=1} S_{i1}' \mu^{(1)}(\xbf) \mid \Xbf_1, W_1] \right) =\O(s_n^{-(1+2w)}),
% \end{align}
% \begin{align}\label{whatweneed12}
%     \Var\left(  \E[ \sum_{i \geq 2, W_i=0}  (-S_{i0}') k(\Ybf_i, \cdot )\mid \Xbf_1, W_1] -   \E[ \sum_{i \geq 2, W_i=1} (-S_{i0}')  \mu^{(0)}(\xbf) \mid \Xbf_1, W_1] \right) =\O(s_n^{-(1+2w)}),
% \end{align}
% \begin{align}\label{whatweneed13}
%     &\Cov\Big(\E[ \sum_{i \geq 2, W_i=1}  S_{i1}' k(\Ybf_i, \cdot )\mid \Xbf_1, W_1] - \mu^{(1)}(\xbf)\E[ \sum_{i \geq 2, W_i=1}  S_{i1}' \mid \Xbf_1, W_1], \nonumber \\
%     & \E[ \sum_{i \geq 2, W_i=0}  (-S_{i0}') k(\Ybf_i, \cdot )\mid \Xbf_1, W_1] - \mu^{(0)}\E[ \sum_{i \geq 2, W_i=0}  (-S_{i0}')\mid \Xbf_1, W_1] \Big)=\O(s_n^{-(1+2w)}).
% \end{align}
\begin{align}\label{whatweneed}
\Var\left(  \sum_{i=2}^{s_n} \E[S_i' k(\Ybf_i, \cdot)\mid \Xbf_1, W_1] - \sum_{i=2}^{s_n} \E[S_i' \mu^{(W_i)}(\xbf)\mid \Xbf_1, W_1] \right)=\O(s_n^{-(1+\epsilon)}),%\\
  %  \Var\left(  \E[ \sum_{i \geq 2, W_i=1}  S_{i1}' k(\Ybf_i, \cdot )\mid \Xbf_1, W_1] -   \E[ \sum_{i \geq 2, W_i=1} S_{i1}' \mu^{(1)}(\xbf) \mid \Xbf_1, W_1] \right) =\O(s_n^{-(1+2w)}),
\end{align}
\end{claim}

\begin{claimproof}
We first note that, by \ref{causalityass2},
\begin{align*}
     \E[ \sum_{i=2}^{s_n} S_{i}' \E[  k(\Ybf_i, \cdot) \mid  \Xbf_i, W_i]  \mid \Xbf_1, W_1]&=\E[ \sum_{i=2}^{s_n} S_{i}' \E[  k(\Ybf_i^{(W_i)}, \cdot) \mid  \Xbf_i]  \mid \Xbf_1, W_1]\\
     &=\E[ \sum_{i=2}^{s_n} S_{i}'  \mu^{(W_i)}(\Xbf_i)  \mid \Xbf_1, W_1].
\end{align*}
Combining this with~\eqref{honestchange_0}, we have,
\begin{align}\label{rewriting}
      &\Var\left(  \E[ \sum_{i=2}^{s_n} S_{i}' k(\Ybf_i, \cdot)\mid \Xbf_1, W_1] -  \E[\sum_{i=2}^{s_n} S_{i}'  \mu^{(W_i)}(\xbf)  \mid \Xbf_1, W_1]  \right)\nonumber\\
      &=\Var\left(  \E[ \sum_{i=2}^{s_n} S_{i}' \mu^{(W_i)}(\Xbf_i)  \mid \Xbf_1, W_1] -   \E[\sum_{i=2}^{s_n} S_{i}'  \mu^{(W_i)}(\xbf) \mid \Xbf_1, W_1]  \right) \nonumber \\
      &=\Var\left( \E[ \sum_{i=2}^{s_n} S_{i}' (\mu^{(W_i)}(\Xbf_i)- \mu^{(W_i)}(\xbf) )   \mid \Xbf_1, W_1]  \right)\nonumber \\
      &=\Var\left( \E[ \sum_{i=2}^{s_n} S_{i}' \Delta^{(W_i)}(\Xbf_i)  \mid \Xbf_1, W_1]  \right),
\end{align}
with $\Delta^{(W_i)}(\Xbf_i)=\mu^{(W_i)}(\Xbf_i)- \mu^{(W_i)}(\xbf) $. 
We can now proceed to show \eqref{whatweneed} in the same way we showed \eqref{A_whatweneed} in Section \ref{DRFproofcorrection}:
\begin{claim}
\begin{align}\label{crucialclaimforsumSiCausalDRF}
        \E[\|\E[ \sum_{i=2}^{s_n} S_i'\Delta^{(W_i)}(\Xbf_i)  \mid \Lcal(\xbf), \Xbf_1, W_1] - \E[\sum_{i=2}^{s_n} S_i'\Delta^{(W_i)}(\Xbf_i)  \mid \Lcal(\xbf)]\|_{\H}^2 ]= \o(s_n^{-(1+\epsilon)}).
\end{align}
\end{claim}

\begin{claimproof}
First, using \eqref{forestass1starCausalDRF_2} in place of \eqref{forestass1starDRF_2}, the same proof as in Claim \eqref{newextendedconnection} shows that,
\begin{align}\label{newextendedconnection_CausalDRF}
    \E[ \| \E[\sum_{i=2}^{s_n} S_i'\Delta^{(W_i)}(\Xbf_i)  \mid \Lcal(\xbf), \Xbf_1, W_1] -  \E[ \sum_{i=2}^{s_n}S_i'\Delta^{(W_i)}(\Xbf_i)  \mid \1\{\Xbf_1 \notin \Lcal(\xbf)\}, \Lcal(\xbf)] \|_{\H}^2 ] = \o(s_n^{-(1+\epsilon)}),
\end{align}
for some $\epsilon > 0$. The remainder of the proof then works analogously as for the claim \eqref{crucialclaimforsumSi} in case of a regression tree. 
% \begin{align}\label{A_firstpartofthepuzzle_CausalDRF}
% %\E[ \| \E[\sum_{i=2}^{s_n} S_i'\Delta(\Xbf_i)  \mid \1\{\Xbf_1 \notin \Lcal(\xbf)\}, \Lcal(\xbf)] -  \E[ \sum_{i=2}^{s_n} S_i'\Delta(\Xbf_i)  \mid \{\Xbf_1 \notin \Lcal(\xbf)\}, \Lcal(\xbf)]\|_{\H}^2 ]=\o(s_n^{-(1+\)}).
%      \E[ \| \E[ \sum_{i=2}^{s_n} S_i'\Delta^{(W_i)}(\Xbf_i)  \mid \Lcal(\xbf), \Xbf_1, W_1] -  \E[ \sum_{i=2}^{s_n} S_i'\Delta^{(W_i)}(\Xbf_i)  \mid \{\Xbf_1 \notin \Lcal(\xbf)\}, \Lcal(\xbf)]\|_{\H}^2 ]=\o(s_n^{-(1+\epsilon)}),
% \end{align} 
% for some $\epsilon > 0$.
\end{claimproof}

%%%%% Copied from above, adapt!!! %%%%%%%
We can then split as follows:
\begin{align*}
    &\Var(\E[\sum_{i=2}^{s_n} S_i'\Delta^{(W_i)}(\Xbf_i)  \mid \Xbf_1, W_1] ) \\
    &= \Var\Big( \E[ \E[\sum_{i=2}^{s_n} S _i'\Delta^{(W_i)}(\Xbf_i) \mid \Lcal(\xbf)]  \mid \Xbf_1, W_1] \Big)\\
    &+ \Var\Big( \E[ \E[ \sum_{i=2}^{s_n} S_i'\Delta^{(W_i)}(\Xbf_i) \mid \Lcal(\xbf)]  \mid \Xbf_1, W_1]-\E[ \sum_{i=2}^{s_n} S_i'\Delta^{(W_i)}(\Xbf_i)  \mid \Xbf_1, W_1] \Big)\\
    &- 2 \Cov\Big(\E[ \E[\sum_{i=2}^{s_n} S _i'\Delta^{(W_i)}(\Xbf_i) \mid \Lcal(\xbf)]  \mid \Xbf_1, W_1],\\
    &\E[ \E[ \sum_{i=2}^{s_n} S _i'\Delta^{(W_i)}(\Xbf_i) \mid \Lcal(\xbf)]  \mid \Xbf_1, W_1]-\E[\sum_{i=2}^{s_n} S_i'\Delta^{(W_i)}(\Xbf_i)  \mid \Xbf_1, W_1]  \Big).
\end{align*}
For the first term, by assumption \eqref{forestass1starCausalDRF_1}, and the boundedness of $\E[ \sum_{i=2}^{s_n} S _i'\Delta^{(W_i)}(\Xbf_i) \mid \Lcal(\xbf)]$,
\begin{align*}
    \Var( \E[ \E[ \sum_{i=2}^{s_n} S _i'\Delta^{(W_i)}(\Xbf_i) \mid \Lcal(\xbf)]  \mid \Xbf_1, W_1] ) = \o(s_n^{-(1+\epsilon)}).
\end{align*}
Moreover, using Jensen's inequality and \eqref{crucialclaimforsumSiCausalDRF},
\begin{align*}
    &\Var( \E[ \E[\sum_{i=2}^{s_n} S _i'\Delta^{(W_i)}(\Xbf_i) \mid \Lcal(\xbf)]  \mid \Xbf_1, W_1]-\E[ \sum_{i=2}^{s_n} S_i'\Delta^{(W_i)}(\Xbf_i)  \mid \Xbf_1, W_1] )  \\
    &\leq\E[  \|\E[ \E[\sum_{i=2}^{s_n}  S _i'\Delta^{(W_i)}(\Xbf_i) \mid \Lcal(\xbf)] - \E[\sum_{i=2}^{s_n} S _i'\Delta^{(W_i)}(\Xbf_i) \mid \Xbf_1, W_1, \Lcal(\xbf)]  \mid \Xbf_1, W_1] \|_{\H}^2 ]\\
    &\leq \E[  \E[\| \E[\sum_{i=2}^{s_n} S _i'\Delta^{(W_i)}(\Xbf_i) \mid \Lcal(\xbf)] - \E[ \sum_{i=2}^{s_n} S _i'\Delta^{(W_i)}(\Xbf_i) \mid \Xbf_1, W_1, \Lcal(\xbf)]\|_{\H}^2  \mid \Xbf_1, W_1]  ]\\
    &= \E[\| \E[\sum_{i=2}^{s_n} S _i'\Delta^{(W_i)}(\Xbf_i) \mid \Lcal(\xbf)] - \E[\sum_{i=2}^{s_n} S _i'\Delta^{(W_i)}(\Xbf_i) \mid \Xbf_1, W_1, \Lcal(\xbf)]\|_{\H}^2    ]\\
    &=\o(s_n^{-(1+\epsilon)}).
\end{align*}
Adding the Cauchy-Schwarz inequality to bound the covariances, we obtain
\begin{align*}
    \Var(\E[\sum_{i=2}^{s_n} S_i'\Delta^{(W_i)}(\Xbf_i)  \mid \Xbf_1, W_1] )=\o(s_n^{-(1+\epsilon)}),
\end{align*}
showing \eqref{whatweneed}.

%As outlined above, \eqref{A_whatweneed} and \eqref{A_whatweneed2} imply \eqref{ConnectionbetweenSkandS}.
%%%%% Copied from above, adapt!!! %%%%%%%

\end{claimproof}

Now, we further split
\begin{align*}
     \E[\sum_{i=1}^{s_n} S_{i}'\mid \Xbf_1, W_1]=\E[\sum_{W_i=1} S_{i1}'\mid \Xbf_1, W_1] - \E[\sum_{W_i=0} (-S_{i0}')\mid \Xbf_1, W_1].
\end{align*}
and note that
\begin{align*}
 \E[\sum_{W_i=1} S_{i1}'\mid \Xbf_1, W_1] &=  \E[\sum_{W_i=0} (-S_{i0}')\mid \Xbf_1, W_1] \\
    &=\P\left(\text{diam}(\Lcal(\xbf))\leq s_n^{-w} \mid \Xbf_1, W_1\right),
\end{align*}
so that in particular,
% \begin{align}\label{zerosumgame}
%     \E[\sum_{i=1}^{s_n} S_{i}'\mid \Xbf_1, W_1] &= \P\left(\text{diam}(\Lcal(\xbf))\leq s_n^{-w} \mid \Xbf_1, W_1\right)-\P\left(\text{diam}(\Lcal(\xbf))\leq s_n^{-w} \mid \Xbf_1, W_1\right) \\
%     &=0
% \end{align}
\begin{align*}
         \E[\sum_{W_i=1, i\geq 2} S_{i1}'\mid \Xbf_1, W_1] &= \E[ (\1_{w,s_n}  - S_{11}') W_1 + \1_{w,s_n} (1-W_1)\mid \Xbf_1, W_1]\\
     %\P\left(\text{diam}(\Lcal(\xbf))\leq s_n^{-w} \mid \Xbf_1, W_1\right)-\P\left(\text{diam}(\Lcal(\xbf))\leq s_n^{-w} \mid \Xbf_1, W_1\right) \\
    &=\E[ \1_{w,s_n}  - S_{11}' W_1 \mid \Xbf_1, W_1],
    %&=\P\left(\text{diam}(\Lcal(\xbf))\leq s_n^{-w} \mid \Xbf_1, W_1\right) - \E[S_{11}' \mid \Xbf_1, W_1] W_1\\
\end{align*}
and
\begin{align*}
    \E[\sum_{W_i=0, i\geq 2} (-S_{i0}')\mid \Xbf_1, W_1]=\E[ \1_{w,s_n}  - S_{10}' (1-W_1) \mid \Xbf_1, W_1].
\end{align*}
Thus, if follows that:
\begin{align}\label{zerosumgame}
\E[\sum_{W_i=0, i\geq 2} (-S_{i0}')\mid \Xbf_1, W_1] &=\P\left(\text{diam}(\Lcal(\xbf))\leq s_n^{-w} \mid \Xbf_1, W_1\right) - \E[(-S_{10}') \mid \Xbf_1, W_1] (1-W_1)\nonumber \\
    \E[\sum_{W_i=1, i\geq 2} S_{i1}'\mid \Xbf_1, W_1] &=\P\left(\text{diam}(\Lcal(\xbf))\leq s_n^{-w} \mid \Xbf_1, W_1\right) - \E[S_{11}' \mid \Xbf_1, W_1] W_1 
\end{align}
Following this logic and splitting \eqref{honestchange_0} according to $W_i$, we obtain
\begin{align}\label{overviewforsgreater1}
    &\Var\left(  \sum_{i=2}^{s_n} \E[S_i' \mu^{(W_i)}(\xbf)\mid \Xbf_1, W_1] \right) \nonumber\\
    &=  \Var\left( \E[ \sum_{i \geq 2, W_i=1}  S_{i1}' \mu^{(1)}(\xbf)\mid \Xbf_1, W_1] \right) +\Var\left(  \E[ \sum_{i \geq 2, W_i=0}  (-S_{i0}') \mu^{(0)}(\xbf)\mid \Xbf_1, W_1]  \right)  \\
    &-2 \Cov\left(\E[ \sum_{i \geq 2, W_i=1}  S_{i1}' \mu^{(1)}(\xbf)\mid \Xbf_1, W_1],  \E[ \sum_{i \geq 2, W_i=0}  (-S_{i0}')  \mu^{(0)}(\xbf)\mid \Xbf_1, W_1]\right) \nonumber \\
        &= (I) + (II) - (III)
\end{align}

and
\begin{align}\label{whatweneed2}
\Var\left( \E[ \sum_{i \geq 2, W_i=1} S_{i1}' \mu^{(1)}(\xbf) \mid \Xbf_1, W_1]  \right)= \|\mu^{(1)}(\xbf)\|_{\H}^2\Var\left( \E[S_{11}'\mid \Xbf_1, W_1=1]\right) \Prob(W_1=1) + \O(s_n^{-(1+2w)})
\end{align}
\begin{align}\label{whatweneed22}
\Var\left( \E[ \sum_{i \geq 2, W_i=0} (-S_{i0}') \mu^{(0)}(\xbf) \mid \Xbf_1, W_1]  \right)= \|\mu^{(0)}(\xbf)\|_{\H}^2\Var\left( \E[S_{10}'\mid \Xbf_1,  W_1=0]\right) \Prob(W_1=0) +\O(s_n^{-(1+2w)})
\end{align}
\begin{align}\label{whatweneed23}
    &\Cov\Big(\mu^{(1)}(\xbf)\E[ \sum_{i \geq 2, W_i=1}  S_{i1}' \mid \Xbf_1, W_1], \mu^{(0)}(\xbf)\E[ \sum_{i \geq 2, W_i=0}  (-S_{i0}')\mid \Xbf_1, W_1] \Big)=\O(s_n^{-(1+2w)}).
\end{align}

%We first note that in \eqref{whatweneed12}, \eqref{whatweneed13}, we can also replace $(-S_{i0}')$ by $S_{i0}'$ by the properties of the variance/covariance.

\begin{claim}
   For some $\epsilon > 0$, \eqref{whatweneed2}, \eqref{whatweneed22}, and \eqref{whatweneed23}  hold.
\end{claim}

\begin{claimproof}

We first show \eqref{whatweneed2}. By \eqref{zerosumgame},
\begin{align*}
    &\Var\left( \E[ \sum_{i \geq 2, W_i=1} S_{i1}' \mu^{(1)}(\xbf) \mid \Xbf_1, W_1]  \right)\\ 
    &=\|\mu^{(1)}(\xbf)\|_{\H}^2\Var\left(   \P\left(\text{diam}(\Lcal(\xbf))\leq s_n^{-w} \mid \Xbf_1, W_1\right) - \E[S_{11}' \mid \Xbf_1, W_1] W_1\right)\\
    &=\|\mu^{(1)}(\xbf)\|_{\H}^2 (\Var\left(\P\left(\text{diam}(\Lcal(\xbf))\leq s_n^{-w} \mid \Xbf_1, W_1\right)  \right) + \Var(\E[S_{11}' \mid \Xbf_1, W_1] W_1 ) \\
    &- 2 \Cov(\P\left(\text{diam}(\Lcal(\xbf))\leq s_n^{-w} \mid \Xbf_1, W_1\right),  \E[S_{11}' \mid \Xbf_1, W_1] W_1))
\end{align*}
By assumption \ref{forestass1starCausalDRF}, $\Var\left(\P\left(\text{diam}(\Lcal(\xbf))\leq s_n^{-w} \mid \Xbf_1, W_1\right)  \right)=\o(s_n^{-(1+\epsilon)})$. Moreover, 
\[
\Var(\E[S_{11}' \mid \Xbf_1, W_1] W_1 )\leq \E[(\E[S_{11}' \mid \Xbf_1, W_1] )^2 W_1 ]\leq \E[S_{11}'] =\O(s_n^{-1}),
\]
so that by Cauchy-Schwarz also $\Cov(\P\left(\text{diam}(\Lcal(\xbf))\leq s_n^{-w} \mid \Xbf_1, W_1\right),  \E[S_{11}' \mid \Xbf_1, W_1] W_1)= \o(s_n^{-(1+\epsilon)})$. Thus,
\begin{align*}
    \Var\left( \E[ \sum_{i \geq 2, W_i=1} S_{i1}' \mu^{(1)}(\xbf) \mid \Xbf_1, W_1]  \right)=\|\mu^{(1)}(\xbf)\|_{\H}^2 \Var(\E[S_{11}' \mid \Xbf_1, W_1] W_1 ) + \o(s_n^{-(1+\epsilon)}).
\end{align*}
Since in addition, $\E[\E[S_{11}' \mid \Xbf_1, W_1] W_1]^2= \E[S_{11}'W_1]^2= \E[S_{11}']^2 =\O(s_n^{-2})$, it further holds that
\begin{align}\label{variancerelationS1_1}
    \Var(\E[S_{11}' \mid \Xbf_1, W_1] W_1 ) &= \E[(\E[S_{11}' \mid \Xbf_1, W_1] W_1)^2] + \O(s_n^{-2}) \nonumber \\
    &= \E[ \E[S_{11}' \mid \Xbf_1, W_1]^2 W_1] + \O(s_n^{-2})\nonumber\\
    &= \E[ \E[\E[S_{11}' \mid \Xbf_1, W_1]^2 \mid W_1] W_1] + \O(s_n^{-2})\nonumber\\
    &= \E[ \Var(\E[S_{11}' \mid \Xbf_1, W_1] \mid W_1) W_1] + \O(s_n^{-2}),
\end{align}
where the last step again followed because $\E[ \E[\E[S_{11}' \mid \Xbf_1, W_1] \mid W_1]^2 W_1]=\O(s_n^{-2})$. Finally, 
\begin{align}\label{variancerelationS1_2}
     \E[ \Var(\E[S_{11}' \mid \Xbf_1, W_1] \mid W_1) W_1]= \Var(\E[S_{11}' \mid \Xbf_1, W_1] \mid W_1=1) \Prob(W=1),
\end{align}
showing \eqref{whatweneed2}. Condition \eqref{whatweneed22} can be shown analogously.

Finally, we need to show \eqref{whatweneed23}. By \eqref{zerosumgame},
\begin{align}\label{whatweneed23_1}
    & \Cov\Big(\mu^{(1)}(\xbf)\E[ \sum_{i \geq 2, W_i=1}  S_{i1}' \mid \Xbf_1, W_1], \mu^{(0)}(\xbf)\E[ \sum_{i \geq 2, W_i=0}  (-S_{i0}')\mid \Xbf_1, W_1] \Big) \nonumber \\
    &= \langle \mu^{(1)}(\xbf), \mu^{(0)}(\xbf) \rangle \Cov\Big(\P\left(\text{diam}(\Lcal(\xbf))\leq s_n^{-w} \mid \Xbf_1, W_1\right) - \E[S_{11}' \mid \Xbf_1, W_1] W_1, \nonumber\\
    &\P\left(\text{diam}(\Lcal(\xbf))\leq s_n^{-w} \mid \Xbf_1, W_1\right) + \E[S_{10}' \mid \Xbf_1, W_1] (1-W_1)\Big)\nonumber\\
    &= -\langle \mu^{(1)}(\xbf), \mu^{(0)}(\xbf) \rangle  \Cov(\E[S_{11}' \mid \Xbf_1, W_1] W_1,\E[S_{10}' \mid \Xbf_1, W_1] (1-W_1))+\o(s_n^{-(1+\epsilon)}),
\end{align}
where the last step again followed because $\Var(\P\left(\text{diam}(\Lcal(\xbf))\leq s_n^{-w} \mid \Xbf_1, W_1\right))=\o(s_n^{-(1+\epsilon)})$ by Assumption \ref{forestass1starCausalDRF}, $\Var(\E[S_{11}' \mid \Xbf_1, W_1] W_1)=\O(s_n^{-1})$, $\Var(\E[S_{10}' \mid \Xbf_1, W_1] (1-W_1))=\O(s_n^{-1})$, combined with Cauchy-Schwarz. Now since,
\begin{align*}
    \E[ \E[S_{11}' \mid \Xbf_1, W_1] W_1  ] \E[\E[S_{10}' \mid \Xbf_1, W_1] (1-W_1)] =\O(s_n^{-2}),
\end{align*}
as before, we arrive at
\begin{align}\label{whatweneed23_2}
    &\Cov(\E[S_{11}' \mid \Xbf_1, W_1] W_1, \E[S_{10}' \mid \Xbf_1, W_1] (1-W_1)) \nonumber \\
    &= \E[\E[S_{11}' \mid \Xbf_1, W_1]  \E[S_{10}' \mid \Xbf_1, W_1] W_1 (1-W_1)]+ \O(s_n^{-2}) \nonumber \\
    &=\O(s_n^{-2}).
\end{align}
Combining, \eqref{whatweneed23_1} with \eqref{whatweneed23_2} gives \eqref{whatweneed23}.

\end{claimproof}

Combining \eqref{whatweneed2}-\eqref{whatweneed23} with \eqref{overviewforsgreater1}, we get:
\begin{align}\label{overviewforsgreater2}
    &\Var\left(  \sum_{i=2}^{s_n} \E[S_i' \mu^{(W_i)}(\xbf)\mid \Xbf_1, W_1] \right) \nonumber \\
    &= \|\mu^{(1)}(\xbf)\|_{\H}^2\Var\left( \E[S_{11}'\mid \Xbf_1, W_1=1]\right) \Prob(W_1=1) \nonumber\\
    &+  \|\mu^{(0)}(\xbf)\|_{\H}^2\Var\left( \E[S_{10}'\mid \Xbf_1,  W_1=0]\right) \Prob(W_1=0) + \o(s_n^{-(1+\epsilon)}) \nonumber\\
    &=\E\left[ \Var(\E[S_1'|\Xbf_1, W_1] \mid W_1) \|\mu^{(W_1)}(\xbf)\|_{\H}^2 \right] + \o(s_n^{-(1+\epsilon)}).
\end{align}

In addition, as $\Var\left( \E[S_{11}'\mid \Xbf_1, W_1=1]\right)=\O(s_n^{-1})$, $\Var\left( \E[S_{10}'\mid \Xbf_1,  W_1=0]\right)=\O(s_n^{-1})$ and since $\|\mu^{(j)}(\xbf)\|_{\H}^2$ is bounded for $j \in \{0,1\}$ by \ref{kernelass1}, it follows that 
\begin{align}\label{Varosn1}
   \Var\left(  \sum_{i=2}^{s_n} \E[S_i' \mu^{(W_i)}(\xbf)\mid \Xbf_1, W_1] \right)=\O(s_n^{-1}). 
\end{align}
Combining \eqref{whatweneed} with \eqref{Varosn1} and Cauchy Schwarz, we get
\begin{align}\label{whatweneedfinal}
   &\Var(\sum_{i=2}^{s_n} \E[S_i' k(\Ybf_i, \cdot)\mid \Xbf_1, W_1]) \nonumber \\
   &=\Var\left(  \sum_{i=2}^{s_n} \E[S_i' \mu^{(W_i)}(\xbf)\mid \Xbf_1, W_1] \right)+ \Var\left(  \sum_{i=2}^{s_n} \E[S_i' k(\Ybf_i, \cdot)\mid \Xbf_1, W_1] - \sum_{i=2}^{s_n} \E[S_i' \mu^{(W_i)}(\xbf)\mid \Xbf_1, W_1] \right)\nonumber\\
   &+ \Cov( \sum_{i=2}^{s_n} \E[S_i' \mu^{(W_i)}(\xbf)\mid \Xbf_1, W_1] ,   \sum_{i=2}^{s_n} \E[S_i' k(\Ybf_i, \cdot)\mid \Xbf_1, W_1] - \sum_{i=2}^{s_n} \E[S_i' \mu^{(W_i)}(\xbf)\mid \Xbf_1, W_1])\nonumber\\
   &=\Var\left(  \sum_{i=2}^{s_n} \E[S_i' \mu^{(W_i)}(\xbf)\mid \Xbf_1, W_1] \right) + \o(s_n^{-(1+\epsilon)}).
\end{align}

Combining \eqref{overviewforsgreater2} and \eqref{whatweneedfinal} gives Claim \eqref{variance2term_1}.

\end{claimproof}

Finally, we consider the covariances between $\E[S_1' k(\Ybf_i, \cdot)\mid \Xbf_1, W_1]$ and $ \E[\sum_{i=2}^{s_n} S_i' k(\Ybf_i, \cdot)\mid \Xbf_1, W_1]$:

\begin{claim}
For some $\epsilon > 0$,
\begin{align}\label{covariance2term_1}
      &\Cov\left(\E[S_1' k(\Ybf_i, \cdot)\mid \Xbf_1, W_1],   \E[\sum_{i=2}^{s_n} S_i' k(\Ybf_i, \cdot)\mid \Xbf_1, W_1]\right) \nonumber \\
      &= -\E\left[ \Var(\E[S_1'|\Xbf_1, W_1] \mid W_1) \|\mu^{(W_1)}(\xbf)\|_{\H}^2 \right] +\o(s_n^{-(1+\epsilon)}).
\end{align}
\end{claim}

\begin{claimproof}
    
Similar to above, we first simplify the expression:

\begin{claim}
For some $\epsilon > 0$, we have   
\begin{align}\label{covariance1}
& \Cov\left(\E[S_1' k(\Ybf_1, \cdot)\mid \Xbf_1, W_1],   \E[ \sum_{i=2}^{s_n} S_{i}' k(\Ybf_i, \cdot) \mid \Xbf_1, W_1]\right) \nonumber \\
    &=\Cov\left(\E[ S_1' \mu^{(W_1)}(\xbf) \mid \Xbf_1, W_1],   \E[ \sum_{i=2}^{s_n} S_{i}' \mu^{(W_i )}(\xbf) \mid \Xbf_1, W_1]\right) +\o(s_n^{-(1+\epsilon)}).
\end{align}

% \begin{align}\label{covariance1}
% & \Cov\left(\E[S_1' k(\Ybf_1, \cdot)\mid \Xbf_1, W_1],   \E[ \sum_{i \geq 2, W_i=1} S_{i1} k(\Ybf_i, \cdot) \mid \Xbf_1, W_1]\right) \nonumber \\
%     &=\Cov\left(\E[ S_1 \mu^{(W_1)}(\xbf) \mid \Xbf_1, W_1],   \E[ \sum_{i \geq 2, W_i=1} S_{i1} \mu^{(1)}(\xbf) \mid \Xbf_1, W_1]\right) +\o(s_n^{-(1+\epsilon)}).
% \end{align}

% \begin{align}\label{covariance11}
% & \Cov\left(\E[S_1' k(\Ybf_1, \cdot)\mid \Xbf_1, W_1],   \E[ \sum_{i \geq 2, W_i=0} S_{i0} k(\Ybf_i, \cdot) \mid \Xbf_1, W_1]\right) \nonumber \\
%     &=\Cov\left(\E[ S_1 \mu^{(W_1)}(\xbf) \mid \Xbf_1, W_1],   \E[ \sum_{i \geq 2, W_i=0} S_{i0} \mu^{(0)}(\xbf) \mid \Xbf_1, W_1]\right) +\o(s_n^{-(1+\epsilon)}).
% \end{align}

\end{claim}

\begin{claimproof}
It holds that:
\begin{align*}
    & \Cov\left(\E[ S_1' k(\Ybf_1,\cdot) \mid \Xbf_1, W_1],   \E[ \sum_{i=2}^{s_n} S_{i}' k(\Ybf_i,\cdot) \mid \Xbf_1, W_1]\right)   \\
    &=\Cov\left( \E[ S_1' k(\Ybf_1,\cdot) \mid \Xbf_1, W_1],   \E[ \sum_{i=2}^{s_n} S_{i}' k(\Ybf_i,\cdot) \mid \Xbf_1, W_1]- \E[ \sum_{i=2}^{s_n} S_{i}' \mu^{(W_i)}(\xbf) \mid \Xbf_1, W_1]\right)\\
        %&+\Cov\left(\E[ S_1 \mu^{(W_1)}(\Xbf_1) \mid \Xbf_1, W_1],   \E[ \sum_{i \geq 2, W_i=1} S_{i1} \mu^{(1)}(\xbf) \mid \Xbf_1, W_1]\right)\\
    &+\Cov\left( \E[ S_1' k(\Ybf_1,\cdot) \mid \Xbf_1, W_1]-\E[ S_1' \mu^{(W_1)}(\xbf) \mid \Xbf_1, W_1],  \E[ \sum_{i=2}^{s_n} S_{i}' \mu^{(W_i)}(\xbf) \mid \Xbf_1, W_1]\right)\\
    &+\Cov\left( \E[ S_1' \mu^{(W_1)}(\xbf) \mid \Xbf_1, W_1],  \E[ \sum_{i=2}^{s_n} S_{i}' \mu^{(W_i)}(\xbf) \mid \Xbf_1, W_1]\right)\\
    &=(C1) + (C2) + \Cov\left( \E[ S_1' \mu^{(W_1)}(\xbf) \mid \Xbf_1, W_1],  \E[ \sum_{i=2}^{s_n} S_{i}' \mu^{(W_i)}(\xbf) \mid \Xbf_1, W_1]\right).
\end{align*}
We now use the Cauchy-Schwarz inequality to bound the first two covariances by the product of the square root of the variances of the respective terms. For (C1), as $k$ is bounded \ref{kernelass1},
\begin{align*}
   \Var( \E[ S_1' k(\Ybf_1,\cdot) \mid \Xbf_1, W_1]) &\leq \E[ \| \E[ S_1' k(\Ybf_1,\cdot) \mid \Xbf_1, W_1]\|_{\H}^2] \\
   &\leq \E[ \E[ S_1' \| k(\Ybf_1,\cdot)\|_{\H} \mid \Xbf_1, W_1]^2]\\
   & \leq C \E[S_1']\\
   &=\O(s_n^{-1}),
\end{align*}
and, as shown in \eqref{whatweneed},
\begin{align}
    \Var\left(  \E[ \sum_{i=2}^{s_n} S_{i}' k(\Ybf_i, \cdot )\mid \Xbf_1, W_1] -   \E[\sum_{i=2}^{s_n} S_{i}' \mu^{(W_i)}(\xbf) \mid \Xbf_1, W_1] \right) =\o(s_n^{-(1+\epsilon)}),
\end{align}
so that $|(C1)|\leq \O(s_n^{-1/2})\o(s_n^{-(1+\epsilon)/2})=\o(s_n^{-(1+\epsilon/2)})$. Similarly, by \eqref{27}, it holds that
\begin{align*}
    &\Var(\E[ S_1' k(\Ybf_1,\cdot) \mid \Xbf_1, W_1] - \E[ S_1' \mu^{(W_1)}(\xbf) \mid \Xbf_1, W_1]) \\
    &= \Var(\E[S_1' |\Xbf_1, W_1]( \E[ k(\Ybf_1,\cdot)|\Xbf_1, W_1]- \mu^{(W_1)}(\xbf)))\\
    &= \o(s_n^{-(1+\epsilon)}),
\end{align*}
whereby we recall that $\E[ S_1' k(\Ybf_1,\cdot) \mid \Xbf_1, W_1]=\E[S_1' |\Xbf_1, W_1] \E[ k(\Ybf_1,\cdot)|\Xbf_1, W_1]$ by honesty \ref{forestass1}. Combining this with \eqref{Varosn1}, we obtain $|(C2)|\leq \O(s_n^{-1/2})\o(s_n^{-(1+\epsilon)/2})=\o(s_n^{-(1+\epsilon/2)})$, thus showing \eqref{covariance1}.

\end{claimproof}

We now again split up the term into the part for control and treatment groups:

\begin{align}\label{covarianceoverview}
    &\Cov\left(\E[ S_1' \mu^{(W_1)}(\xbf) \mid \Xbf_1, W_1],   \E[ \sum_{i=2}^{s_n} S_{i}' \mu^{(W_i )}(\xbf) \mid \Xbf_1, W_1]\right) \nonumber \\
    &= \Cov\left(\E[S_1' \mu^{(W_1)}(\xbf)\mid \Xbf_1, W_1],   \E[ \sum_{i \geq 2, W_i=1} S_{i1}' \mu^{(1 )}(\xbf) \mid \Xbf_1, W_1]\right) \nonumber  \\
    &+ \Cov\left(\E[S_1' \mu^{(W_1)}(\xbf)\mid \Xbf_1, W_1],  \E[ \sum_{i \geq 2, W_i=0} S_{i0}'\mu^{(0)}(\xbf) \mid \Xbf_1, W_1]\right)
\end{align}

We then study each term separately:

\begin{claim}
For some $\epsilon > 0$,
\begin{align}\label{covariance2}
    &\Cov\left(\E[ S_1' \mu^{(W_1)}(\xbf) \mid \Xbf_1, W_1],   \E[ \sum_{i \geq 2, W_i=1} S_{i1}' \mu^{(1)}(\xbf) \mid \Xbf_1, W_1]\right) \\
    &=-\|\mu^{(1)}(\xbf)\|_{\H}^2\Var\left( \E[S_{11}'\mid \Xbf_1, W_1=1]\right) \Prob(W_1=1) +\o(s_n^{-(1+\epsilon)})
\end{align}

\begin{align}\label{covariance21}
  &\Cov\left(\E[S_1' \mu^{(W_1)}(\xbf)\mid \Xbf_1, W_1],   \E[ \sum_{i \geq 2, W_i=0} S_{i0}'\mu^{(0)}(\xbf) \mid \Xbf_1, W_1]\right)\\
  &=-\|\mu^{(0)}(\xbf)\|_{\H}^2\Var\left( \E[S_{10}'\mid \Xbf_1,  W_1=0]\right) \Prob(W_1=0) +\o(s_n^{-(1+\epsilon)})
\end{align}
\end{claim}

\begin{claimproof}
Following \eqref{zerosumgame}, it holds that
\begin{align}\label{covariance21_1}
    &\Cov\left(\E[ S_1' \mu^{(W_1)}(\xbf) \mid \Xbf_1, W_1],   \E[ \sum_{i \geq 2, W_i=1} S_{i1}' \mu^{(1)}(\xbf) \mid \Xbf_1, W_1]\right)\nonumber\\
    &=\Cov\left(\E[ S_1' \mu^{(W_1)}(\xbf) \mid \Xbf_1, W_1],   \mu^{(1)}(\xbf)(\P\left(\text{diam}(\Lcal(\xbf))\leq s_n^{-w} \mid \Xbf_1, W_1\right) - \E[S_{11}' \mid \Xbf_1, W_1] W_1) \right)\\
     &=\Cov\Big(\E[ S_1' \mu^{(W_1)}(\xbf) \mid \Xbf_1, W_1],   \mu^{(1)}(\xbf)\P\left(\text{diam}(\Lcal(\xbf))\leq s_n^{-w} \mid \Xbf_1, W_1\right) \Big)\nonumber \\
     &-\Cov\Big(\E[ S_1' \mu^{(W_1)}(\xbf) \mid \Xbf_1, W_1], \E[S_{11}'\mu^{(1)}(\xbf) \mid \Xbf_1, W_1] W_1 \Big).
    %&=-\|\mu^{(1)}(\xbf)\|_{\H}^2\Var\left( \E[S_{11}'\mid \Xbf_1, W_1=1]\right) \Prob(W_1=1) +\o(s_n^{-(1+\epsilon)})
\end{align}
The first element is $\o(s_n^{-(1+\epsilon)})$, as 
\begin{align}\label{covariance21_2}
    &|\Cov\Big(\E[ S_1' \mu^{(W_1)}(\xbf) \mid \Xbf_1, W_1],   \mu^{(1)}(\xbf)\P\left(\text{diam}(\Lcal(\xbf))\leq s_n^{-w} \mid \Xbf_1, W_1\right) \Big)| \nonumber\\
    &\leq \Var(\E[ S_1' \mu^{(W_1)}(\xbf) \mid \Xbf_1, W_1])^{1/2} \|  \mu^{(1)}(\xbf)\|_{\H}\Var(\P\left(\text{diam}(\Lcal(\xbf))\leq s_n^{-w} \mid \Xbf_1, W_1\right))^{1/2} \nonumber\\
    &=\O(s_n^{-1/2})\o(s_n^{-(1+\epsilon)/2}),
\end{align}
where the last step follows because of \ref{forestass1starCausalDRF}, the fact that $\Var(\E[ S_1' \mu^{(W_1)}(\xbf) \mid \Xbf_1, W_1])=\O(s_n^{-1})$ and $\|  \mu^{(1)}(\xbf)\|_{\H}\leq C$ by boundedness of $k$ \ref{kernelass1}. 
In addition, we write
\begin{align}\label{covariance21_3}
    &\Cov\Big(\E[ S_1' \mu^{(W_1)}(\xbf) \mid \Xbf_1, W_1], \E[S_{11}' \mid \Xbf_1, W_1] W_1 \Big)\nonumber\\
    &=\Cov\Big(\E[ S_{11}' \mu^{(1)}(\xbf) \mid \Xbf_1, W_1=1] W_1, \E[S_{11}'\mu^{(1)}(\xbf) \mid \Xbf_1, W_1=1] W_1 \Big)\nonumber\\
    &+\Cov\Big(\E[ S_{10}' \mu^{(0)}(\xbf) \mid \Xbf_1, W_1=0] (1-W_1), \E[S_{11}'\mu^{(1)}(\xbf) \mid \Xbf_1, W_1=1] W_1 \Big).
\end{align}
As
\[
\E[\E[ S_{10}' \mid \Xbf_1, W_1=0] (1-W_1) ] \cdot \E[ \E[S_{11}' \mid \Xbf_1, W_1=1] W_1]=\O(s_n^{-2}),
\]
the second covariance is $ \O(s_n^{-2})$:
\begin{align}\label{covariance21_4}
    &\Cov\Big(\E[ S_{10}' \mu^{(0)}(\xbf) \mid \Xbf_1, W_1=0] (1-W_1), \E[S_{11}'\mu^{(1)}(\xbf) \mid \Xbf_1, W_1=1] W_1 \Big)\nonumber\\
    &=\E[\E[ S_{10}' \mu^{(0)}(\xbf) \mid \Xbf_1, W_1=0] (1-W_1) \E[S_{11}'\mu^{(1)}(\xbf) \mid \Xbf_1, W_1=1] W_1] + \O(s_n^{-2})\nonumber\\
    &= \O(s_n^{-2}).
\end{align}
Similarly, the first covariance can be rewritten as
\begin{align}\label{covariance21_5}
    &\|\mu^{(1)}(\xbf)\|_{\H}^2 \Cov\Big(\E[ S_{11}'  \mid \Xbf_1, W_1=1] W_1, \E[S_{11}'\mid \Xbf_1, W_1=1] W_1 \Big)\nonumber\\
    &=\|\mu^{(1)}(\xbf)\|_{\H}^2 \Var(\E[S_{11}' \mid \Xbf_1, W_1] \mid W_1=1) \Prob(W=1),
\end{align}
using the same steps as in \eqref{variancerelationS1_1} and \eqref{variancerelationS1_2}. Plugging \eqref{covariance21_2}, \eqref{covariance21_3}, \eqref{covariance21_4} and \eqref{covariance21_5} into \eqref{covariance21_1} gives \eqref{covariance2}. Using analogous arguments, we also obtain \eqref{covariance21}.

\end{claimproof}

Plugging \eqref{covariance1}, \eqref{covariance2}, \eqref{covariance21} into \eqref{covarianceoverview}, we obtain \eqref{covariance2term_1}.
\end{claimproof}

Plugging \eqref{variance1term_1}, \eqref{variance2term_1}, \eqref{covariance2term_1}
into \eqref{overviewoverview}, shows
\begin{align*}
    &\E[ \Gamma'(\Zcal_{s_n}) \mid \Xbf_1, W_1 ]\\
    &= 2\E\left[ \Var(\E[S_1'|\Xbf_1, W_1] \mid W_1) \|\mu^{(W_1)}(\xbf)\|_{\H}^2 \right]-2\E\left[ \Var(\E[S_1'|\Xbf_1, W_1] \mid W_1) \|\mu^{(W_1)}(\xbf)\|_{\H}^2 \right]  +\o(s_n^{-(1+\epsilon)})\\
    &=\o(s_n^{-(1+\epsilon)}),
\end{align*}
showing \eqref{Varboundnew}.
\end{claimproof}

Combining \eqref{varexpansion} with \eqref{Varboundnew}, we obtain
\begin{align}\label{VarexpansionS'}
    \Var( \E[\Gamma'(\Zcal_{s_n})| \Zbf_1 ])&=  \E \Big[ \Var(\E[S_1' |\Xbf_1, W_1] \mid  W_1) \sigma_f^2(W_1,\xbf) \Big ]  + \o(s_n^{-(1+\epsilon)}).
\end{align}

We now replace $S_1'$ with $S_1$:

\begin{claim}
For $\epsilon > 0$ and $j \in \{0,1\}$,
\begin{align}\label{truncation_S}
     \left|\Var(\E[S_{1j}'  |\Xbf_1, W_1=j]) -  \Var(\E[S_{1j}  |\Xbf_1, W_1=j])\right| = \o(s_n^{-(1+\epsilon)}).
\end{align}
\end{claim}

\begin{claimproof}
We have 
\begin{align*}
    \Var(\E[S_{1j}|\Xbf_1, W_1=j] - \E[S_{1j}'  |\Xbf_1, W_1=j]) &= \Var( \E[S_{1j} - S_{1j}'  |\Xbf_1, W_1=j])\nonumber\\
    &= \Var( \E[S_{1j} \1\{ \text{diam}(\Lcal(\xbf, \Zcal_{s_n})) > s_n^{-w}\}  |\Xbf_1,W_1=j])\nonumber\\
    &\leq \E[\E[S_{1j} \1\{ \text{diam}(\Lcal(\xbf, \Zcal_{s_n})) > s_n^{-w}\}|\Xbf_1,W_1=j]^2 ]\nonumber\\
    & \leq \E[|S_{1j}|\1\{ \text{diam}(\Lcal(\xbf, \Zcal_{s_n})) > s_n^{-w}\} \mid W_1=j]\nonumber\\
    & =\frac{1}{\Prob(W_1=j)}\E[|S_{1j}|\1\{ \text{diam}(\Lcal(\xbf, \Zcal_{s_n})) > s_n^{-w}\}]\nonumber\\
      & \leq \frac{1}{\min(\varepsilon,1-\varepsilon)}\E[|S_{1j}|\1\{ \text{diam}(\Lcal(\xbf, \Zcal_{s_n})) > s_n^{-w}\}],
%%
    %& \leq \frac{1}{\varepsilon s_n} \sum_{i=1}^{s_n} \E[|S_{1j}| \1\{ \text{diam}(\Lcal(\xbf, \Zcal_{s_n})) > s_n^{-w}\} ], 
    %
    %\nonumber\\
    %&= \frac{1}{s_n} \E[\1\{ \text{diam}(\Lcal(\xbf, \Zcal_{s_n})) > s_n^{-w}\}  \sum_{W_i=j} S_{1j}  ]\nonumber\\
    %&=\frac{1}{s_n}  \Prob( \text{diam}(\Lcal(\xbf, \Zcal_{s_n})) > s_n^{-w})\nonumber\\
    %& = \mathcal{O}(s_n^{-(1+w)})
\end{align*}
where the last step followed due to \ref{causalityass1}. Since moreover, $|S_{1j}| \leq |S_{1}|=|S_{11}|+|S_{10}|$, we further have
\begin{align}\label{eq:varDecompSprime}
    \Var(\E[S_{1j}|\Xbf_1, W_1=j] - \E[S_{1j}'  |\Xbf_1, W_1=j])&\leq \frac{1}{\min(\varepsilon,1-\varepsilon)}\E[|S_{1}|\1\{ \text{diam}(\Lcal(\xbf, \Zcal_{s_n})) > s_n^{-w}\}]\nonumber\\
    &\leq \frac{1}{\min(\varepsilon,1-\varepsilon) s_n} \sum_{i=1}^{s_n}\E[|S_{i}|\1\{ \text{diam}(\Lcal(\xbf, \Zcal_{s_n})) > s_n^{-w}\}]\nonumber\\
    &= \frac{2}{\min(\varepsilon,1-\varepsilon) s_n} \Prob( \text{diam}(\Lcal(\xbf, \Zcal_{s_n})) > s_n^{-w})\nonumber\\
    & = \O(s_n^{-(1+w)}),
\end{align}
due to $\sum_{i=1}^{s_n} S_i=2$ and where the last step followed due to \eqref{Wagerprobbound}. As $\Var(\E[S_1  |\Xbf_1, W_1=j])=\O(s_n^{-1})$ from Lemma \ref{lemma4} and analogously $\Var(\E[S_1'  |\Xbf_1, W_1=j])=\O(s_n^{-1})$, it holds that 
\begin{align*}
    \Var(\E[S_{1j}'  |\Xbf_1, W_1=j]) &=  \Var(\E[S_{1j}  |\Xbf_1, W_1=j]) +  \Var(\E[S_{1j}|\Xbf_1, W_1=j] - \E[S_{1j}'  |\Xbf_1, W_1=j]) \\
    &+ 2 \Cov(\E[S_{1j}  |\Xbf_1, W_1=j], \E[S_{1j}|\Xbf_1, W_1=j] - \E[S_{1j}'  |\Xbf_1, W_1=j])\\
    &=\Var(\E[S_{1j}  |\Xbf_1, W_1=j]) +  \O(s_n^{-(1+w/2)}),
\end{align*}
again making use of Cauchy-Schwarz to bound the covariances.
\end{claimproof}

Combining Claim \eqref{truncation_S} with \eqref{VarexpansionS'}, we have:
\begin{align}\label{VarexpansionS}
    \Var( \E[\Gamma'(\Zcal_{s_n})| \Zbf_1 ])&=  \E \Big[ \Var(\E[S_1 |\Xbf_1, W_1] \mid  W_1) \sigma_f^2(W_1,\xbf) \Big ]  + \o(s_n^{-(1+\epsilon)}).
\end{align}

It remains to show that:

\begin{claim}
For some $\epsilon > 0$,
\begin{align}\label{BounddiffbetweenTandTd}
|\Var( \E[\Gamma(\Zcal_{s_n})| \Zbf_1 ]) -  \Var( \E[\Gamma'(\Zcal_{s_n})| \Zbf_1 ])| =   \o(s_n^{-(1+\epsilon)}).    
\end{align}

\end{claim}

\begin{claimproof}
This can be proven in exactly the same way as Claim (85) in Theorem 24 of \citep{näf2023confidence} with the regression trees $T$, $T'$ exchanged by the causal trees $\Gamma$, $\Gamma'$.
\end{claimproof}

Thus, as $\Gamma_{n}(\Zbf_1)=\E[\Gamma(\Zcal_{s_n})| \Zbf_1 ]$, combining \eqref{VarexpansionS} and \eqref{BounddiffbetweenTandTd}, we arrive at \eqref{varexpansionfull}.

We recall that $ \Var(\E[S_1|\Xbf_1, W_1=j])=\Omega((s_n\log(s_n))^{-1})$, by Lemma \ref{lemma4} and thus the $\o(s_n^{-(1+\epsilon)})$ terms can be ignored. This together with assumption \ref{forestass6CausalDRF} establishes~\eqref{varianceassumption}.
\end{proof}

Having established these basics, we can show asymptotic normality of the estimator in $\H$.

\asymptoticnormality*

\begin{proof}

Recall that 
\begin{align*}
    \tilde{\Gamma}(\Zcal_{s_n})=\sum_{i=1}^{s_n} \E[\Gamma(\mathcal{Z}_{s_n})\mid\Zbf_i] - \E[T(\mathcal{Z}_{s_n})]=\sum_{i=1}^{s_n}\Gamma_n(\Zbf_i),
\end{align*}
the first order approximation of $\Gamma(\Zcal_{s_n})$ as in Lemma \ref{Hdecomposition}.

\begin{claim}\label{eq: U_statsapproximation}
    \begin{align} \label{asymptoticlin}
   \htauk -  \tauk = \frac{s_n}{n} \sum_{i=1}^n  \Gamma_{n}(\Zbf_i) + \o_{p}(\sigma_n),
\end{align}
with 
\begin{align}\label{eq: sigmandef}
   \sigma_n^2=\frac{s_n^2}{n} \Var(\Gamma_{n}(\Zbf_1))=\frac{s_n}{n}\Var(\tilde{\Gamma}(\Zcal_{s_n})),
\end{align} 
% and 
% \begin{align}\label{Tn}
%     \Gamma_{n}(\Zbf_i)=\E[\Gamma(\Zcal_{s_n})\mid\Zbf_i] - \E[\Gamma(\Zcal_{s_n})].
% \end{align}  
\end{claim}

\begin{claimproof}
% First, we notice that due to the fact that \ref{dataass2d} and \ref{dataass3d} hold by Proposition \ref{prop: olddataass}, together with the fact that $\text{\textup{diam}} (\Lcal(\xbf)) \stackrel{p}{\to} 0 $, implies that 
% \begin{align*}
%    \E[k(\Ybf, \cdot) \mid \Xbf \in \Lcal(\xbf)] \stackrel{p}{\to} \E[k(\Ybf, \cdot) \mid \Xbf = \xbf]\\ 
%    \E[\|k(\Ybf, \cdot)\|_{\H}^2 \mid \Xbf \in \Lcal(\xbf)] \stackrel{p}{\to} \E[\|k(\Ybf, \cdot)\|_{\H}^2 \mid \Xbf = \xbf]
% \end{align*}
% Then, using the same arguments as \citep[Theorem 24]{näf2023confidence} with $\Gamma(\Zcal_{s_n})$ replacing $T(\Zcal_{s_n})$, we obtain: 
% \begin{align}\label{varexpansion0}
%   \Var( \E[\Gamma(\Zcal_{s_n})| \Zbf_1 ]) \succsim \Var(\E[S_1 \mid \Zbf_1]) \Var(k(\Ybf, \cdot) \mid \Xbf=\xbf),
% \end{align}
We first note that by boundedness of the kernel, $\Var(\Gamma(\Zcal_{s_n})) \leq C$ for all $n$, for some $C < \infty$. Second, as a consequence of Theorem \ref{thm:varianceassumption}, Lemma \ref{lemma4}, \eqref{loverboundvarianceSbehavior2} and \eqref{overalliszeroplusone}, \eqref{lowervariancebehavior} holds, that is:
    \begin{align*}
        \Var(\tilde{\Gamma}(\Zcal_{s_n})) = s_n \Var(\Gamma_n(\Zbf_1))=\Omega\left(\frac{1}{\log(s_n)^p}\right).
    \end{align*} 
Thus it follows that:
\begin{align}\label{vincrementality}
    \frac{\Var(\tilde{\Gamma}(\Zcal_{s_n}))}{\Var(\Gamma(\Zcal_{s_n}))} = \Omega \left( \frac{1}{\log(s_n)^p}\right).
\end{align}
The same steps as in \citep[Theorem 3]{näf2023confidence} in turn implies that
%%% Difference to expecation
\begin{align*}
    \frac{1}{\sigma_n^2}\E[\norm{ \htauk -  \E[\htauk] - \frac{s_n}{n} \sum_{i=1}^n  \Gamma_{n}(\Zbf_i)}_{\H}^2]\leq&   \frac{s_n^2}{n^2} \frac{\Var(\Gamma(\Zcal_{s_n}))}{\sigma_{n}^2}\\
    =&\frac{s_n}{n}   \frac{\Var(\Gamma(\Zcal_{s_n}))}{\Var(\tilde{\Gamma}(\Zcal_{s_n}))}\\
     =&  \O\left( \frac{s_n}{n}   \frac{ \log(s_n)^p}{ \varepsilon C_{f,p}} \right),
\end{align*}
the latter going to zero as $s_n=n^{\beta}$, with $\beta < 1$ by \forestass{5}.\footnote{We note that \citep{näf2023confidence} uses the abbreviation $k(\Ybf_i, \cdot)=\xi_i$.} By Theorem \ref{bias}, we also have that
\begin{align*}
    \frac{1}{\sigma_n^2}\norm{ \E[\htauk] -  \tauk}_{\H}^2=\frac{1}{\sigma_n^2} \O\left( (\varepsilon s_n)^{- \frac{\log((1-\alpha)^{-1})}{\log(\alpha^{-1})} \frac{\pi}{p}}\right).
\end{align*}
As, with \eqref{lowervariancebehavior},
\begin{align*}
    \sigma_n^2 =\frac{s_n}{n}\Var(\tilde{\Gamma}(\Zcal_{s_n}))= \Omega \left( \frac{s_n}{n} \frac{1}{\log(s_n)^p}\right)
    %\frac{C_1}{ \beta^{d} \log(n)^{d} n^{(\beta-1)} },
    %C_1 \frac{\sqrt{s_n}}{\log(s_n)^{p/2} \sqrt{n}}
\end{align*}
it holds with $C_{\alpha}=\frac{\log((1-\alpha)^{-1})}{\log(\alpha^{-1})} $, that,
\begin{align*}
    \frac{1}{\sigma_n^2}\norm{ \E[\htauk] -  \tauk}_{\H}^2=\frac{1}{\sigma_n^2} \O\left( s_n^{- C_{\alpha} \frac{\pi}{p}}\right)=\O\left( \frac{ n \log(s_n)^p}{s_n^{1+C_{\alpha} \pi/p}} \right)= \O\left( \frac{ n \log(s_n)^p}{n^{\beta(1+C_{\alpha} \pi/p)}} \right)  .
\end{align*}
This goes to zero using the scaling of $\beta$ in \forestass{5}.
\end{claimproof}

Claim \eqref{asymptoticlin} is the first key to show asymptotic normality, as it shows that $1/\sigma_n (\htauk -  \tauk)$ is asymptotically equivalent to a sum of independent elements in the Hilbert space $\H$. Thus we can use CLT results for $\H$, and in particular, we will show (I) marginal convergence and (II) uniform tightness:

\begin{claim}
    (I) For all $f \in \H$, 
    \begin{align}\label{marginalconvergence}
         \frac{s_n}{n \sigma_n} \sum_{i=1}^n  \langle f, \Gamma_{n}(\Zbf_i) \rangle \stackrel{D}{\to} N\left(0,\frac{c \sigma_f^2(\xbf, 1) \Prob(W=1) + \sigma_f^2(\xbf, 0) \Prob(W=0)}{c \sigma^2(\xbf, 1) \Prob(W=1) +  \sigma^2(\xbf, 0) \Prob(W=0)}  \right)
    \end{align}
\end{claim}

\begin{claimproof}
    
For $f \in \H$, consider the univariate marginal $\frac{s_n}{n } \sum_{i=1}^n  \langle \Gamma_{n}(\Zbf_i) ,f \rangle$. Assumptions~\ref{forestass1}--\ref{forestass5}, \ref{causalityass1}--\ref{causalityass2} and Assumptions~\ref{dataass1}--\ref{dataass7} verify all assumptions of Theorem 11 of~\citep{wager2017estimation}. Consequently, there exists a $\sigma_{n}(f) > 0$ converging to zero with $n$ such that
\begin{align}\label{univariatconvergence0}
 \left\langle \frac{1}{\sigma_n(f)} \left( \frac{s_n}{n } \sum_{i=1}^n  \Gamma_{n}(\Zbf_i) \right), f  \right\rangle
 \stackrel{D}{\to} N(0,1).
\end{align}
Unfortunately, the scaling factor 
\[
\sigma_n^2(f)=\frac{s_n^2}{n} \Var( \langle \Gamma_{n}(\Zbf_1), f \rangle)
\]
obtained from~\citep{wager2018estimation} depends on $f$. However, from Theorem \ref{thm:varianceassumption}, it follows that
\begin{align*}
    \frac{\sigma_n(f)}{\sigma_n}= \frac{\sqrt{\Var(\langle \Gamma_{n}(\Zbf_1), f \rangle)}}{\sqrt{\Var(\Gamma_{n}(\Zbf_1))}} \to \frac{c \sigma_f^2(\xbf, 1) \Prob(W=1) + \sigma_f^2(\xbf, 0) \Prob(W=0)}{c \sigma^2(\xbf, 1) \Prob(W=1) +  \sigma^2(\xbf, 0) \Prob(W=0)}, 
\end{align*}
which is finite and strictly larger zero by \ref{dataass4}, \ref{dataass6}. %     \begin{align}\label{varexpansion}
%     \Var( \E[\langle \Gamma(\Zcal_{s_n}), f \rangle \mid \Zbf_1 ])&= \Var(\E[S_1  \mid \Xbf_1])\Var(\langle \xi_1, f \rangle\mid \Xbf=\xbf)  + \mathcal{O}(s_n^{-(1+\epsilon)}), \nonumber \\
% \Var( \E[\Gamma(\Zcal_{s_n})| \Zbf_1 ])&=  \Var(\E[S_1  \mid \Xbf_1])\Var(\xi_1 \mid \Xbf=\xbf)  + \mathcal{O}(s_n^{-(1+\epsilon)}),
% \end{align}
% based on Lemma \ref{lemma4}. The second line in \eqref{varexpansion} is a refinement of \eqref{varexpansion0}. 
Thus
\begin{align*}
    \left\langle \frac{1}{\sigma_n} \left( \frac{s_n}{n } \sum_{i=1}^n  \Gamma_{n}(\Zbf_i) \right), f  \right\rangle &=\frac{\sigma_n(f)}{\sigma_n}  \frac{1}{\sigma_n(f)}\frac{s_n}{n } \sum_{i=1}^n \left\langle  \Gamma_{n}(\Zbf_i) , f  \right\rangle\\
    & \stackrel{D}{\to} N\left(0,\frac{c \sigma_f^2(\xbf, 1) \Prob(W=1) + \sigma_f^2(\xbf, 0) \Prob(W=0)}{c \sigma^2(\xbf, 1) \Prob(W=1) +  \sigma^2(\xbf, 0) \Prob(W=0)}  \right),
\end{align*}
proving the claim. %This variance expression is moreover strictly positive by \dataass{3}, \dataass{4}.
\end{claimproof}

Before continuing, by the definition of $\sigma_n$, we define
\[
\xi_n^0(\xbf) := \sum_{i=1}^n   \frac{s_n}{n \sigma_n} \Gamma_{n}(\Zbf_i) =   \sum_{i=1}^{n} \frac{\Gamma_{n}(\Zbf_i)}{\sqrt{n \Var(\Gamma_n(\Zbf_1))}}.
\]

\begin{claim}
    (II) $\xi_n^0(\xbf)$ is uniformly tight.
\end{claim}

\begin{claimproof}
Because $\H$ is separable due to our assumptions on the kernel, there exists a complete orthogonal basis $\left( e_j\right)_{j \in \N}$ of $\H$; see for instance~\citep{hilbertspacebook}. Let $P_k$ be the projection operator onto the linear span of the first $k$ elements of $ \left( e_j\right)_{j \in \N} $, $S_k=\mbox{span}(e_1, \ldots, e_k)$. Given the tools from above, we can use the exact same arguments as in \citep[Theorem 6]{näf2023confidence} to verify the condition $\limsup_n \E[\| \xi_n^0(\xbf) - P_k(\xi_n^0(\xbf)) \|_{\H}^2] \to 0$, as $k \to \infty$, which is sufficient for tightness~\citep[Lemma 3.2]{HilbertspaceCLTs}.
\end{claimproof}

Univariate convergence together with tightness imply $\xi_n^0(\xbf) \stackrel{D}{\to} N(0, \boldsymbol{\Sigma}_{\xbf})$; see for example~\citep[Lemma 3.1/3.2]{HilbertspaceCLTs} or~\citep[Chapter 7]{hilbertspacebook}. Since by Equation \eqref{asymptoticlin} we have 
\[
\frac{1}{\sigma_n} (\htauk - \tauk) = \xi_n^0(\xbf) + o_{p}(1),
\]
the result follows.

\end{proof}

Finally, having obtained these prior results Theorems \ref{thm: asymptoticnormalityhalfsampling}, \ref{powerprop} and \ref{coverprop} can be proven using the same approach as in \citep{näf2023confidence}, making use of the general results given in \cite{B1,B2,B3}. The proofs are thus omitted.

\end{document}